\pdfoutput=1
\documentclass[letterpaper,11pt]{article}
\usepackage[margin=1in]{geometry}
\usepackage[
            CJKbookmarks=true,
            bookmarksnumbered=true,
            bookmarksopen=true,
            colorlinks=true,
            citecolor=red,
            linkcolor=blue,
            anchorcolor=red,
            urlcolor=blue
            ]{hyperref}

\usepackage[title]{appendix}
            
\usepackage{graphicx}
\usepackage{subfigure}
\usepackage{booktabs} 

\usepackage{algorithm}
\usepackage{algorithmic}
\usepackage{amsmath,amsthm,amsfonts,amssymb}
\usepackage{hyperref}
\usepackage{color}
\usepackage{enumitem}
\usepackage{bm}
\usepackage{multirow}
\usepackage{bbm}
\usepackage{subfiles}
\usepackage{xspace}
\usepackage{caption}
\usepackage{mathtools} 
\usepackage{tikz}

\newtheorem{theorem}{Theorem}
\newtheorem{lemma}{Lemma}

\newtheorem{proposition}{Proposition}

\theoremstyle{definition}
\newtheorem{remark}{Remark}
\newtheorem{definition}{Definition}

\usetikzlibrary{patterns.meta}
\usetikzlibrary{patterns}
\usetikzlibrary{intersections,calc,decorations.pathreplacing,through,arrows}
\usetikzlibrary{ decorations.markings,positioning}
\newcommand{\arrowIn}{
\tikz \draw[-stealth] (-1pt,0) -- (1pt,0);
}
\newcommand{\arrowOut}{
\tikz \draw[-stealth] (1pt,0) -- (-1pt,0);
}
\tikzset{global scale/.style={
    scale=#1,
    every node/.append style={scale=#1}
  }
}

\newcommand{\maQ}{\mathcal{Q}}

\newcommand{\maM}{\mathcal{M}}

\newcommand{\maE}{\mathcal{E}}
\newcommand{\maP}{\mathcal{P}}
\newcommand{\maH}{\mathcal{H}}
\newcommand{\maS}{\mathcal{S}}
\newcommand{\maN}{\mathcal{N}}
\newcommand{\maT}{\mathcal{T}}
\newcommand{\maG}{\mathcal{G}}

\newcommand{\ma}{\mathcal}

\newcommand{\Expect}{\mathbb{E}}

\newcommand{\prob}[1]{ \mathbb{P}\left[ #1 \right] }
\newcommand{\expect}[1]{ \mathbb{E}\left[ #1 \right] }

\newcommand{\hd}{}

\newcommand{\reals}{\mathbb{R}}

\newcommand{\en}{\mathsf{{e}}}
\newcommand{\vn}{\mathsf{{v}}}

\newcommand{\pth}[1]{\left( #1 \right)}
\newcommand{\qth}[1]{\left[ #1 \right]}
\newcommand{\sth}[1]{\left\{ #1 \right\}}

\newcommand{\calC}{{\mathcal{C}}}

\newcommand{\calE}{{\mathcal{E}}}

\newcommand{\calH}{{\mathcal{H}}}

\newcommand{\calM}{{\mathcal{M}}}

\newcommand{\calP}{{\mathcal{P}}}

\newcommand{\calS}{{\mathcal{S}}}
\newcommand{\calT}{{\mathcal{T}}}

\newcommand{\ti}{\tilde}

\newcommand{\Bin}{\mathrm{Bin}}

\newcommand{\overlap}{\mathrm{overlap}}

\newcommand{\multibern}{\mathrm{Bern}}
\newcommand{\gaussianrho}{\mathcal{N}\pth{\begin{pmatrix} 0 \\ 0 \end{pmatrix}, \begin{pmatrix} 1 & \rho \\ \rho & 1 \end{pmatrix}}}

\newcommand{\sfC}{{\mathsf{C}}}

\newcommand{\sfE}{{\mathsf{E}}}

\newcommand{\sfP}{{\mathsf{P}}}

\newcommand{\tr}{\mathrm{Tr}}

\newcommand{\Norm}[1]{\|{#1} \|}

\newcommand{\indc}[1]{{\mathbf{1}_{\left\{{#1}\right\}}}}

\newcommand{\ER}{Erd\H{o}s-R\'enyi }

\title{
Information-Theoretic Thresholds for the Alignments of Partially Correlated Graphs
}
\author{Dong Huang, Xianwen Song, and Pengkun Yang
\thanks{
D.\ Huang and P.\ Yang are with the Department of Statistics and Data Science, Tsinghua
University.
X.\ Song is with the Department of Statistics, University of Chicago.
P. Yang is supported in part by National Key R\&D Program of China 2024YFA1015800, National Natural Science Foundation of China (NSFC) Grant 12101353, and Tsinghua University Dushi Program 2025Z11DSZ001. This paper was presented in part at the 37th Annual Conference on Learning Theory (COLT 2024) \cite{huang2024information}.
}}
\begin{document}
\date{}
\maketitle
\begin{abstract}
  This paper studies the problem of recovering hidden vertex correspondences between two correlated random graphs. We introduce the partially correlated \ER model and the partially correlated Gaussian Wigner model, where a pair of induced subgraphs is correlated. We investigate the information-theoretic thresholds for recovering these latent correlated subgraphs and their hidden vertex correspondences. For the partially correlated \ER model, we establish the optimal rate for partial recovery: above this threshold, a positive fraction of vertices can be correctly matched, while below it, matching any positive fraction is impossible. We also determine the optimal rate for exact recovery. In the partially correlated Gaussian Wigner model, the optimal rates for partial and exact recovery coincide.
  To prove the achievability results, we introduce correlated functional digraphs to partition the edges and bound error probabilities using lower-order cumulant generating functions. Our impossibility results rely on a generalized Fano's inequality and the recovery thresholds for correlated \ER graphs. 
\end{abstract}

\begin{keywords}%
  {Graph alignments, information-theoretic thresholds, \ER model, Gaussian Wigner model, partial recovery, exact recovery}
\end{keywords}
\section{Introduction}

Recently, there has been a surge of interest in the problems of detecting graph correlations and recovering the alignments of two correlated graphs. These questions have emerged across various domains. For instance, in social networks,  determining the similarity between friendship networks across different platforms  has garnered attention \cite{narayanan2008robust, narayanan2009anonymizing}. 
In the realm of computer vision, the identification of whether two graphs represent the same object holds significant importance in pattern recognition and image processing \cite{berg2005shape, cour_balanced_2006}. In computational biology, the representation of biological networks as graphs aids in understanding and quantifying their correlation \cite{singh2008global, vogelstein2015fast}. Furthermore, in natural language processing, the task of determining whether a given sentence can be inferred from the text directly relates to graph matching problems \cite{haghighi2005robust}. 

Numerous graph models exist, with the \ER random graph model being a prominent example, as proposed by \cite{paul1959random} and \cite{gilbert1959random}: 
\begin{definition}[\ER graph]
    The \ER random graph is the graph on $n$ vertices where each edge connects with probability $0<p<1$ independently.
    Let $\ma{G}(n,p)$ denote the  distribution of \ER random graphs with $n$ vertices and edge connecting probability $p$. 
\end{definition}

While there are inherent disparities between the \ER graph and networks derived from real-world scenarios, comprehensively understanding the \ER graphs remains profoundly significant. This understanding serves as a pivotal step in transitioning from solving detection and matching problems on \ER graphs to addressing challenges inherent in practical applications. The graph alignment problem entails identifying latent vertex correspondences between two graphs based on their structures.  
Following \cite{pedarsani2011privacy}, for two random graphs $G_1,G_2$, 
a common graph model is the correlated \ER model. For a weighted graph $G$ with vertex set $V(G)$ and edge set $E(G)$, the weight associated with each edge $uv$ is denoted as $\beta_{uv}(G)$ for any $u, v\in V(G)$. For an unweighted graph $G$, we define $\beta_{uv}(G)=\indc{uv\in E(G)}$.

\begin{definition}[Correlated \ER graphs]
    \label{def:corr-ER-full}
    Let $\pi$ denote a latent bijective mapping from $V(G_1)$ to $V(G_2)$.     
    We say a pair of graphs $(G_1,G_2)$ is correlated \ER graphs if both marginal distributions are $\maG(n,p)$ and each pair of edges $(\beta_{uv}(G_1),\beta_{\pi(u)\pi(v)}(G_2))$ for $u,v\in V(G_1)$ follows the correlated bivariate Bernoulli distribution with correlation coefficient $\rho$.
\end{definition}
We note that the edges in the \ER model are binary 0-1 random variables, where $\beta_{uv}(G_1),\beta_{\pi(u)\pi(v)}(G_2)\in \sth{0,1}$ for any $u,v\in V(G_1)$. The Bernoulli-based \ER model, while useful for modeling binary relationships, can be limited when we aim to represent weighted edges or capture more complex dependencies between nodes. 
Another important model is the correlated Gaussian Wigner model proposed in \cite{ding2021efficient} as a prototypical model for random graphs, where the edges follow Gaussian distributions. 
\begin{definition}[Correlated Gaussian Wigner model]\label{def:gaussian-full}
    Let $\pi$ denote a latent bijective mapping from $V(G_1)$ to $V(G_2)$. We say a pair of graphs $(G_1,G_2)$ follows correlated Gaussian Wigner model if each pair of weighted edges $(\beta_{uv}(G_1),\beta_{\pi(u)\pi(v)}(G_2))$ follows bivariate normal distribution $\mathcal{N}\left(\begin{pmatrix} 0 \\ 0 \end{pmatrix}, \begin{pmatrix} 1 & \rho \\ \rho & 1 \end{pmatrix} \right)$ for any vertices $u,v\in V(G_1)$.
\end{definition}

Given observations on $G_1$ and $G_2$ under the correlated \ER graphs model or correlated Gaussian Wigner model, the goal is to recover the latent vertex mapping $\pi$. 
To quantify the performance of an estimator $\hat \pi$, we consider the following two recovery criteria: 
\begin{itemize}
    \item \emph{Partial recovery}: given a constant $\delta\in (0,1)$, we say $\hat\pi$ succeeds for partial recovery if
    \begin{equation}
        \label{eq:partial-criterion-previous}
        |\{v\in \hd{V(G_1)}:\pi(v) = \hat{\pi}(v)\}|\ge \delta |\hd{V(G_1)}|.
    \end{equation}
    \item \emph{Exact recovery}: we say $\hat\pi$ succeeds for exact recovery if 
    \begin{equation}
        \label{eq:exact-criterion-previous}    
        \pi(v) = \hat{\pi}(v),\quad \forall~v\in \hd{V(G_1)}.
    \end{equation}
\end{itemize}

The information-theoretic thresholds for partial and exact recoveries of $\pi$ under correlated \ER model and correlated Gaussian Wigner model have been extensively studied in the recent literature. 
\begin{itemize}
\item 
\emph{\ER model, partial recovery}. In the sparse regime where $np$ and $\rho$ are constant, partial recovery is impossible \hd{when $n(p^2+\rho p(1-p) )\le 1$~\cite{ganassali2021impossibility, wu2022settling}.} 
It is shown in \cite{hall2023partial} that $np(p\vee \rho)\gtrsim \log\pth{1+\frac{\rho}{p}} \vee 1$ suffices for partial recovery, while $n\gtrsim d_{\mathrm{KL}}(p+\rho-p\rho\| p) \log n$ is necessary, where $d_{\mathrm{KL}}(p \Vert q)$ denotes the Kullback–Leibler (KL) divergence between Bernoulli distributions with mean $p$ and $q$, respectively.
The recent work \cite{wu2022settling} settled the sharp threshold for dense graphs with $\frac{p}{p\vee \rho} = n^{-o(1)}$ and the thresholds within a constant factor for sparse ones with $\frac{p}{p\vee \rho} = n^{-\Omega(1)}$. For the sparse case, a sharp threshold has been proven when $\frac{p}{p\vee \rho} = n^{-\alpha+o(1)}$ for $\alpha\in (0,1]$  in \cite{ding2023matching}.
\item 
\emph{\ER model, exact recovery}.
Based on the properties of the intersection graph under a permutation $\pi$, it is shown in \cite{cullina2016improved,cullina2017exact} that the Maximal Likelihood Estimator (MLE) achieves exact recovery and establishes an information-theoretical lower bound with a gap of $\omega(1)$. The results are sharpened by \cite{wu2022settling} where the sharp threshold for exact recovery are derived.
\item \emph{Gaussian Wigner model.}  It is shown in \cite{ganassali2022sharp} that if $n\rho^2\ge (4+\epsilon)\log n$ for any constant $\epsilon>0$, then the MLE  achieves exact recovery; if instead $n\rho^2\le (4-\epsilon) \log n$, then exact recovery is impossible. The results are strengthened by \cite{wu2022settling} by showing that even partial recovery is impossible under the same condition. 
\end{itemize} 

While numerous studies have extensively investigated recovery procedures in correlated \ER and correlated Gaussian Wigner models, it is however imperative to recognize that, in real-world applications, many nodes in one graph may not have corresponding counterparts in the other graph, leading to incomplete or misaligned structural information. To offer a resolution to this concern, we propose the following models where only a subset of the nodes between the two graphs are correlated.
\begin{definition}[Partially correlated \ER graphs]\label{def:planted correlated ER graph}
    Let $S^*\subseteq V(G_1)$ be a latent subset of vertices and $\pi^*:S^*\mapsto V(G_2)$ be a latent injective mapping.
    We say a pair of graphs $(G_1,G_2)$ is \emph{partially} correlated \ER graphs if both marginal distributions are $\maG(n,p)$ and each pair of weighted edges $(\beta_{uv}(G_1),\beta_{\pi^*(u)\pi^*(v)}(G_2))$ for $u,v\in S^*$ follows the correlated bivariate Bernoulli distribution with correlation coefficient $\rho$.    
\end{definition}

\begin{definition}[Partially correlated Gaussian Wigner model]\label{def:gaussian-partial}
    Let $S^*\subseteq V(G_1)$ be a latent subset of vertices and $\pi^*:S^*\mapsto V(G_2)$ be a latent injective mapping. We say a pair of graphs $(G_1,G_2)$ follows \emph{partially} correlated Gaussian Wigner model if the marginal distribution of each edge in both graphs is standard normal, and for $u,v\in S^*$, the pair $(\beta_{uv}(G_1),\beta_{\pi^*(u)\pi^*(v)}(G_2))$ follows bivariate normal distribution with correlation coefficient $\rho$.
\end{definition}

Let $G[S]$ denote the induced subgraph of $G$ with vertex set $S\subseteq V(G_1)$.
For the partially models in Definitions~\ref{def:planted correlated ER graph} and \ref{def:gaussian-partial}, given $S^*\subseteq V(G_1)$ and the range of $\pi^*$ denoted by $T^* = \pi^*(S^*)\subseteq V(G_2)$, the induced subgraphs $G_1[S^*]$ and $G_2[T^*]$  
follow correlated \ER model and correlated Gaussian Wigner model  on $m$ vertices, respectively. Specifically, the case $S^* = V(G_1)$ reduces to correlated \ER model and correlated Gaussian Wigner model in Definitions~\ref{def:corr-ER-full} and \ref{def:gaussian-full}.


In this paper, we investigate the information-theoretic thresholds for recovering the set of correlated nodes $S^*$ and the mapping $\pi^*$.
For notational simplicity, we also refer to the problem as recovering $\pi^*$ while keeping $S^*$ implicit as the domain of $\pi^*$.
The success criteria in the fully correlated graph models are given by \eqref{eq:partial-criterion-previous} and \eqref{eq:exact-criterion-previous}.
In the partially correlated graph models, owing to the potential inconsistency between the domain of $\pi^*$ and that of the estimator $\hat \pi: \hat{S} \mapsto V(G_2)$, we define their overlap by
\begin{align}\label{def:overlap}
    \overlap(\pi^*,\hat{\pi})\triangleq \frac{|v\in S^*\cap \hat{S}:\pi^*(v) = \hat{\pi}(v)|}{|S^*|}.
\end{align}
With the notion of overlap, the success criteria are equivalent to
\begin{itemize}
    \item \emph{Partial recovery}: $\hat \pi$ succeeds if $\overlap(\pi^*,\hat{\pi})\ge \delta$ for a given constant $\delta\in(0,1)$;
    \item \emph{Exact recovery}: $\hat \pi$ succeeds if $\overlap(\pi^*,\hat{\pi})=1$.
\end{itemize}

\hd{
By analogy with classification problems, we refer to $(u,v)$ as a true pair if $u\in S^*$ and $v=\pi^*(u)$. 
Under this notion,
the numbers of true positives, false positives, false negatives, and true negatives are $m\cdot \overlap(\pi^*,\hat{\pi})$, $m\cdot (1-\overlap(\pi^*,\hat{\pi}))$, $m(1-\overlap(\pi^*,\hat{\pi}))$, and $n^2-m(2-\overlap(\pi^*,\hat{\pi}))$, respectively. Since both $n$ and $m$ are fixed, the analysis on the overlap is sufficient for characterizing all quantities.}
\hd{In this work, we assume that the cardinality of $S^*$ is known. 
When $|S^*|$ is unknown, one potential solution is to employ a penalized estimator to select the model size adaptively. 
We leave this extension for future research.} \hd{See more discussions in Remark~\ref{rmk:unknown-case}.}


\subsection{Main Results}\label{subsec: main results}
In this subsection, we present the main results of the paper. 
We first introduce some notations for the presentation of the main theorems. 
Throughout the paper, we assume that $0<\rho< 1$, $0<p\le \frac{1}{2}$, and the cardinality $|S^*|=m$ is known.  
We denote the bivariate distribution of a pair of Bernoulli random variables with means $p_1,p_2$, and correlation coefficient $\rho$ as $\multibern(p_1,p_2,\rho)$. Specifically, for $\multibern(p,p,\rho)$, we denote the following probability mass function:
\begin{align}\label{eq:def_of_pab}
    p_{11} \triangleq p^2+\rho p(1-p),
\quad 
p_{10} = p_{01} \triangleq (1-\rho)p(1-p),
\quad 
p_{00} \triangleq (1-p)^2+\rho p(1-p).
\end{align}
\begin{itemize}
    \item In the \ER model, a pair of correlated edges 
    $$\pth{\beta_{uv}(G_1),\beta_{\pi^*(u)\pi^*(v)}(G_2)}\sim \multibern(p,p,\rho).$$
Specifically, two correlated edges are both present with probability $p_{11}$, whereas two independent edges are both present with probability $p^2$.
The relative signal strength is quantified by $\gamma \triangleq \frac{p_{11}}{p^2}-1 = \frac{\rho(1-p)}{p}$.
This reparametrization of the correlation coefficient is crucial in determining the fundamental limits of the graph alignment problem.
\item  In the Gaussian Wigner model, a pair of correlated edges consists of two standard Gaussian random variables with correlation coefficient $\rho\in (0,1) $:
$$(\beta_{uv}(G_1),\beta_{\pi^*(u)\pi^*(v)}(G_2))\sim \gaussianrho.$$ 
Here, the relative  signal strength is directly characterized by the correlation coefficient $\rho$.
\end{itemize} 

In the \ER model, we assume $p\ge \frac{1}{n}$, as partial recovery is otherwise impossible \cite{ganassali2021impossibility,wu2022settling}.
Define 
\begin{equation}
    \phi(\gamma) \triangleq (1+\gamma)\log(1+\gamma)-\gamma,
    \quad
    \gamma \triangleq \frac{\rho(1-p)}{p},
\end{equation}
and let $\maS_{n,m}$ denote the set of injective mappings $\pi: S\subseteq V(G_1)\mapsto V(G_2)$ with $|S|=m$. Our goal is to determine the minimum number of correlated nodes $m$ required for successful recovery of $\pi^*$. Next, we introduce our main theorems.



\begin{theorem}[\ER model, partial recovery]\label{thm:partial recovery}
    There exists an estimator $\hat{\pi}$ such that, for any constant $\delta\in(0,1)$ and $\pi^*\in \maS_{n,m}$,  
    \[
\prob{\overlap(\pi^*,\hat{\pi})\ge \delta} = 1-o(1),
    \]
    when $m\ge \frac{c_1(\delta)\log n}{p^2 \phi(\gamma)}$, where $c_1(\delta)$ is a constant depending on $\delta$.
    
    Conversely, for any constant $c,\delta \in(0,1)$, there exists a constant $c_2(c,\delta)$ such that, when $m\le \frac{c_2(c,\delta)\log n}{p^2 \phi(\gamma)}$, for any estimator $\hat \pi$,
    \[
    \prob{\overlap(\pi^*,\hat{\pi})<\delta} \ge  1-c,
    \]
    where $\pi^*$ is uniformly distributed over $\calS_{n,m}$.
\end{theorem}
The possibility result is established in the minimax sense, while the impossibility result is under a Bayesian framework. 
Hence, the threshold applies to both minimax and Bayesian risks. 
Theorem \ref{thm:partial recovery} implies that for partial recovery, the threshold for the number of correlated nodes $m$ is of the order $\frac{\log n}{p^2 \phi(\gamma)}$, beyond which partial recovery is possible and below which partial recovery is impossible. 
The dependency on the ambient graph order $n$ is logarithmic, while the scaling with respect to $p$ and $\rho$ is characterized by $\frac{1}{p^2 \phi(\gamma)}$.


\begin{theorem}[\ER model, exact recovery]\label{thm:exact recovery}
     There exists an estimator $\hat{\pi}$ such that, for any $\pi^*\in \maS_{n,m}$,
     \begin{align*}
        \prob{\overlap(\pi^*,\hat{\pi}) = 1} = 1-o(1),
    \end{align*}
    when $m\ge C_1 \pth{ \frac{\log n}{p^2 \phi(\gamma)}\vee \frac{\log (1/(p^2\gamma)) }{p^2 \gamma} } $, where $C_1$ is a universal constant.
    
    Conversely, for any $c\in(0,1)$, there exists a constant $c_3$ only depending on $c$ such that, when $m\le c_3 \pth{ \frac{\log n}{p^2 \phi(\gamma)}\vee \frac{\log (1/(p^2\gamma))}{p^2 \gamma} }$, for any estimator $\hat{\pi}$
    \begin{align*}
        \prob{\overlap(\pi^*,\hat{\pi}) < 1}\ge 1-c, 
    \end{align*}
    where $\pi^*$ is uniformly distributed over $\calS_{n,m}$.
\end{theorem}

For exact recovery, Theorem~\ref{thm:exact recovery} establishes the threshold for $m$ of the order $\frac{\log n}{p^2 \phi(\gamma)}\vee \frac{\log (1/(p^2\gamma))}{p^2 \gamma}$.
Under the weak signal regime where $\gamma = O(1)$, this rate coincides with that for partial recovery in Theorem~\ref{thm:partial recovery}. 
While the logarithmic scaling in $n$ is common to many other problems on random graphs, 
under the strong signal regime where $\gamma = \omega(1)$, Theorem~\ref{thm:exact recovery} reveals a transition from $\frac{\log n}{p^2 \phi(\gamma)}$ to $\frac{\log (1/(p^2\gamma))}{p^2 \gamma}$ when $\log^2 \frac{1}{p}-\log^2 \frac{1}{\rho} \gtrsim \log n$.
In this regime, the challenge is essentially the oracle recovery of mapping given the sets of correlated nodes  $(S^*,T^*)$. See more discussions in Section~\ref{sec:impossibility}.

\begin{theorem}[Gaussian Wigner model]\label{thm:gauss-main} 
    For any $\rho \in (0,1)$, there exists an estimator $\hat{\pi}$ such that, for any $\pi^*\in \maS_{n,m}$,
     \begin{align*}
\prob{\overlap(\pi^*,\hat{\pi}) =1} =1-o(1),
    \end{align*}
    when $m\ge C_2\pth{\frac{\log n}{\log\pth{1/(1-\rho^2)}}\vee 1}$, where $C_2$ is a universal constant.
    
    Conversely, for any constant $c,\delta\in (0,1)$, there exists a constant $c_4(c,\delta)$ such that, when $m\le c_4(c,\delta)\pth{\frac{\log n}{\log\pth{1/(1-\rho^2)}}\vee 1}$, for any estimator $\hat{\pi}$,\begin{align*}
        \prob{\overlap(\pi^*,\hat{\pi})<\delta}\ge 1-c,
    \end{align*}
    where $\pi^*$ is uniformly distributed over $\maS_{n,m}$.
\end{theorem}

Theorem~\ref{thm:gauss-main} implies that in the Gaussian Wigner model, the thresholds for the number of correlated nodes $m$ are of the order $\frac{\log n}{\log\pth{1/(1-\rho^2)}}\vee 1$, beyond which exact recovery is possible and below which partial recovery is impossible. There is no gap between partial and exact recovery.

\begin{remark}
     The optimal rates for partial and exact recoveries under both correlated \ER model and correlated Gaussian Wigner model are derived in \cite{wu2022settling}. Our results applied with $S^* = V(G_1)$ match the thresholds established in their work up to a constant factor. Furthermore, the lower bound $\frac{\log (1/(p^2\gamma)) }{p^2 \gamma}$ for exact recovery is derived from addressing the alignment problem for the subgraphs with the additional information on the domain and range of $\pi^*$, which applies the result in their work.
\end{remark}

\subsection{\hd{Interpretation of Recovery Thresholds}}\label{subsec: interpretation}
{
Estimating $\pi^*$ under the partially correlated graph model conceptually consists of two subproblems: 1) recovery of the support sets $S^*$ and $T^*$; and 2) recovery of the matching between the two sets.
The existing literature mostly focuses on the latter, yielding necessary conditions under \ER and Gaussian Wigner models, respectively.
In our results, new thresholds $\frac{\log n}{p^2 \phi(\gamma)}$ and $\frac{\log n}{\log(1/(1-\rho^2))}$ reflect the complexity of the support recovery problem.
We provide more discussions on the support recovery problem in Appendix~\ref{apd:support-recovery}.

Those new thresholds have a natural information-theoretical interpretation. 
The entropy of $\pi^*$ (or the support sets $S^*$ and $T^*$) is on the order of $m\log n$.
On the other hand, the mutual information between $\pi^*$ and the observed pair of graphs $(G_1,G_2)$ can be upper bounded as follows:
\begin{itemize}
    \item \ER model: $I(\pi^*;G_1,G_2)\le 25\binom{m}{2}p^2 \phi(\gamma)$,
    \item Gaussian Wigner model: $I(\pi^*;G_1,G_2)\le\frac{1}{2}\binom{m}{2} \log(\frac{1}{1-\rho^2})$,
\end{itemize}
where $p^2 \phi(\gamma)$ and $\log(\frac{1}{1-\rho^2})$ arise from the KL divergence between a pair of correlated edges and independent edges (see Lemma~\ref{lem:mutual-pack}). 
Combining those yields the necessary conditions on the exact recovery by 
\[
I(\pi^*;G_1,G_2)\gtrsim m \log n.
\]
For partial recovery, similar arguments can be applied to the metric entropy of $\pi^*$ using standard volume argument. 
Then the necessary conditions on recovery thresholds follow from Fano's method.

The key quantities $p^2\phi(\gamma)$ and $\log(\frac{1}{1-\rho^2})$ also have an intuitive large deviation interpretation from the upper bounds. 
In our estimator, an edge pair with $e_1\in E(G_1)$ and $e_2\in E(G_2)$ contributes $f(\beta_{e_1}(G_1),\beta_{e_2}(G_2))$ to the total similarity score, where $f$ is a prescribed similarity function.
An incorrect matching involves two independent edges, while a correct matching consists of correlated edges with a higher expected score.
For the \ER model, we use $f(x,y)=xy$ and the rate function analysis of Bernoulli distributions yields $d_{\mathrm{KL}}(p_{11}\| p^2)\ge p^2 \phi(\gamma)$\footnote{This is also known as Bennett inequality.};
for the Gaussian Wigner model, the quantity $\log(\frac{1}{1-\rho^2})$ can be obtained by combining the rate function analyses of Gaussian distributions with $f(x,y)=xy$ and $f(x,y)=-\frac{1}{2}(x-y)^2$.
However, classical large deviation theory is not directly applicable due to the correlations among the incorrectly matched edges.
We address this issue by carefully analyzing the correlation structures in Section~\ref{sec:digraph}.

To illustrate the recovery thresholds, consider $p=n^{-a_1}$ and $\rho = n^{-a_2}$ for constants $a_1,a_2\in (0,1)$ in the \ER model. 
The exact recovery and partial recovery thresholds are given by 
\begin{align*}
    \text{Exact Recovery: }& \begin{cases}
        n^{a_1+a_2}\log n,&a_1>a_2\\
        n^{a_1+a_2}\log n,&a_1=a_2\\
        n^{2a_2}\log n,&a_1<a_2
    \end{cases},\quad 
    \text{Partial Recovery: }\begin{cases}
        n^{a_1+a_2},&a_1>a_2\\
        n^{a_1+a_2}\log n,&a_1=a_2\\
        n^{2a_2}\log n,&a_1<a_2
    \end{cases}.
\end{align*}
There is a transition in the recovery thresholds from $n^{a_1+a_2}\log n$ to $n^{2a_2}\log n$ as the relative signal strength $\gamma=\frac{\rho(1-p)}{p}$ changes from the weak signal regime to the strong signal regime.
Let $m = n^{a_3}$. The critical threshold of the exponent $a_3$ as a function of parameters $a_1$ and $a_2$ is illustrated in the contour plot in Figure~\ref{fig:phase-diagram}. In this phase diagram, we focus on the polynomial dependence on $n$ and ignore logarithmic factors. 
The red line with $m=n$ was established in prior work~\cite{wu2022settling}.
Our contribution is to determine the critical recovery thresholds on $m$ in the green region for all configurations of $p$ and $\rho$.
}
\begin{figure}[htb]
    \centering
\begin{tikzpicture}[scale=1.5, >=stealth]
    \draw[->, thick] (0,0) -- (4.2,0) node[right]{$a_1$};
    \draw[->, thick] (0,0) -- (0,3.2) node[above]{$a_2$};
    \draw (0,0) node[below left]{0};
    
    \draw (1,0) -- (1,-0.05) node[below]{$1/4$};
    \draw (2,0) -- (2,-0.05) node[below]{$1/2$};
    \draw (3,0) -- (3,-0.05) node[below]{$3/4$};
    \draw (4,0) -- (4,-0.05) node[below]{$1$};

    \draw (0,1) -- (-0.05,1) node[left]{$1/4$};
    \draw (0,2) -- (-0.05,2) node[left]{$1/2$};
    \draw (0,3) -- (-0.05,3) node[left]{$3/4$};

    \fill[green!20] (0,0) -- (0,2) -- (2,2) -- (4,0);
    \fill[gray!20] (0,2) -- (2,2) -- (4,0) -- (4,3) -- (0,3);
    \draw[color = gray!50] (0,0) -- (3,3);
    \fill[pattern=vertical lines, pattern color=gray!50]
       (0,0) -- (3,3) -- (0,3) -- cycle;
    \fill[pattern=horizontal lines, pattern color=gray!50]
       (0,0) -- (4,0) -- (4,3) -- (3,3) -- cycle;
    
    \fill[pattern={Lines[angle=90, line width=0.3pt, distance=1.5pt]},pattern color=gray] (4.8,0.5) rectangle (5,0.7);
  \draw (4.8,0.5) rectangle (5,0.7);
  \node[right=0.2cm] at (5,0.57) {Weak signal ($\rho\ll p$)};

  \fill[
    pattern={Lines[angle=0, line width=0.3pt, distance=1.5pt]},
    pattern color=gray
  ] (4.8,0.1) rectangle (5,0.3);
  \draw (4.8,0.1) rectangle (5,0.3);
  \node[right=0.2cm] at (5,0.17) {Strong signal ($\rho\gg p$)};

    \draw[red, thick] (0,2) -- (2,2) -- (4,0);
    \draw[blue] (0,1.5) -- (1.5,1.5) -- (3,0);
    \draw[blue] (0,1) -- (1,1) -- (2,0);
    \draw[blue] (0,0.5) -- (0.5,0.5) -- (1,0);
    
    \node at (2.1,2.15) {$a_3=1$};
    \node at (1.6,1.65) {$a_3=3/4$};
    \node at (1.1,1.15) {$a_3 = 1/2$};
    \node at (0.6,0.65) {$a_3 = 1/4$};
    
    \draw[fill=green!20] (4.8,2.7) rectangle (5,2.9); \node[right=0.2cm] at (5,2.77) {Possible regime};
    \draw[fill=gray!20] (4.8,2.3) rectangle (5,2.5);
    \node[right=0.2cm] at (5,2.37) {Impossible regime};

\end{tikzpicture}
    \caption{Phase diagram for recovery thresholds with $p=n^{-a_1}$, $\rho = n^{-a_2}$, and $m = n^{a_3}$.}
    \label{fig:phase-diagram}
\end{figure}


\subsection{Related Work}
\label{sec:related}

\paragraph{Graph sampling.}

When analyzing the properties of graphs, several challenges arise, such as limited data, high acquisition costs~\cite{stumpf2005subnets}, and incomplete knowledge of hidden structures~\cite{lancichinetti2009community, yang2013community, fortunato2016community}. Due to these challenges, graph sampling is  a powerful approach for exploring graph structure.
When data is sampled from two large networks, it is often unrealistic to assume full correlation among nodes in the sampled subgraphs. This naturally leads to partially correlated graphs. While the exact number of correlated nodes may be unknown, we usually have some estimate of the size. As a simplification, our model considers the case where the size is fixed. For the case where the size of correlated nodes is unknown, we leave it as our future work. See more details in Remark~\ref{rmk:unknown-case}.



\paragraph{\hd{Subgraphs isomorphism problem.}}
{
In our partially correlated \ER model, there are two latent subgraphs of order $m$ with correlated edges. 
Specifically, when $\rho = 1$, the two subgraphs are isomorphic. 
Intuitively, a necessary condition for a successful alignment is that $m$ exceeds the order the largest induced isomorphic subgraph between the independent parts.
It is recently shown in~\cite{CD2023,surya2025isomorphisms} that the size of maximal isomorphic subgraphs between two independent \ER graphs $\maG(n,p)$ is $\Theta(\log n)$ when $p$ is a constant. 
In this regime, our threshold $m\ge C\log n$ for some constant $C$ aligns with such results. 
When $p = p_n\to 0$, the order of the largest isomorphic subgraphs between two independent \ER graphs remains open~\cite{surya2025isomorphisms}.
}

\paragraph{Efficient algorithms and computational hardness.}
Numerous algorithms have been developed for the recovery problem. 
For example,  percolation graph matching algorithm \cite{yartseva2013performance},  subgraph matching algorithm \cite{barak2019nearly}, and local neighborhoods based algorithm \cite{mossel2020seeded}. 
However, these algorithms may be computationally inefficient. 
There are several polynomial-time algorithms for recovery, catering to different regimes correlation coefficients $\rho$. These include works by \cite{babai1980random,bollobas1982distinguishing,dai2019analysis,ganassali2020tree,ding2021efficient,mao2023exact,mao2023random,ding2023polynomial,muratori2024faster}. For instance, a polynomial-time algorithm for recovery is proposed in \cite{mao2023random} by counting chandeliers when the correlation coefficient  $\rho>\sqrt{\alpha}$, where $\alpha\approx 0.338$ is the Otter's constant introduced in \cite{otter1948number}. Additionally,
 there exists an efficient iterative polynomial-time algorithm for sparse \ER graphs when the correlation coefficient is a constant \cite{ding2023polynomial}.

It is postulated in \cite{hopkins2017efficient, hopkins2018statistical,kunisky2019notes} that the framework of low-degree polynomial algorithms effectively demonstrates computation hardness of detecting and recovering latent structures, and it bears similarities to sum-of-square methods \cite{hopkins2017power, hopkins2018statistical}. Based on the conjecture on the hardness of low-degree polynomial algorithms, it is shown in \cite{mao2024testing} that there is no polynomial-time test or matching algorithm when the correlation coefficient satisfies $\rho^2\le \frac{1}{\text{polylog}(n)}$. Furthermore, the recent work \cite{ding2023low} showed computation hardness for detection and exact recovery when $p = n^{-1+o(1)}$ and the correlation coefficient $\rho<\sqrt{\alpha}$, where $\alpha\approx 0.338$ is the Otter's constant, suggesting that several polynomial algorithms may be essentially optimal. In this paper, we do not address polynomial algorithms or the computational hardness on partial models and leave these topics for future work.


\paragraph{Correlation detection} Besides the recent literature on the graph alignment problem, the correlation detection is another related topic. 
Given a pair of graphs, their correlation detection is formulated as a hypothesis testing problem, wherein the null hypothesis assumes independent random graphs, while the alternative assumes edge correlation under a latent permutation. 
A hypothesis testing model for correlated \ER graphs is proposed in \cite{barak2019nearly}.  The sharp threshold for dense \ER graphs and the threshold within a constant factor for sparse \ER graphs on this hypothesis testing model are established in \cite{wu2023testing}. It is shown in \cite{ding2023detection} that the sharp threshold for sparse \ER graphs can be derived by analyzing the densest subgraph. Additionally, a polynomial time algorithm for detection is also possible by counting trees when the correlation coefficient exceeds a constant value \cite{mao2024testing}.
It is natural to ask whether the correlation can be detected when only a subsample from the graphs is collected. 
The probabilistic model is similar to the one present in the current paper, and we leave the exploration as our future work.

\paragraph{Other graph models.}
Many properties of the correlated \ER  model have been extensively investigated. However, the strong symmetry and tree-like structure inherent in this model distinguish it significantly from graph models encountered in practical applications. Therefore, it is crucial to explore more general graph models.  One such model is inhomogeneous random graph model, where the edge connecting probability varies among edges in the graph \cite{racz2023matching, song2023independence, ding2023efficiently}. Besides, geometric random graph model \cite{wang2022random, sentenac2023online,  bangachev2024detection, gong2024umeyama}, planted cycle model \cite{mao2023detection, mao2024informationtheoretic}, planted subhypergraph model \cite{dhawan2023detection} and corrupt model \cite{ameen2024robust} have also been subjects of recent studies.

\subsection{Notations}

For any $n\in \mathbb{N}$, let $[n]\triangleq \{1,2,\cdots,n\}$. For any $a,b\in \mathbb{R}$, let $a\wedge b = \min\{a,b\}$ and $a\vee b = \max(a,b)$.
We use standard asymptotic notation: for two positive sequences $\{a_n\}$ and $\{b_n\}$, we write $a_n = O(b_n)$ or $a_n\lesssim b_n$, if $a_n\le C b_n$ for some absolute constant $C$ and all $n$; $a_n = \Omega(b_n)$ or $a_n\gtrsim b_n$, if $b_n = O(a_n)$; $a_n = \Theta(b_n)$ or $a_n\asymp b_n$, if $a_n = O(b_n)$ and $a_n = \Omega(b_n)$; $a_n = o(b_n)$ or $b_n =\omega(a_n)$, if $a_n/b_n\to 0$ as $n\to \infty$.

For a given weighted graph $G$, let $V(G)$ denote its vertex set and $E(G)$ its edge set. For a set $V$, let $\binom{V}{2}\triangleq \{ \{x,y\}: x,y\in V, x\neq y\}$ denote the collection of all subsets of $V$ with cardinality two. 
We write $uv$ to represent an edge $\{u,v\}$, and $\beta_{e}(G)$ for the weight of the edge $e$. 
 For unweighted graphs $G$, we define $\beta_{uv}(G)=\indc{uv\in E(G)}$.
Let $\vn(G)=|V(G)|$ denote the number of vertices in $G$, and $\en(G)=\sum_{e\in E(G)} \beta_{e}(G)$ the total weight of its edges. 
The induced subgraph of $G$ over a vertex set $V$ is denoted by $G[V]$.
Given an injective mapping of vertices $\pi:S \subseteq V(G_1) \mapsto  V(G_2)$, the induced injective mapping of edges is defined as $\pi^{\sfE}:\binom{S}{2} \mapsto \binom{V(G_2)}{2}$, where $\pi^{\sfE}(uv)=\pi(u)\pi(v)$ for $u,v\in S$. 
For simplicity, we write $\pi(e)$ to denote $\pi^{\sfE}(e)$ when it is clear from the context.

\section{Correlated functional digraph} 
\label{sec:digraph}

A mapping from a set to itself can be graphically represented as a \emph{functional digraph} (see, e.g., \cite[Definition 1.3.3]{west2021combinatorial}). 
Here, we extend this notion to mappings between distinct domain and range sets, where elements from the two sets are correlated. 
Although our focus in this section is on the mappings between the edges of $G_1$ and $G_2$, the graphical representation can be easily generalized to mappings between any two finite sets, such as sets of vertices.

Given a domain subset $S\subseteq V(G_1)$, an injective function $\pi: S\mapsto V(G_2)$, and a bivariate function $f:\mathbb{R}\times \mathbb{R}\mapsto \mathbb{R}$, we define 
the \emph{$f$-intersected graph} $\calH_\pi^f$ as 
\[
V(\calH_\pi^f)=V(G_1),\quad \beta_{e}(\maH_\pi^f) =
\begin{cases}
    f\pth{\beta_{e}(G_1),\beta_{\pi^\sfE(e)}(G_2)}, &\text{ if }e\in E(G_1)\cap \binom{S}{2},\\ 0,&\text{ otherwise.}
\end{cases} 
\]
The weight in $\calH_\pi^f$ represents the similarity score under $\pi$.
To establish the possibility results, our estimator maximizes the total similarity score:  
\begin{align}
    \label{eq:hat-pi}    
    \hat{\pi}(f)&\in 
    \mathop{\text{argmax}}\limits_{ \pi\in  \maS_{n,m}} \en(\maH_{\pi}^f)
    =\mathop{\text{argmax}}\limits_{ \pi\in  \maS_{n,m}} \sum_{e\in E(G_1)} \beta_{e}(\maH_\pi^f).
\end{align}
Specifically, we use $\hat{\pi}(f)$ with $f(x,y) = xy$ to maximize overlap, and $f(x,y) = -\frac{1}{2}(x-y)^2$ to minimize the mean-squared distance. For notational simplicity, we write $\hat{\pi}=\hat{\pi}(f)$ when the choice of $f$ is clear from the context.

More generally, in our analysis in Section~\ref{sec:possibility}, we need to calculate the total weight within a subset $\calE\subseteq \binom{S}{2}$ given by 
\begin{equation}
\label{eq:calE-counter}    
    \beta_\maE(\maH_\pi^f)
\triangleq \sum_{e \in\calE} \beta_{e}(\maH_\pi^f).
\end{equation}
Due to the correlation between the edges in $G_1$ and $G_2$, the summands $\beta_{e}(\maH_\pi^f)$ are correlated random variables. 
The main idea is to decompose $\calE$ into independent parts. 
Specially, the correlation is governed by the underlying mapping $\pi^*$, as illustrated in Figure~\ref{fig:pi-pistart} where the correlated edges are represented by red dashed lines. 
To formally describe all correlation relationships, we introduce the \emph{correlated functional digraph} of a mapping $\pi$ between a pair of graphs. 
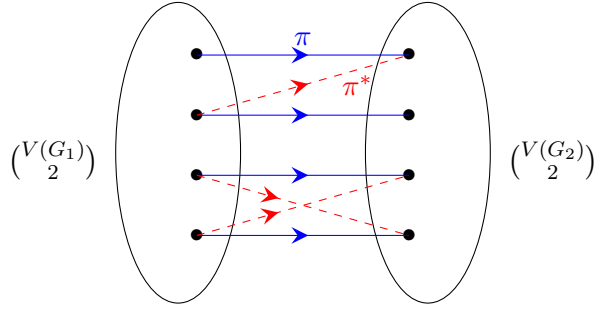
\begin{figure}[t]
    \begin{center}
        \begin{tikzpicture}[yscale = 1, xscale=0.83] 
            \draw (1,1.2) ellipse (1 and 2);
            \draw (5,1.2) ellipse (1 and 2);
            \draw (1.3,2.5) coordinate(A1) node[left,scale=0.7]{} node{$\bullet$};
            \draw (1.3,1.7) coordinate(A2) node[left,scale=0.7]{} node{$\bullet$};
            \draw (1.3,0.9) coordinate(A3) node[left,scale=0.7]{} node{$\bullet$};
            \draw (1.3,0.1) coordinate(A4) node[left,scale = 0.7]{} node{$\bullet$};
            \draw (4.7,2.5) coordinate(B1) node[right,scale=0.7]{} node{$\bullet$};
            \draw (4.7,1.7) coordinate(B2) node[right,scale=0.7]{} node{$\bullet$};
            \draw (4.7,0.9) coordinate(B3) node[right,scale=0.7]{} node{$\bullet$};
            \draw (4.7,0.1) coordinate(B4) node[right,scale=0.7]{} node{$\bullet$};
            \draw[blue] (A1) -- (B1) node[midway,above,sloped,scale=1](TextNode){$\pi$} node[sloped, pos=0.5,scale=2]{\arrowIn};
            \draw[blue] (A2) -- (B2) node[sloped, pos=0.5,scale=2]{\arrowIn};
            \draw[blue] (A3) -- (B3) node[sloped, pos=0.5,scale=2]{\arrowIn};
            \draw[blue] (A4) -- (B4) node[sloped, pos=0.5,scale=2]{\arrowIn};
            \draw[red,dashed] (A2) -- (B1) node[midway,right=1em,scale=1](TextNode){$\pi^*$} node[sloped, pos=0.5,scale=2]{\arrowIn};
            \draw[red,dashed] (A3) -- (B4) node[sloped, pos=0.36,scale=2]{\arrowIn};
            \draw[red,dashed] (A4) -- (B3) node[sloped, pos=0.36,scale=2]{\arrowIn};
            \draw (-1,1.5) node[below]{$\binom{V(G_1)}{2}$};
            \draw (7,1.5) node[below]{$\binom{V(G_2)}{2}$};
        \end{tikzpicture}
    \end{center}
    \caption{Examples of the mapping $\pi$ and the underlying correlation $\pi^*$, where the domain and range of $\pi$ and $\pi^*$ could be different.}
    \label{fig:pi-pistart}
\end{figure}

\begin{definition}[Correlated functional digraph]\label{def:func-diagraph}
    Let $\pi^*: S^* \mapsto T^*$ be the underlying mapping between correlated elements. The correlated functional digraph of the function $\pi: S \mapsto T$ is constructed as follows. 
    Let the vertex set be $\binom{S}{2}\cup \binom{S^*}{2} \cup \binom{T}{2} \cup \binom{T^*}{2}$.
    We first add every edge $e\mapsto \pi(e)$ for $e\in \binom{S}{2}$, and then merge each pair of nodes $(e,\pi^*(e))$ for $e\in \binom{S^*}{2}$ into one node. 
\end{definition}

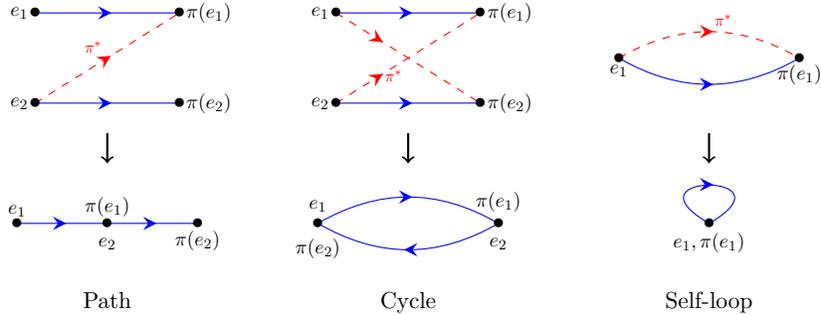
\begin{figure}[htb]
    \begin{center}
        \begin{tikzpicture}[every loop/.style={min distance=17mm,in=150,out=30,looseness=20},global scale = 0.8] 
            \draw (1.3,3.5) coordinate (A1);
            \draw (1.3,2) coordinate (A2);
            \draw (3.7,3.5) coordinate (B1);
            \draw (3.7,2) coordinate (B2);
            \draw[blue] (A1) -- node[above right,scale = 0.7]{} (B1) node[sloped, pos=0.5,scale=2]{\arrowIn};
            \draw[blue] (A2) -- node[above right,scale = 0.7]{} (B2) node[sloped, pos=0.5,scale=2]{\arrowIn};
            \draw[red,dashed] (A2) -- node[above left,scale = 0.7] {$\pi^*$} (B1) node[sloped, pos=0.5,scale=2]{\arrowIn};
            \draw[->,thick] (2.5,1.5) -- (2.5,1);
            \draw (A1) node[left,scale = 0.85] {$e_1$} node{$\bullet$};
            \draw (A2) node[left,scale = 0.85] {$e_2$} node{$\bullet$};
            \draw (B1) node[right,scale = 0.85] {$\pi(e_1)$} node{$\bullet$};
            \draw (B2) node[right,scale = 0.85] {$\pi(e_2)$} node{$\bullet$};
            \draw (1,0) coordinate (C1);
            \draw (2.5,0) coordinate (C2);
            \draw (4,0) coordinate (C3);
            \draw[blue] (C1) -- node[above,scale = 0.7]{} (C2) node[sloped, pos=0.5,scale=2]{\arrowIn};
            \draw[blue] (C2) -- node[above,scale = 0.7]{} (C3) node[sloped, pos=0.5,scale=2]{\arrowIn};
            \draw (C1) node[above,scale = 0.85]{$e_1$} node{$\bullet$};
            \draw (C2) node[above,scale = 0.85]{$\pi(e_1)$} node[below=3pt,scale = 0.85]{$e_2$} node{$\bullet$};
            \draw (C3) node[below,scale = 0.85] {$\pi(e_2)$} node{$\bullet$};
            \draw (2.5,-1) node[below]{Path};

            \draw (6.3,3.5) coordinate (D1);
            \draw (6.3,2) coordinate (D2);
            \draw (8.7,3.5) coordinate (E1);
            \draw (8.7,2) coordinate (E2);
            \draw[blue] (D1) -- node[above right,scale = 0.7]{} (E1) node[sloped, pos=0.5,scale=2]{\arrowIn};
            \draw[blue] (D2) -- node[above right,scale = 0.7]{} (E2) node[sloped, pos=0.5,scale=2]{\arrowIn};
            \draw[red,dashed] (D2) -- node[right,pos = 0.3,scale = 0.7] {$\pi^*$} (E1) node[sloped, pos=0.3,scale=2]{\arrowIn};
            \draw[red,dashed] (D1) -- (E2) node[sloped, pos=0.3,scale=2]{\arrowIn};
            \draw[->,thick] (7.5,1.5) -- (7.5,1);
            \draw (D1) node[left,scale = 0.85] {$e_1$} node{$\bullet$};
            \draw (D2) node[left,scale = 0.85] {$e_2$} node{$\bullet$};
            \draw (E1) node[right,scale = 0.85] {$\pi(e_1)$} node{$\bullet$};
            \draw (E2) node[right,scale = 0.85] {$\pi(e_2)$} node{$\bullet$};
            \draw (6,0) coordinate (F1);
            \draw (9,0) coordinate (F2);
            \draw[blue] (F1) edge[bend left=30] node[above,scale = 0.7]{} node[sloped, pos=0.5,scale=2]{\arrowIn} (F2);
            \draw[blue] (F2) edge[bend left=30] node[below,scale = 0.7]{}  node[sloped, pos=0.5,scale=2]{\arrowOut} (F1);
            \draw (F1) node[above=2pt,scale = 0.85]{$e_1$} node[below=3pt,scale = 0.85]{$\pi(e_2)$} node{$\bullet$};
            \draw (F2) node[above=2pt,scale = 0.85]{$\pi(e_1)$} node[below=3pt,scale = 0.85]{$e_2$} node{$\bullet$};
            \draw (7.5,-1) node[below]{Cycle};

            \draw (11,2.75) coordinate (G1);
            \draw (14,2.75) coordinate (H1);
            \draw[blue] (G1) edge[bend right = 30] node[below,scale = 0.7] {} node[sloped, pos=0.5,scale=2]{\arrowIn} (H1);
            \draw[red,dashed] (G1) edge[bend left = 30] node[above right,pos = 0.5,scale = 0.7] {$\pi^*$}  node[sloped, pos=0.5,scale=2]{\arrowIn} (H1);
            \draw[->,thick] (12.5,1.5) -- (12.5,1);
            \draw (G1) node[below,scale = 0.85]{$e_1$} node{$\bullet$};
            \draw (H1) node[below,scale = 0.85]{$\pi(e_1)$} node{$\bullet$};
            \draw (12.5,0) coordinate (I1);
            \path[blue] (I1) edge [loop above] node[above right,scale = 0.7] {}  (I1) ;
            \draw[blue] (12.5,0.93) node[below,scale=2]{\arrowIn};
            \draw (I1) node[below,scale = 0.85] {$\begin{array}{c} e_1, \pi(e_1) \end{array}$} node{$\bullet$};
            \draw (12.5,-1) node[below]{Self-loop};

        \end{tikzpicture}
    \end{center}
    \caption{The connected components in the correlated functional digraph.}
    \label{fig:digraph}
\end{figure}


It should be noted that both $\pi$ and $\pi^*$ are injective mappings under our model. 
After merging all pairs of nodes according to $\pi^*$, the degree of each vertex in the correlated functional digraph is at most two. 
Therefore, the connected components consist of paths and cycles, where a self-loop is understood as a cycle of length one. 
The connected components are illustrated in Figure~\ref{fig:digraph}.
{
Following prior works on graph matching such as~\cite{racz2021correlated} and~\cite{yang2023graph}, we lift the vertex mapping $\pi$ to an edge mapping $\pi^{\sfE}$. 
When the support sets are known, the lifted permutations can be decomposed into cycles as independent components. 
{
This viewpoint is consistent with the standard cycle-path decomposition that has been fruitfully used in graph alignment analysis. For instance, \cite{cullina2020partial} explicitly decomposed the lifted matching into cycles and paths in their partial-recovery analysis. \cite{gaudio2022exact} employed a lifted-object perspective that accommodates path and cycle components when analyzing correlated stochastic block models. \cite{onaran2016optimal} studied de-anonymization under community structure, where analogous decompositions are handled in a more general fashion. Our contribution is to tailor this decomposition to the partially correlated setting with unknown support, and to exploit the resulting  factorization in our threshold analysis.}


}

Let $\calP$ and $\calC$ denote the collections of subsets of $\calE$ belonging to different connected paths and cycles, respectively. 
Note that the sets from $\calP$ and $\calC$ are disjoint. 
Consequently, 
\[
\beta_{\maE}(\maH_\pi^f)
= \sum_{P\in\calP} \beta_P(\maH_\pi^f)
+ \sum_{C\in\calC} \beta_C(\maH_\pi^f),
\]
where the summands are mutually independent.

In our models, edge correlations are assumed to be homogeneous, implying that the distributions of $\beta_P(\maH_\pi^f)$ and $\beta_C(\maH_\pi^f)$ depend only on the size of the component. 
Let $\kappa^{\sfP,f}_\ell(t)$ and $\kappa^{\sfC,f}_\ell(t)$ denote the cumulant generating functions of $\beta_P(\maH_\pi^f)$ and $\beta_C(\maH_\pi^f)$ for components of order $\ell$, respectively. Then, we have 
\begin{align*}
    \log\expect{e^{t \beta_P(\maH_\pi^f)}} = \kappa^{\sfP,f}_{|P|}(t),
    \qquad 
    \log\expect{e^{t \beta_C(\maH_\pi^f)}} = \kappa^{\sfC,f}_{|C|}(t).
\end{align*}
For simplicity, we write $\kappa_{\ell}^\sfC(t) = \kappa_{\ell}^{\sfC,f}(t)$ and $\kappa_{\ell}^\sfP(t) = \kappa_{\ell}^{\sfP,f}(t)$ when the function $f$ is specified. 
Lower-order cumulants can be calculated directly. See more details in Appendix~\ref{apdsec: upper bound for kappa}. 
However, it is crucial to establish upper bounds for higher-order cumulants in terms of the lower-order ones. To this end, we introduce the following lemma.


\begin{lemma}\label{lem: upper bound for kappa}
For any $0<p,\rho<1$,
    \begin{align*}
        \kappa^{\sfP}_1(t) \le \frac{1}{2} \kappa^{\sfC}_2(t) \le \kappa^{\sfC}_1(t)
        \quad \text{and} \quad 
        \kappa^{\sfP}_\ell(t)\le   \kappa^{\sfC}_\ell(t) \le \frac{\ell}{2} \kappa^{\sfC}_2(t),\quad \forall~\ell\ge 2,
    \end{align*}
    under any of the following three conditions:\begin{itemize} 
        \item In the \ER model with $f(x,y) = xy$ and $t>0$;
        \item In the Gaussian Wigner model with $f(x,y) = xy$ and $0<t<\frac{1}{1+\rho}$;
        \item In the Gaussian Wigner model with $f(x,y) = -\frac{1}{2}(x-y)^2$ and $t>0$.
    \end{itemize}
    Consequently, 
    \begin{equation}
        \label{eq:cumulant-ub}    
        \log \expect{e^{t\beta_\maE(\maH_\pi^f)}} \le \frac{|\calE|}{2} \kappa^{\sfC}_2(t)+L\pth{\kappa^{\sfC}_1(t)-\frac{1}{2} \kappa^{\sfC}_2(t)},
    \end{equation}
    where $L$ denotes the number of self-loops.
\end{lemma}


The proof of Lemma~\ref{lem: upper bound for kappa} 
is deferred to Appendix~\ref{apdsec: upper bound for kappa}. 
The special case where both $\pi$ and $\pi^*$ are bijective has been studied in \cite{wu2022settling, ding2023matching, hall2023partial}, where the correlation relationships can be characterized by the permutation $(\pi^*)^{-1} \circ \pi$. 
In this case, the connected components of the functional digraph of permutations are all cycles. 
In contrast, in our setting, the domains and ranges of $\pi$ and $\pi^*$ may differ, requiring us to deal with the intricate correlations among edges involving both cycles and paths, as addressed by Lemma~\ref{lem: upper bound for kappa}. 

\begin{remark}
    \hd{For the \ER model, when $f(x,y) = xy$, the estimator~\eqref{eq:hat-pi} is equivalent to computing the maximal edge overlap over all possible injective mappings. 
    The recent work \cite{ding2024polynomial} approximated the maximal edge overlaps within a constant factor in polynomial-time for sparse \ER graphs, and \cite{du2023algorithmic} established a sharp transition on approximating problem on the performance of online algorithms for dense \ER graphs.} 
    \hd{Assume $p = n^{-a_1}$, $\rho = n^{-a_2}$, and $m = n^{a_3}$. It is shown in~\cite{ding2024polynomial} that the maximal edge overlap between two independent \ER graphs $\mathcal{G}(m, p)$ is $\frac{m}{2\alpha - 1}$ when $p = m^{-\alpha}$ for $\alpha \in (1/2, 1)$. 
In contrast, in the partially correlated model, the maximal edge overlap is approximately
\(
\binom{m}{2} p_{11} \asymp n^{2a_3 - a_1 - a_1 \wedge a_2}.
\)
Under the condition for successful recovery, we have $a_3 \ge a_2 + a_1 \vee a_2\ge a_1 + a_1 \wedge a_2$, and it follows that
\(
\binom{m}{2} p_{11} \gtrsim n^{a_3} \asymp \frac{m}{2\alpha - 1}.
\)
Therefore, the maximal edge overlap in the partially correlated model exceeds that in the independent case.
}
\end{remark}

\begin{remark}
For the Gaussian Wigner model, the MLE is defined by~\eqref{eq:hat-pi} with the following similarity function: 
\begin{equation}\label{eq:MLE}  
    f_{\mathrm{MLE}}(x,y) = -\frac{\rho}{2}\pth{x^2+y^2}+xy.
\end{equation}
{
Specifically, \begin{itemize}
    \item when $1-\rho = \Omega(1)$, $f_{\mathrm{MLE}}(x,y)\asymp xy$;
    \item when $1-\rho  = o(1)$, $f_{\mathrm{MLE}}(x,y)\asymp -\frac{1}{2}(x-y)^2$.
\end{itemize}}%
We analyze approximations of the above similarity score under different regimes.
Specifically, when $\rho=1-\Omega(1)$, we apply
the estimator with $f(x,y) =xy$, 
whereas for $\rho=1-o(1)$ the estimator with $f(x,y) =-\frac{1}{2}(x-y)^2$ turns out to be crucial.
\end{remark}



\section{Recovery by maximizing total similarity score}
\label{sec:possibility}
In this section, we establish the possibility results by analyzing the estimators $\hat\pi$ defined in \eqref{eq:hat-pi}. For any two injections $\pi:S\mapsto T$ and $\pi':S'\mapsto T'$ with $|S| = |S'| = m$, let $d(\pi,\pi')\triangleq m - |\{v\in S\cap S':\pi(v)=\pi'(v)\}|$. Indeed, $d(\pi,\pi') = m(1-\overlap(\pi,\pi'))$, where $\overlap(\pi,\pi')$ is defined in~\eqref{def:overlap}.
By the optimality condition, it suffices to show that 
\begin{align}\label{eq:optimal-condition}
    \en(\calH_{\pi^*}^f)
> \max_{\pi:d(\pi,\pi^*)\ge d_0}\en(\calH_{\pi}^f)
= \max_{k\ge d_0}\max_{\pi:d(\pi,\pi^*)=k}\en(\calH_{\pi}^f)
\end{align}
with high probability, 
where the thresholds $d_0=1$ and $(1-\delta) m$ correspond to exact and partial recoveries, respectively. 
\hd{Recall that the weight in $\maH_\pi^f$ represents the similarity score under injection $\pi$, and our estimator $\hat{\pi}$ maximizes this similarity score. To achieve partial or exact recovery, it suffices to show that the total similarity score under the true mapping $\pi^*$ is larger as in~\eqref{eq:optimal-condition}.}

In the following, we 
outline a general recipe to derive an upper bound for 
the failure event $
\{\max_{\pi\in \maT_k}\en(\calH_{\pi}^f)\ge \en(\calH_{\pi^*}^f)\}$ for a fixed $k$, where $\maT_k\subseteq \maS_{n,m}$ denotes the set of injections $\pi$ such that $d(\pi,\pi^*) = k$.
The error probabilities in the main theorems are then obtained by summing over the corresponding range of $k$.

For any $\pi\in \maT_k$, by definition, there exists a set of correctly matched vertices of cardinality $m-k$---corresponding to the self-loops in the correlated functional digraph of $\pi$ over the vertices---denoted by 
$F_\pi\triangleq \{v\in S^*\cap S:\pi^*(v) = \pi(v)\}$ with $|F_\pi|=m-k$. 
\hd{By the definition of $F_\pi$, the two induced subgraphs $\calH_\pi^f[F_\pi]$ and $\calH_{\pi^*}^f[F_\pi]$ are identical, and so are the corresponding edge summations $\en(\maH_\pi^f[F_\pi])=\en(\maH_{\pi^*}^f[F_\pi])$.}
Consequently, 
\[
\en(\calH_{\pi}^f)\ge \en(\calH_{\pi^*}^f)
\iff 
\en(\calH_{\pi}^f) - \en( \calH_{\pi}^f[F_\pi] ) \ge \en(\calH_{\pi^*}^f) - \en( \calH_{\pi^*}^f[F_\pi] ).
\]
It should be noted that correlated random variables are present on both sides of the inequality.
Nevertheless, for any threshold $\tau_k$, either $\en(\calH_{\pi^*}^f) - \en(\calH_{\pi^*}^f[F_\pi])<\tau_k$ or $\en(\calH_{\pi}^f)-\en(\calH_{\pi}^f[F_\pi])\ge \tau_k$ must hold.
This leads to the following upper bound:
\begin{align}\label{eq:two_event}
\sth{\max_{\pi\in \maT_k}\en(\calH_{\pi}^f)\ge \en(\calH_{\pi^*}^f)}
\subseteq \bigcup_{\pi\in\calT_{k}}\{\en(\calH_{\pi^*}^f) - \en(\calH_{\pi^*}^f[F_\pi])<\tau_k\} \cup \{\en(\calH_{\pi}^f)-\en(\calH_{\pi}^f[F_\pi])\ge \tau_k\}  .
\end{align}
The first event indicates the presence of a weak signal, while the second reflects the impact of strong noise. 
The key result to establish is that, for a suitable $\tau_k$, the probabilities of these bad events are small. Here, $\tau_k$ can be selected as a function of $m,k,p,\rho$, and  $f$. For brevity, we write $\tau_k = \tau(m,k,p,\rho,f)$.

\paragraph{Bad event of signal.}
For a fixed $\pi\in \maT_k$, the random variable $\en(\calH_{\pi^*}^f) - \en(\calH_{\pi^*}^f[F_\pi])$ represents the total weight across $N_k\triangleq \binom{m}{2}-\binom{m-k}{2}=mk(1-\frac{k+1}{2m})$ pairs of vertices.
Additionally, $F_{\pi}$ is a subset of $S^*$ of cardinality $m-k$,
and the total number of possible configurations for $F_{\pi}$ is at most $\binom{m}{m-k}=\binom{m}{k}$.
Therefore, 
\begin{align}
    \prob{ \bigcup_{\pi\in\calT_{k}} \sth{\en(\calH_{\pi^*}^f) - \en(\calH_{\pi^*}^f[F_\pi])<\tau_k}}
& \le \prob{ \bigcup_{\substack{F\subseteq S^*\\ |F|=m-k}} \sth{ \en(\calH_{\pi^*}^f) - \en(\calH_{\pi^*}^f[F])<\tau_k}}\nonumber \\&\le\binom{m}{k}\prob{\en(\maH_{\pi^*}^f)-\en(\maH_{\pi^*}^f[F]) < \tau_k}.\label{eq:bad-signal-total} 
\end{align}

\paragraph{Bad event of noise.}
The analysis of the noise component is more challenging due to the mismatch between $\pi$ and the underlying $\pi^*$. Let $S_\pi$ denote the domain of $\pi$, and $\calE_{\pi} \triangleq \binom{S_\pi}{2}-\binom{F_\pi}{2}$ with $|\calE_{\pi}|=N_k$. The total weight $\en(\calH_{\pi}^f)-\en(\calH_{\pi}^f[F_\pi])$ can be equivalently expressed as $\beta_{\maE_\pi}(\maH_\pi^f)$. The cumulant generating function for this quantity is upper bounded in Lemma~\ref{lem: upper bound for kappa}, utilizing the decomposition provided by the correlated functional digraph.
As a result, the error probability can be evaluated using the Chernoff bound by optimizing over $t>0$ in~\eqref{eq:cumulant-ub}. 

To this end, we need to upper bound the number of self-loops in~\eqref{eq:cumulant-ub}.
For a self-loop over an edge $e=uv$, we have $\pi(uv)=\pi^*(uv)$. 
Note that $\calE_{\pi}$ excludes the edges in the induced subgraph over $F_\pi$. 
Therefore, it necessarily holds that $\pi(u)=\pi^*(v)$ and $\pi(v)=\pi^*(u)$, which contribute two mismatched vertices in the reconstruction of the underlying mapping.
Since the total number of mismatched vertices for $\pi\in\calT_k$ equals $k$, the number of self-loops is at most $\frac{k}{2}$. 
Consequently, applying~\eqref{eq:cumulant-ub} with the formula for lower-order cumulants provides an upper bound for $\prob{\beta_{\maE_\pi}(\maH_\pi^f)\ge \tau_k}$. 

It remains to upper bound the total number of $\pi$ in $\maT_k$. 
To do so, we first choose $m-k$ elements from the domain of $\pi^*$ and map them to the same value as $\pi^*$. Then, the remaining domain and range, each of size $k$, are matched arbitrarily. 
This yields: 
\begin{align*}
    |\maT_k|\le \binom{m}{m-k}\binom{n-m+k}{k}^2 k!
    \overset{\mathrm{(a)}}{\le
    } \frac{m^k n^{2k}}{k!^2}\overset{\mathrm{(b)}}{\le} \frac{n^{3k}}{k!^2},
\end{align*}
where $\mathrm{(a)}$ uses the bound $\binom{m}{m-k}\le \frac{m^k}{k!}$ and $\binom{n-m+k}{k}\le \frac{(n-m+k)^k}{k!}\le \frac{n^k}{k!}$, and $\mathrm{(b)}$ is because $m\le n$. Therefore, \begin{align}\label{eq:strong-noise-union-bound}
    \prob{ \bigcup_{\pi\in\calT_{k}} \sth{ \en(\calH_{\pi}^f) - \en(\calH_{\pi}^f[F_\pi])\ge \tau_k } } \le \frac{n^{3k}}{k!^2} \prob{\beta_{\maE_\pi}(\maH_\pi^f)\ge \tau_k}.
\end{align}


\subsection{\ER model}
\label{sec:ER}

In this subsection, we focus on the \ER model and use $\hat{\pi}$ from \eqref{eq:hat-pi} with $f(x,y) = xy$ as the estimator. 
For the bad event of signal, since $\en(\maH_{\pi^*}^f) - \en(\maH_{\pi^*}^f[F])\sim \Bin(N_k,p_{11})$ with $|F| =k$, the error probability follows from the standard Chernoff bound (see, e.g., \eqref{eq:chernoff_bound_left}): for $0<\eta<1$,
\begin{align}\label{eq:weak-signal}
    \prob{ \en(\calH_{\pi^*}^f) - \en(\calH_{\pi^*}^f[F_\pi])<\tau_k} \le \exp\pth{-\frac{N_k p_{11}\eta^2}{2}},
    \quad 
    \tau_k=N_k p_{11}(1-\eta).
\end{align}
For the bad event of noise, applying \eqref{eq:cumulant-ub} along with the lower-order cumulant~\eqref{eq:cumulants-C1} and~\eqref{eq:cumulants-C2} yields the following lemma, whose proof is deferred to Appendix~\ref{apdsec: strong noise event}.

\begin{lemma}\label{lem: strong noise event}
In the \ER model, for $f(x,y) = xy$ and $\pi\in \maS_{n,m}$ with $k=d(\pi,\pi^*)$, if $\tau_k>|\maE_\pi|p^2$, then
\begin{equation}
    \prob{\beta_{\maE_\pi}(\maH_\pi^f) \ge \tau_k}
    \le \exp\pth{-\frac{\tau_k}{2}\log\pth{\frac{\tau_k}{ |\maE_\pi|p^2}}+\frac{\tau_k}{2}-\frac{|\maE_\pi|p^2}{2}+\frac{k\gamma}{4(2+\gamma)}}.\label{eq:strong noise event}
\end{equation}
\end{lemma}

Since $p_{11}=p^2(1+\gamma)$, we need to choose an appropriate $\eta\in(0,1)$ depending on the model parameters, such that $(1+\gamma)(1-\eta)>1$, in order to apply the error probabilities in~\eqref{eq:weak-signal} and~\eqref{eq:strong noise event} to upper bound \eqref{eq:bad-signal-total} and \eqref{eq:strong-noise-union-bound}, respectively. 
When $m$ exceeds the sample complexities for partial and exact recovery, we explicitly construct $\eta$ to prove the following propositions, which establish the possibility results stated in Theorems~\ref{thm:partial recovery} and~\ref{thm:exact recovery}. The proofs of Propositions~\ref{prop: Upper bound for partial recovery} and~\ref{prop: Upper bound for exact recovery} are deferred to Appendix~\ref{subapd:Upper bound for partial recovery} and~\ref{subapd:Upper bound for exact recovery}, respectively.

\begin{proposition}[\ER model, upper bound for partial recovery]\label{prop: Upper bound for partial recovery}
    For any $\delta\in(0,1)$, there exists a constant $c_1(\delta)>0$ such that, 
    when $m\ge \frac{c_1(\delta) \log n}{p^2 \phi(\gamma)}$, 
    for any $\pi^*\in\maS_{n,m}$, the estimator in \eqref{eq:hat-pi} with $f(x,y) = xy$ satisfies
    \[
    \prob{\overlap(\hat{\pi},\pi^*)<\delta}\le \frac{4\log m+2}{m}.
    \]
\end{proposition}

\begin{proposition}[\ER model, upper bound for exact recovery]\label{prop: Upper bound for exact recovery}
    There exists a universal constant $C_1>0$ such that, when $m\ge C_1\pth{\frac{\log(1/p^2 \gamma)}{p^2 \gamma} \vee\frac{\log n}{p^2 \phi(\gamma)}}$, 
    for any $\pi^*\in\maS_{n,m}$, the estimator in \eqref{eq:hat-pi} with $f(x,y) = xy$ satisfies
    \[
    \prob{\hat{\pi}\ne \pi^*}\le \frac{1}{m-1}+\frac{1}{n-1}.
    \]
\end{proposition}

\subsection{Gaussian Wigner model}

In this subsection, we focus on the Gaussian Wigner model. \hd{Different from Section~\ref{sec:ER}, the Gaussian Wigner model requires a different estimator as $\rho \to 1$. Intuitively, when $\rho = 1$, the edge weights satisfy $\beta_e(G_1) = \beta_{\pi^*(e)}(G_2)\in \mathbb{R}$ for all $e\in \binom{S^*}{2}$. Since each weighted edge is marginally standard normal, the ground truth permutation $\pi^*$ can be exactly recovered by comparing edge weights pairwise, as long as $m \ge 3$ (the node correspondence is not identifiable when $m=2$).}
To establish upper bounds for exact recovery, we use different estimators depending on the value of $\rho$, where $0 < \rho \le 1 - e^{-12}$ corresponds to the weak signal regime, and $1 - e^{-12} < \rho < 1$ corresponds to the strong signal regime.

\paragraph{Weak signal}
Under the weak signal regime $0<\rho\le 1-e^{-12}$, we use $\hat{\pi}$ in \eqref{eq:hat-pi} with $f(x,y) = xy$ as our estimator. Similar to \ER model, we upper bound the event $\{d(\hat{\pi},\pi^*) = k\}$ by the bad event of signal and the bad event of noise according to \eqref{eq:two_event}. For the bad event of signal, we note that $\en(\maH_{\pi^*}^f) - \en(\maH_{\pi^*}^f[F]) = \sum_{i=1}^{N_k} A_i B_i$, where $(A_i,B_i)$ are independent and identically distributed (i.i.d.) pairs of standard normals with correlation coefficient $\rho$. Then, the error probability can be upper bounded by 
\begin{align}
\prob{ \bigcup_{\pi\in\calT_{k}} \sth{ \en(\calH^f_{\pi^*}) - \en(\calH^f_{\pi^*}[F_\pi])<\tau_k } }
&  \le \binom{m}{k}\prob{ \sum_{i=1}^{N_k} A_i B_i< \tau_k}, \label{eq:weak-signal-Gaussian}
\end{align}
the tail follows from the Hanson-Wright inequality in Lemma~\ref{lem:Hanson-Wright}. For the bad event of noise, applying \eqref{eq:cumulant-ub} with the formula of lower-order cumulants~\eqref{eq:cumulants-C1-Gaussian} and~\eqref{eq:cumulants-C2-Gaussian} yields the following Lemma, whose proof is deferred to Appendix~\ref{apdsec: strong noise event-Gaussian}.

\begin{lemma}\label{lem: strong noise event-Gaussian}
In the Gaussian Wigner model, for $f(x,y) = xy$, any threshold $\tau_k$, and $\pi\in \maS_{n,m}$ with $k=d(\pi,\pi^*)$, we have \begin{align}\label{eq:strong-noise-Gaussian}
    \prob{\beta_{\maE_\pi}\pth{\maH_\pi^f} \ge \tau_k}\le \exp\pth{-\frac{\rho\tau_k}{6}+\frac{\rho^2|\maE_\pi|}{14}+\frac{\log 5}{8}k}.
\end{align}    
\end{lemma}
The following proposition provides sufficient condition on $m$ for exact recovery in the Gaussian Wigner model when $0<\rho\le 1-e^{-12}$, whose proof is deferred to Appendix~\ref{subapd:Upper bound for exact recovery-Gaussian}.
\begin{proposition}[Gaussian Wigner model, upper bound for exact recovery]\label{prop: Upper bound for exact recovery-Gaussian}
When $0<\rho \le 1-e^{-12}$, there exists a universal constant $C_3>0$ such that, when $m\ge \frac{C_3\log n}{\rho^2}$, for any $\pi^*\in \maS_{n,m}$, the estimator in \eqref{eq:hat-pi} with $f(x,y) = xy$ satisfies \begin{align*}
    \prob{\hat{\pi}\neq \pi^*}\le \frac{1}{m-1}+\frac{1}{n-1}.
\end{align*}
\end{proposition}

\paragraph{Strong signal}\label{para:strong-signal}
Under the strong signal regime $1-e^{-12}<\rho<1$, we use $\hat{\pi}$ in \eqref{eq:hat-pi} with $f(x,y) = -\frac{1}{2}(x-y)^2$ as our estimator. 
Indeed, the MLE is also of the form~\eqref{eq:hat-pi}, but with $f(x,y) = \rho xy - \frac{\rho^2}{2}(x^2+y^2)$, which is approximately $-\frac{1}{2}(x-y)^2$ when $\rho = 1-o(1)$. This choice reflects the phase transition phenomenon in the Gaussian Wigner model when $\rho = 1-o(1)$. Specifically, for a correlated edge pair, the expected mean squared difference $\expect{\pth{\beta_e(G_1) - \beta_{\pi^*(e)}(G_2)}^2} = o(1)$, while for an uncorrelated edge pair, it stays bounded away from zero with high probability. Our proof follows a similar structure with the weak signal regime.
For the bad event of signal, since $\en(\maH_{\pi^*}^f) - \en(\maH_{\pi^*}^f[F]) = \sum_{i=1}^{N_k} -\frac{1}{2}(A_i- B_i)^2$, where $(A_i,B_i)$ are i.i.d. pairs of standard normals with correlation coefficient $\rho$. Then, the error probability can be upper bounded by 
\begin{align}
\prob{ \bigcup_{\pi\in\calT_{k}} \sth{ \en(\calH_{\pi^*}^f) - \en(\calH_{\pi^*}^f[F_\pi])<\tau_k } }
&  \le \binom{m}{k}\prob{ \sum_{i=1}^{N_k} -\frac{1}{2}(A_i- B_i)^2< \tau_k}, \label{eq:weak-signal-Gaussian-2}
\end{align}
the tail follows from the concentration inequality for chi-squared distribution in Lemma~\ref{lem:chisquare}. For the bad event of noise, applying \eqref{eq:cumulant-ub} with the formula of lower-order cumulants~\eqref{eq:kappa1=general} and~\eqref{eq:kappa2-general} yields the following Lemma, whose proof is deferred to Appendix~\ref{apdsec: strong noise event-Gaussian-2}.

\begin{lemma}\label{lem:strong noise event-Gaussian-2}
    In the Gaussian Wigner model, for $f(x,y) = -\frac{1}{2}(x-y)^2$, any threshold $\tau_k$, and $\pi\in \maS_{n,m}$ with $k=d(\pi,\pi^*)$, we have \begin{align*}
        \prob{\beta_{\maE_\pi}(\maH_\pi^f)\ge \tau_k}\le \exp\pth{ -\frac{\rho\tau_k}{4(1-\rho)}- \frac{|\maE_\pi|}{4}\log\pth{\frac{1}{1-\rho}}+\frac{k}{8}\log\pth{\frac{1}{1-\rho}}}.
    \end{align*}
\end{lemma}


The following proposition provides a sufficient condition on $m$ for exact recovery in the Gaussian Wigner model when $1-e^{-12}<\rho<1$, whose proof is deferred to Appendix~\ref{subapd:upbd_gaussian_general}. 

\begin{proposition}\label{prop: upbd_gaussian_general}
    When $1-e^{-12}<\rho<1$, there exists a universal constant $C_4>0$ such that, when $m\ge C_4\pth{\frac{\log n}{\log(1/(1-\rho))}\vee 1}$, for any $\pi^*\in \maS_{n,m}$, the estimator in~\eqref{eq:hat-pi} with $f(x,y) = -\frac{1}{2}(x-y)^2$ satisfies \begin{align*}
        \prob{\hat{\pi}\neq \pi^*}\le  \frac{2}{n-1}.
    \end{align*}
\end{proposition}

In view of Propositions~\ref{prop: Upper bound for exact recovery-Gaussian} and~\ref{prop: upbd_gaussian_general}, we note that $\rho^2\asymp \log\pth{\frac{1}{1-\rho^2}}$ when $0<\rho\le 1-e^{-12}$, and $\log\pth{\frac{1}{1-\rho}}\asymp \log\pth{\frac{1}{1-\rho^2}}$ when $1-e^{-12}<\rho<1$. Then, there exists a universal constant $C_2$ such that, \begin{align*}
    C_2 \pth{\frac{\log n}{\log(1/(1-\rho^2))}\vee 1}\ge \begin{cases}
        \frac{C_3 \log n}{\rho^2}, &0<\rho\le 1-e^{-12},\\
        C_4\pth{\frac{\log n}{\log(1/(1-\rho))}\vee 1}, &1-e^{-12}<\rho<1.
    \end{cases}
\end{align*}
Therefore, we prove the possibility results in Theorem~\ref{thm:gauss-main}.

\begin{remark}
    While the proofs for the possibility results in the \ER and Gaussian Wigner models share a similar structure, they also exhibit several key differences. 
    Following the intuition of MLE in~\eqref{eq:MLE}, we analyze two different estimators for the Gaussian Wigner model.
    Furthermore, by applying different concentration inequalities, we observe a significant difference in the choice of $\tau_k$. Specifically, we set $\tau_k = \Theta\pth{\expect{\en\pth{\maH_{\pi^*}^f} - \en\pth{\maH_{\pi^*}^f[F_\pi]}}}$ for both models with $\rho = 1-\Omega(1)$. In contrast, we choose $\tau_k = \omega\pth{\expect{\en\pth{\maH_{\pi^*}^f} - \en\pth{\maH_{\pi^*}^f[F_\pi]}}}$ for the Gaussian model with $\rho = 1-o(1)$. The thresholds are chosen to balance the error probabilities arising from the bad events of signal and noise, ultimately achieving the optimal rates for the correlated size $m$.
\end{remark}

\begin{remark}\label{rmk:unknown-case}
    Our estimator~\eqref{eq:hat-pi} requires knowledge of the cardinality $|S^*| = m$.  \hd{For certain problems, such as induced subgraph sampling, the random variable $m$ can be shown to be highly concentrated by analyzing its hypergeometric distribution~\cite{huang2025sample}.} When $m$ is unknown, the objective function in~\eqref{eq:hat-pi} will increase monotonically with $m$.
    To address this, we can introduce a penalty term for $m$ in the estimator to identify the correct size. 
    \hd{Specifically, when $m$ is unknown, the  estimator takes the form \begin{align*}
    \hat{\pi}\in \mathop{\text{argmax}}\limits_{ \pi\in \maS_{n,m},m\le n} \qth{\en(\maH_\pi^f)-F(m)},
\end{align*}
where $F(m) = \lambda m^2$ is a penalty term and $\lambda$ is a tuning parameter.
Indeed, we choose $F(m)$ such that the estimator encourages correlated pairs and penalizes independent pairs in order to correctly identify $m$.} 
\end{remark}



\section{Impossibility results}
\label{sec:impossibility}

In this section, we present the impossibility results for the graph alignment problem.
Under our proposed model, the alignment problem aims to recover the domain $S^*\subseteq V(G_1)$, range $T^*\subseteq V(G_2)$, and the mapping $\pi^*:S^*\mapsto T^*$. 
When equipped with the additional knowledge on $S^*$ and $T^*$, our problem can be reduced to recovery with full observations on smaller graphs, the reconstruction threshold for which is settled in \cite{wu2022settling}. 
The lower bound therein remains valid when the number of correlated nodes is substituted with $m$. 
However, such reduction only proves tight in a limited number of regimes (see Proposition~\ref{prop:lb-exact}).
We will establish the impossibility results for the remaining regimes by Fano's method \hd{(see, e.g., \cite{fano1961transmission}, ~\cite{yu1997assouad}, and \cite[Section 2.10]{cover2006elements}), which consists of the following two ingredients:} 
\begin{itemize}
    \item \hd{\emph{Separation.} 
    Construct a subset of well-separated parameters. 
    Any testing error then leads to an estimation error proportional to the minimum separation between parameters. 
    }
    \item \hd{\emph{Insufficient information.} The mutual information between the underlying parameter and the observed data is small.
    In this case, any test incurs a non-negligible error probability.}
\end{itemize}

\hd{The technical outlines of the Fano's inequality are sketched below.}

\paragraph{\hd{Constructing a subset of parameters.}} 
Let $\calM_{\delta}$ be a packing set of $\maS_{n,m}$ such that any two distinct elements $\pi,\pi'\in\calM$ differ by a prescribed threshold. 
Specifically, in partial recovery, 
we need $\min_{\pi\ne \pi'\in\calM}d(\pi,\pi')> (1-\delta) m$; in exact recovery, we simply choose $\calM_1=\maS_{n,m}$. 
    The cardinality of $\calM_\delta$ measures the complexity of the parameter space under the target metric.
 
\paragraph{\hd{Bounding the mutual information $I(\pi^*; G_1,G_2)$.}} 
Given $\pi^*$, the conditional distribution of the observed graphs $(G_1,G_2)$ is specified in Definitions~\ref{def:planted correlated ER graph} \hd{and~\ref{def:gaussian-partial}}.
For the mutual information, let $\maP$ denote the joint distribution of $(G_1,G_2)$ and $\maQ$ be any distribution over $(G_1,G_2)$. Then, 
    \begin{equation}
    \label{eq:MI-radius}
    I(\pi^*;G_1,G_2) = \Expect_{\pi^*}[D(\maP_{G_1,G_2|\pi^*}\Vert\maP_{G_1,G_2})] 
    \le \max_{\pi} D(\maP_{G_1,G_2|\pi}\Vert \maQ_{G_1,G_2}),
    \end{equation}
    \hd{where the inequality is because $$D(\maP_{G_1,G_2|\pi^*}\Vert \maP_{G_1,G_2})=D(\maP_{G_1,G_2|\pi^*}\Vert \maQ_{G_1,G_2})-D(\maP_{G_1,G_2}\Vert \maQ_{G_1,G_2})$$ and the KL divergence $D(\maP_{G_1,G_2}\Vert \maQ_{G_1,G_2})\ge 0$ for any distribution $\maQ$.}

{
\paragraph{Applying Fano's inequality.} By Fano's inequality, with $\pi^*$ being the discrete uniform prior in the packing set $\maM_\delta$, for any estimator $\hat{\pi}$, we have\begin{align*}
        \prob{\overlap(\hat{\pi},\pi^*)<\delta} \ge 1-\frac{I(\pi^*;G_1,G_2)+\log 2}{\log |\maM_\delta|},
    \end{align*} 
    where $I(\pi^*;G_1,G_2)$ denotes the mutual information between $\pi^*$ and $(G_1,G_2)$, and $\pi^*$ is uniformly distributed over $\maM_\delta$. 

The relevant quantities are evaluated in the next lemma.

\begin{lemma}\label{lem:mutual-pack}
    For any $0<\delta<1$, we have \begin{align}\label{eq:pack-main}
        |\maM_\delta|\ge \pth{\frac{\delta n}{e^3}}^{\delta m}.
    \end{align}
    In the \ER model, we have \begin{align}\label{eq:mutual-ER-main}
        I(\pi^*;G_1,G_2)\le 25\binom{m}{2}p^2\phi(\gamma).
    \end{align}
    In the Gaussian Wigner model, we have \begin{align}\label{eq:mutual-Gauss-main}
        I(\pi^*;G_1,G_2)\le \frac{1}{2}\binom{m}{2}\log\pth{\frac{1}{1-\rho^2}}.
    \end{align}
\end{lemma}

The proof of Lemma~\ref{lem:mutual-pack} is deferred to Appendix~\ref{apd:proof-mutual-pack}.}
Fano's method provides a lower bound on the Bayesian risk when $\pi$ is uniformly distributed over $\calM_\delta$, which further lower bounds the minimax risk. 
The above argument also yields a lower bound when $\pi$ is uniform over $\maS_{n,m}$ via Fano's inequality. 
For the \ER model, the following propositions provide lower bounds for $m$ for partial recovery and exact recovery, and thus prove the lower bounds in Theorems~\ref{thm:partial recovery} and~\ref{thm:exact recovery}. For the Gaussian Wigner model, the Proposition~\ref{prop:lb-partial} provides a lower bound for $m$ for partial recovery, which implies that $m\le \frac{c\log n}{\log\pth{1/(1-\rho^2)}} \vee \frac{1}{2}$ is sufficient for the lower bound since $m<1$ is impossible for partial recovery. Consequently, we finish the proof of impossibility results in Theorem~\ref{thm:gauss-main} since $\frac{c\log n}{\log\pth{1/(1-\rho^2)}} \vee \frac{1}{2} \asymp \frac{\log n}{\log\pth{1/(1-\rho^2)}} \vee 1$.

\begin{proposition}[Lower bound for partial recovery]
    \label{prop:lb-partial}
    In the \ER model, for any $\delta\in(0,1)$,
    if $m\le \frac{c\log n}{p^2 \phi(\gamma)}$, 
    then for any estimator $\hat \pi$,
    \[
    \prob{\overlap(\hat{\pi},\pi^*)<\delta}\ge 1-\frac{13c}{\delta}.
    \]
    In the Gaussian Wigner model, for any $\delta\ in (0,1)$, if $m\le \frac{c\log n}{\log(1/(1-\rho^2))}$, then for any estimator $\hat{\pi}$,\begin{align*}
        \prob{\overlap(\hat{\pi},\pi^*)<\delta}\ge 1-\frac{c}{2\delta}.
    \end{align*}
\end{proposition}

\begin{proof}
Applying Fano's inequality with~\eqref{eq:pack-main} and \eqref{eq:mutual-ER-main}, we obtain
\begin{equation*}
\prob{\overlap(\pi^*,\hat{\pi})<\delta}\ge 1-\frac{25\binom{m}{2}p^2 \phi(\gamma)+\log 2}{\delta m\log\pth{\frac{\delta n}{e^3}}} \ge 1- \frac{13 c}{\delta },
\end{equation*}
where $\pi^*$ is uniformly distributed over $\maM_\delta$.

Applying Fano's inequality with~\eqref{eq:pack-main} and \eqref{eq:mutual-Gauss-main}, we obtain
\begin{equation*}
\prob{\overlap(\pi^*,\hat{\pi})<\delta}\ge 1-\frac{\frac{1}{2}\binom{m}{2}\log\pth{\frac{1}{1-\rho^2}}+\log 2}{\delta m\log\pth{\frac{\delta n}{e^3}}} \ge 1- \frac{ c}{2\delta},
\end{equation*}
where $\pi^*$ is uniformly distributed over $\maM_\delta$.
\end{proof}

\begin{proposition}[\ER model, lower bound for exact recovery]
    \label{prop:lb-exact}
For any $c\in(0,1)$ and any estimator $\hat{\pi}$, there exists a constant $c_3$  depending on $c$ such that, when $m\le c_3 \pth{\frac{\log n}{p^2 \phi(\gamma)} \vee \pth{\frac{1}{p^2 \gamma} \log\frac{1}{p^2 \gamma}}}$,
    \[
    \prob{\hat{\pi}\ne \pi^*}\ge 1-c,
    \]
    where $\pi^*$ is uniformly distributed over $\calS_{n,m}$.

\end{proposition}
\begin{proof}
    We first apply the reduction argument.
    With the additional information on the domain and range of $\pi^*$, our problem can be reduced to the reconstruction of mapping as in \cite{wu2022settling}.
    Applying the lower bound in \cite[Theorem 4]{wu2022settling},
    for a fixed $\epsilon\in(0,1)$, when $m(\sqrt{p_{00}p_{11}}-\sqrt{p_{01}p_{10}})^2\le (1-\epsilon)\log m$, we have $\prob{\hat{\pi}\ne \pi^*}\ge 1-o(1)$ for any estimator $\hat\pi$. 
    Note that $(\sqrt{p_{00}p_{11}}-\sqrt{p_{01}p_{10}})^2 \asymp p^2(\gamma\wedge \gamma^2) \asymp (\rho^2)\wedge (\rho p).$
    Therefore, when 
    \begin{align}\label{eq: lower bound for exact recovery: 1}
        m\lesssim \frac{1}{p^2 (\gamma\wedge \gamma^2)}\log\pth{\frac{1}{p^2 (\gamma\wedge \gamma^2)}},
    \end{align}
    we have $\prob{\hat{\pi}\ne \pi^*}\ge 1-o(1)$.
    Applying Proposition~\ref{prop:lb-partial} with  $\delta = 1/2$ yields that,
    when 
    \begin{align}\label{eq: lower bound for exact ercovery: 2}
        m\lesssim \frac{\log n}{p^2 \phi(\gamma)},
    \end{align}
    we have $\prob{\hat{\pi}\ne \pi^*}\ge 1-c$ for $c\in (0,1)$.

    When $\frac{1}{p^2 (\gamma\wedge \gamma^2)}\asymp n$, by \eqref{eq: lower bound for exact recovery: 1}, exact recovery is impossible, even when $m = n$. 
    Next, we consider the regime that $\frac{1}{p^2 (\gamma\wedge \gamma^2)}\lesssim n$. When $\gamma \le 1$, we have  $p^2 (\gamma\wedge \gamma^2) = p^2 \gamma^2 \asymp p^2 \phi(\gamma)$, and thus \begin{align*}
        \frac{1}{p^2 (\gamma\wedge \gamma^2)}\log\pth{\frac{1}{p^2 (\gamma\wedge \gamma^2)}}  \lesssim \frac{\log n}{p^2 \phi(\gamma)}.
    \end{align*}
    When $\gamma \ge 1$, $\gamma\wedge \gamma^2 = \gamma$. By comparing \eqref{eq: lower bound for exact recovery: 1} and \eqref{eq: lower bound for exact ercovery: 2}, we derive that exact recovery is impossible if $m\lesssim \frac{\log n}{p^2 \phi(\gamma)} \vee \pth{\frac{1}{p^2 \gamma} \log\frac{1}{p^2 \gamma}}$.
\end{proof}

\begin{remark}\label{rmk:intuition-lwbd}
    {
    As discussed in Section~\ref{subsec: interpretation}, the estimation of $\pi^*$ consists of two different subproblems: 1) recovering the support sets $S^*$ and $T^*$; and 2) recovering the mapping $\pi^*$ between these two sets. 
    The first lower bound $\frac{\log n}{p^2 \phi(\gamma)}$ arises from the overlap complexity associated with the former (see Appendix~\ref{apd:support-recovery} for a proof), while the second lower bound $\frac{1}{p^2 \gamma} \log\frac{1}{p^2 \gamma}$ is due to the latter. 
    }
\end{remark}


\section{Discussion and future directions}
This paper introduces the partially correlated \ER model and the partially correlated Gaussian Wigner model, wherein a pair of induced subgraphs of a specific size is correlated. We investigate the optimal information-theoretic threshold for recovering both the latent correlated subgraphs and the hidden vertex correspondence in these new models.
In comparison with prior work on the correlated \ER and correlated Gaussian Wigner models, the additional challenge arises from the unknown location of the correlated subsets.
For a candidate mapping $\pi$ whose domain may include both correlated and ambient vertices, we extend the classical notion of functional digraph to formally describe the correlation structure among the edges. 
From the correlated functional digraph, we observe that the independent components consist of cycles and paths.
The graphical representation may be of independent interest for general models.


There are many  problems to be further investigated under our proposed models: 
\begin{itemize}
    \item \emph{Refined results.} The results in the paper could be further refined in various ways, such as deriving the sharp constants and characterizing the optimal scaling in terms of the fraction $\delta$ in partial recovery.
    \item \emph{Efficient algorithms.} It is of interest to investigate the polynomial-time algorithms and identify the  computational hardness under our model.
    More efficient algorithms are also desirable when the signal is stronger. 
    \item \emph{Graph sampling.} One motivation of the paper stems from graph sampling as discussed in
Section~\ref{sec:related}. The sampled subgraphs are partially correlated, where the size of correlated
subsets is a random variable depending on the sampling methods. Thus, it is natural to ask about
the sample size needed for reliable recovery.
\item \emph{Correlation test.} The correlation test problem under our model is also highly relevant.  It is interesting to find out whether the detection problem is strictly easier
than recovery, both in terms of the information thresholds and algorithmic developments.
\end{itemize}


\appendix
\section{\hd{Support Recovery Problem}}\label{apd:support-recovery}

{

We prove the impossibility results for exact recovery of $(S^*,T^*)$ by Fano's method. 
On the one hand, since $\pi^*$ contains the information of $S^*$ and $T^*$,
the mutual information can be upper bounded by \begin{align*}
    I(S^*,T^*;G_1,G_2)\le I(\pi^*;G_1,G_2)\le 25\binom{m}{2}p^2 \phi(\gamma).
\end{align*}
On the other hand, for exact recovery, $\maM$ is simply the set of all $(S,T)$, and thus $|\maM| = \binom{n}{m}^2$.
Applying the Fano's inequality, we obtain that \begin{align*}
    \prob{(\hat{S},\hat{T})\ne (S^*,T^*)}\ge 1-\frac{I(S^*,T^*;G_1,G_2)+\log 2}{\log |\maM|}\ge 1- \frac{25 \binom{m}{2}p^2 \phi(\gamma)+\log 2}{2\log\binom{n}{m}}.
\end{align*}
Since $\binom{n}{m}\ge \pth{\frac{n}{m}}^m$, the impossibility results follow if \begin{align}\label{eq:sufficient-support-recovery}
    25 m p^2 \phi(\gamma)\le 4c(\log n-\log m).
\end{align}

Consider the regime that $\frac{\log n}{p^2 \phi(\gamma)} \ge  13\pth{\frac{1}{p^2(\gamma\wedge\gamma^2)}\log\frac{1}{p^2(\gamma\wedge\gamma^2)}}$.
We first show $\frac{\log n}{p^2 \phi(\gamma)} = O(n^{1-\delta_0})$ for some constant $\delta_0>0$. Indeed, if $\frac{\log n}{p^2 \phi(\gamma)}=n^{1-o(1)}$, we have $p^2 \phi(\gamma) = n^{-1+o(1)}$, and thus $2\log(\frac{1}{p^2 (\gamma\wedge \gamma^2)})\ge \log n$.
When $\gamma\ge 1$, we have $\phi(\gamma)\ge \frac{\gamma}{6}$. When $0<\gamma<1$, we have $\phi(\gamma)\ge \frac{\gamma^2}{6}$.
Therefore, $p^2 \phi(\gamma)\ge \frac{1}{6}p^2 (\gamma\wedge \gamma^2)$. We obtain that $\frac{\log n}{p^2 \phi(\gamma)}\le  \frac{12\log(\frac{1}{p^2 (\gamma\wedge \gamma^2)})}{p^2 (\gamma\wedge \gamma^2)}$, which is contradictory with $\frac{\log n}{p^2 \phi(\gamma)} \ge 13\pth{\frac{1}{p^2(\gamma\wedge\gamma^2)}\log\pth{\frac{1}{p^2(\gamma\wedge\gamma^2)}}}$. 
Consequently, $\frac{\log n}{p^2 \phi(\gamma)} = O(n^{1-\delta_0})$ for some $\delta_0>0$. We note that $m\le \frac{4c\delta_0 \log n}{25 p^2 \phi(\gamma)}$ suffices for~\eqref{eq:sufficient-support-recovery}, and thus proves the impossibility results. 
}

\section{Proof of Propositions}

\subsection{Proof of Proposition~\ref{prop: Upper bound for partial recovery}}\label{subapd:Upper bound for partial recovery}
For any $k\ge (1-\delta)m$, let $\tau_k = N_k p_{11}(1-\eta)$ with \begin{align*}
    \eta = \sqrt{\frac{8h(\frac{k}{m})}{kp_{11}}}\cdot \indc{k\le m-1}     +\sqrt{\frac{\log n}{kmp_{11}}}\cdot \indc{k=m},
\end{align*}
where $h(x)\triangleq -x\log x-(1-x)\log(1-x)$ is the binary entropy function and $p_{11}=p^2(1+\gamma)$. Let $c_1(\delta) = 100\vee \frac{200 h(1-\delta)}{1-\delta}$.
We first show that $\eta \le \frac{\gamma}{4(1+\gamma)}<1$ under $m\ge \pth{100\vee \frac{200 h(1-\delta)}{1-\delta}} \frac{\log n}{p^2 \phi(\gamma)}$. Since $\phi(\gamma) = (1+\gamma)\log (1+\gamma) - \gamma\le \gamma^2$ by $\log(1+\gamma)\le\gamma$, we obtain that $p^2 \phi(\gamma)\le p^2 \gamma^2 = \rho^2(1-p)^2\le 1$ and $m \ge c_1(\delta) \log n$. We note that
\begin{align*}
    \eta &\overset{\mathrm{(a)}}{\le} \sqrt{\frac{8 h(1-\delta)}{(1-\delta) mp_{11}}}\cdot \indc{k\le m-1}+\sqrt{\frac{\log n}{m^2 p_{11}}}\cdot \indc{k=m}\\
    &\overset{\mathrm{(b)}}{\le} \pth{\sqrt{\frac{8h(1-\delta)}{1-\delta}}\vee \frac{1}{\sqrt{c_1(\delta)}}} \frac{1}{\sqrt{m p_{11}}}\\
    &\overset{\mathrm{(c)}}{\le}  \pth{\sqrt{\frac{8h(1-\delta)}{(1-\delta) c_1(\delta)}}\vee \frac{1}{c_1(\delta)}} \sqrt{\frac{\log(1+\gamma)-\gamma/(1+\gamma)}{\log n}}\overset{\mathrm{(d)}}{\le}    \frac{1}{5} \sqrt{\frac{\log(1+\gamma)-\gamma/(1+\gamma)}{\log n}},
\end{align*}
where $\mathrm{(a)}$ uses the fact that $\frac{h(k/m)}{k/m}\le \frac{h(1-\delta)}{1-\delta}$ since $\frac{h(x)}{x}$ decreases in $(0,1)$ and $k\ge (1-\delta) m$; $\mathrm{(b)}$ follows from $m\ge c_1(\delta) \log n$; $\mathrm{(c)}$ is because $m p_{11} = mp^2(1+\gamma)\ge \frac{c_1(\delta) \log n}{\phi(\gamma)/(1+\gamma)} = \frac{c_1(\delta)\log n}{\log(1+\gamma) - \gamma/(1+\gamma)}$; $\mathrm{(d)}$ follows from $c_1(\delta) = 100\vee \frac{200 h(1-\delta)}{1-\delta}$. Recall the assumption stated in Section~\ref{subsec: main results}, where it's asserted that $p\ge n^{-1}$ in the \ER model, thereby implying $\log(1+\gamma)\le \log n$ and thus $\eta\le \frac{1}{5}$. 
When $\gamma >10$, $\eta \le \frac{1}{5}\le \frac{\gamma }{4(1+\gamma)}$.
When $\gamma \le 10$,
since $\log(1+x)-\frac{x}{1+x}-x^2\le 0$ for any $x>0$, we obtain $\eta \le \frac{1}{5}\sqrt{\frac{\log(1+\gamma)-\gamma/(1+\gamma)}{\log n}}\le \frac{\gamma}{5\sqrt{\log n}}\le \frac{\gamma}{4(1+\gamma)}$ when $n$ sufficiently large. Therefore, we obtain $\eta\le \frac{\gamma}{4(1+\gamma)}<1$.

We then upper bound the bad event of signal by Chernoff bound. By \cite[Lemma 17.5.1]{cover2006elements}, we have $\binom{m}{k}\le \exp\qth{mh(k/m)}$. When $k\le m-1$, it follows from~\eqref{eq:bad-signal-total} and~\eqref{eq:weak-signal} that 
\begin{align}
    \nonumber \prob{ \bigcup_{\pi\in\calT_{k}} \sth{ \en(\calH_{\pi^*}^f) - \en(\calH_{\pi^*}^f[F_\pi])<\tau_k } }&\le \binom{m}{k} \exp\pth{-\frac{N_k p_{11}\eta^2}{2}}\\\nonumber &\le \exp\pth{mh\pth{\frac{k}{m}}-\frac{N_k p_{11}\eta^2}{2}}\\&\le\exp\pth{-mh\pth{\frac{k}{m}}}\label{eq:ER-partial-signal-m-1},
\end{align}
where the last inequality follows from $\frac{N_k p_{11}\eta^2}{2} = 2m\pth{2-\frac{k+1}{m}}h\pth{\frac{k}{m}}\ge 2mh\pth{\frac{k}{m}}$ when $k\le m-1$.
When $k=m$, $N_k = \frac{mk}{2}\pth{2-\frac{k+1}{m}}\ge \frac{mk}{3}$. Then, it follows from~\eqref{eq:bad-signal-total} and~\eqref{eq:weak-signal} that
\begin{align}
    \prob{ \bigcup_{\pi\in\calT_{k}} \sth{ \en(\calH_{\pi^*}^f) - \en(\calH_{\pi^*}^f[F_\pi])<\tau_k } }\le \binom{m}{k} \exp\pth{-\frac{N_k p_{11}\eta^2}{2}}\le n^{-1/6}\label{eq:ER-partial-signal-m},
\end{align}
where the last inequality follows from $\binom{m}{k} =1$ and $\frac{N_k p_{11} \eta^2}{2} = \frac{\log n}{2} \frac{N_k}{mk}\ge \frac{\log n}{6}$ when $k=m$.

Next, we upper bound the bad event of noise. Since $\eta\le \frac{\gamma}{4(1+\gamma)}$, it follows from Lemma~\ref{lem: general case for eta} that $\frac{\tau_k}{|\maE_\pi|p^2} = (1+\gamma)(1-\eta)>1$, where $p_{11} = p^2(1+\gamma)$ and $|\maE_\pi| = N_k$. Since $e^{\frac{k\gamma}{4(2+\gamma)}}\le e^{\frac{k}{4}}$ and $k! \ge \pth{\frac{k}{e}}^k$, it follows from~\eqref{eq:strong-noise-union-bound} and Lemma~\ref{lem: strong noise event} that 
\begin{align}
     \nonumber\prob{ \bigcup_{\pi\in\calT_{k}} \sth{ \en(\calH_{\pi}^f) - \en(\calH_{\pi}^f[F_\pi])\ge \tau_k } } &\le n^{3k} \exp\pth{-\frac{\tau_k}{2}\log\pth{\frac{\tau_k}{|\maE_\pi| p^2}}+\frac{\tau_k}{2}-\frac{|\maE_\pi| p^2}{2}}\\\nonumber
     &= n^{3k}\exp\pth{-\frac{N_k p^2}{2} \phi\pth{(1+\gamma)(1-\eta)-1}}\\&\overset{\mathrm{(a)}}{\le} n^{3k}\exp\pth{-\frac{N_kp^2\phi(\gamma)}{8}} \overset{\mathrm{(b)}}{\le} n^{-k}\label{eq:ER-partial-noise},
 \end{align}
where $\mathrm{(a)}$ applies \eqref{eq: bound for Hgamma} in Lemma \ref{lem: general case for eta}; $\mathrm{(b)}$ is because $N_k = \frac{mk}{2}\pth{2-\frac{k+1}{m}}\ge \frac{mk}{3}$ and $\frac{mp^2 \phi(\gamma)}{24}\ge 4\log n$.

Finally, we upper bound the error probability $\prob{\overlap(\hat{\pi},\pi^*)<\delta}$. By summing~\eqref{eq:ER-partial-signal-m-1} over $(1-\delta)m\le k\le m-1$ and summing~\eqref{eq:ER-partial-noise} over $(1-\delta)m\le k\le m$ and~\eqref{eq:ER-partial-signal-m}, we obtain that 
\begin{align*}
    &~\prob{\overlap(\hat{\pi},\pi^*)<\delta}\\=&~\sum_{k=(1-\delta) m}^m\prob{d(\pi^*,\hat{\pi}) = k}\\
    \le&~ \sum_{k=(1-\delta) m}^m \sth{\prob{ \bigcup_{\pi\in\calT_{k}} \sth{ \en(\calH_{\pi^*}^f) - \en(\calH_{\pi^*}^f[F_\pi])<\tau_k } }+\prob{ \bigcup_{\pi\in\calT_{k}} \sth{ \en(\calH_{\pi}^f) - \en(\calH_{\pi}^f[F_\pi])\ge \tau_k } }}\\
    \le&~ n^{-1/6}+\sum_{k=(1-\delta) m}^{m-1}\exp\qth{-mh\pth{\frac{k}{m}}} +\sum_{k=(1-\delta) m}^m n^{-k}\\
\le &~n^{-1/6}+\sum_{k=(1-\delta) m}^{m-1}\exp\qth{-mh\pth{\frac{k}{m}}}+\frac{n^{-(1-\delta) m}}{1-n^{-1}}.
\end{align*}
Since $m\ge c_1(\delta) \log n$ and $c_1(\delta)\ge 100$, it follows from Lemma~\ref{lem: sum over delta m} that $\sum_{k=(1-\delta)  m}^m \prob{d(\pi^*,\hat{\pi}) = k}\le \frac{4\log m+2}{m}.$  

\subsection{Proof of Proposition~\ref{prop: Upper bound for exact recovery}}\label{subapd:Upper bound for exact recovery}
Let $\tau_k = N_k p_{11}(1-\eta)$ with $\eta = \frac{\gamma}{4(1+\gamma)}<1$. By \eqref{eq:bad-signal-total} and \eqref{eq:weak-signal}  and applying $\binom{m}{k}\le m^k$, we get
    \begin{align}
        \prob{ \bigcup_{\pi\in\calT_{k}} \sth{ \en(\calH_{\pi^*}^f) - \en(\calH_{\pi^*}^f[F_\pi])<\tau_k } }\le m^k \exp\pth{-\frac{N_k p_{11}\eta^2}{2}}.\label{eq:ER-exact-signal}
    \end{align}
Since $p_{11}=p^2(1+\gamma)$ and $|\maE_\pi|=N_k$, we have $\frac{\tau_k}{|\maE_\pi|p^2} = (1+\gamma)(1-\eta)>1$ by Lemma \ref{lem: general case for eta}. 
Since $e^{\frac{k\gamma}{4(2+\gamma)}}\le e^{\frac{k}{4}}$ and $k! \ge \pth{\frac{k}{e}}^k$, we obtain $\frac{1}{k!^2} e^{\frac{k\gamma}{4(2+\gamma)}}\le 1$. Then it follows from \eqref{eq:strong-noise-union-bound} and Lemma~\ref{lem: strong noise event} that 
\begin{align}
\prob{ \bigcup_{\pi\in\calT_{k}} \sth{ \en(\calH_{\pi}^f) - \en(\calH_{\pi}^f[F_\pi])\ge \tau_k } }
&\le n^{3k}\exp\pth{-\frac{\tau_k}{2}\log\pth{\frac{\tau_k}{|\maE_\pi|p^2}}+\frac{\tau_k}{2}-\frac{|\maE_\pi|p^2}{2}}\nonumber\\
&= n^{3k}\exp\pth{-\frac{N_k p^2}{2} \phi\pth{(1+\gamma)(1-\eta)-1}}\nonumber\\
&\le n^{3k}\exp\pth{-\frac{N_kp^2}{8}\phi(\gamma)},\label{eq:ER-exact-noise}
\end{align}
where the last inequality applies \eqref{eq: bound for Hgamma} in Lemma \ref{lem: general case for eta}.
By summing~\eqref{eq:ER-exact-signal} and~\eqref{eq:ER-exact-noise} over $k\ge 1$ and applying $N_k\ge \frac{km}{3}$, we obtain that 
\begin{align}
    \prob{\hat{\pi}\neq \pi^*}=&~\sum_{k=1}^m\prob{d(\pi^*,\hat{\pi}) = k}\nonumber\\
    \le&~ \sum_{k=1}^m \sth{\prob{ \bigcup_{\pi\in\calT_{k}} \sth{ \en(\calH_{\pi^*}^f) - \en(\calH_{\pi^*}^f[F_\pi])<\tau_k } }+\prob{ \bigcup_{\pi\in\calT_{k}} \sth{ \en(\calH_{\pi}^f) - \en(\calH_{\pi}^f[F_\pi])\ge \tau_k } }}\nonumber\\
    \le&~ \sum_{k=1}^m \sth{\qth{m\exp\pth{-\frac{mp_{11}\eta^2}{6}}}^k+\qth{n^3\exp\pth{-\frac{mp^2 \phi(\gamma)}{24}}}^k}.\label{eq:ER-exact-total}
\end{align}

Let $C_1=3000$. It remains to upper bound~\eqref{eq:ER-exact-total} under $m\ge 3000\pth{\frac{\log(1/p^2\gamma)}{p^2 \gamma}\vee \frac{\log n}{p^2 \phi(\gamma)}}$. 
Since $m\ge \frac{3000 \log n}{p^2 \phi(\gamma)}$, we have
\begin{equation}
\label{eq:ER-exact-n-term}
n^3\exp\pth{-\frac{mp^2 \phi(\gamma)}{24}}
\le \frac{1}{n}.
\end{equation}
With a weak signal that $\gamma\le 1$, we have $\phi(\gamma)\le \frac{\gamma^2}{1+\gamma}$, and thus $m p_{11} \eta^2 =  \frac{mp^2\gamma^2}{16(1+\gamma)}\ge \frac{mp^2 \phi(\gamma)}{16}$. Since $m\ge \frac{3000 \log n}{p^2 \phi(\gamma)}$, we have 
\begin{equation}
    \label{eq:ER-exact-m-term}
    m\exp\pth{-\frac{mp_{11}\eta^2}{6}}\le \frac{1}{m}.
\end{equation}
With a strong signal that $\gamma\ge 1$, we have $\frac{\gamma}{1+\gamma}\ge \frac{1}{2}$ and thus $m p_{11} \eta^2 =  \frac{mp^2\gamma^2}{16(1+\gamma)}\ge \frac{mp^2 \gamma}{32}$. 
Since $m\ge 3000\pth{\frac{1}{p^2 \gamma}\log \frac{1}{p^2 \gamma}}$, we have that \begin{align*}
    \frac{\log m}{m}&\overset{\mathrm{(a)}}{\le} \frac{{\log 3000+\log\pth{\frac{1}{p^2\gamma}}+\log\pth{\log\pth{\frac{1}{p^2\gamma}}}}}{3000 \pth{\frac{1}{p^2\gamma} \log\pth{\frac{1}{p^2\gamma}}}}\overset{\mathrm{(b)}}{\le} \frac{\pth{2+\frac{\log 3000}{\log 4}}\log\pth{\frac{1}{p^2\gamma}}}{3000 \pth{\frac{1}{p^2\gamma} \log\pth{\frac{1}{p^2\gamma}}}}\le \frac{p^2\gamma}{384},
\end{align*} 
where $\mathrm{(a)}$ is because $\frac{\log x}{x}$ decreases on $[e,\infty]$ and $m\ge 3000$; $\mathrm{(b)}$ is because $\log\pth{\frac{1}{p^2\gamma}} = \log\pth{\frac{1}{\rho p(1-p)}}\ge \log 4$ and $\log\pth{\log\pth{\frac{1}{p^2\gamma}}}\le \log\pth{\frac{1}{p^2\gamma}}$.
Then \eqref{eq:ER-exact-m-term} holds since $m\exp\pth{-\frac{mp_{11}\eta^2}{6}}\le m\exp\pth{-\frac{mp^2\gamma}{192}}\le m\exp(-2\log m) = m^{-1}$.
We conclude the proof by applying \eqref{eq:ER-exact-total} with \eqref{eq:ER-exact-n-term} and \eqref{eq:ER-exact-m-term}.

\subsection{Proof of Proposition~\ref{prop: Upper bound for exact recovery-Gaussian}}\label{subapd:Upper bound for exact recovery-Gaussian}

Let $\tau_k = \rho N_k - c_0\pth{\sqrt{N_k\log(1/\theta)}\vee\log(1/\theta)}$ with $\theta =  \exp\pth{-2k\log m}$, where $c_0$ is the universal constant in Lemma~\ref{lem:Hanson-Wright}. Let $C_3 = 25c_0^2\vee 1100$. 
We first verify that  $\tau_k\ge \frac{1}{2} \rho N_k$ under condition $m\rho^2 \ge C_3\log n$. We show that $c_0\sqrt{N_k\log(1/\theta)}\le \frac{\rho N_k}{2}$ and $c_0\log(1/\theta)\le \frac{\rho N_k}{2}$, respectively. For the first term $c_0\sqrt{N_k \log(1/\theta)}$, we note that
\begin{align*}
    c_0 \sqrt{N_k \log (1/\theta)}&\overset{\mathrm{(a)}}{\le} \sqrt{\frac{2C_3}{25} N_k k\log n}= \sqrt{\frac{C_3\log n}{m}\cdot \frac{mk}{3}\cdot \frac{6N_k}{25}}\overset{\mathrm{(b)}}{\le} \sqrt{\rho^2 \cdot N_k \cdot \frac{6N_k}{25} }\le \frac{1}{2} \rho N_k,
\end{align*}
where $\mathrm(a)$ follows from $25c_0^2\le C_3$ and $\log(1/\theta) = 2k\log m\le 2k\log n$; $\mathrm{(b)}$ follows from $m\rho^2 \ge C_3\log n$ and $N_k = mk\pth{1-\frac{k+1}{2m}}\ge \frac{mk}{3}$. For the second term $c_0\log(1/\theta)$, we have that \begin{align*}
    c_0 \log(1/\theta)&\overset{\mathrm{(a)}}{\le} 6c_0 N_k \frac{\log m}{m}\overset{\mathrm{(b)}}{\le}\frac{6c_0}{\sqrt{C_3}} N_k\cdot  \sqrt{\frac{C_3\log n}{m}}\cdot \sqrt{\frac{\log m}{m}} \overset{\mathrm{(c)}}{\le} \frac{6c_0\sqrt{\log m}}{\sqrt{C_3 m}} \rho N_k\le \frac{1}{2}\rho N_k
\end{align*}
when $m$ sufficiently large, where $\mathrm{(a)}$ is because $\log(1/\theta) = 2k\log m$ and $N_k\ge \frac{mk}{3}$; $\mathrm{(b)}$ follows from $\log m\le \log n$; $\mathrm{(c)}$ uses $m\rho^2\ge C_3\log n$.
Therefore, we obtain $\tau_k\ge \frac{1}{2}\rho N_k$.

For the bad event of signal, by \eqref{eq:weak-signal-Gaussian} and Hanson-Wright inequality in Lemma~\ref{lem:Hanson-Wright} with $M_0 = I_{N_k}$ and applying $\binom{m}{k}\le m^k$, we obtain that \begin{align}\label{eq:Gaussian-sum-weak-1}
    \prob{ \bigcup_{\pi\in\calT_{k}} \sth{ \en(\calH^f_{\pi^*}) - \en(\calH^f_{\pi^*}[F_\pi])<\tau_k } } \le m^k\exp\pth{-2k\log m} = m^{-k}.
\end{align}
For the bad event of noise, it follows from~\eqref{eq:strong-noise-union-bound} and Lemma~\ref{lem: strong noise event-Gaussian} that \begin{align}\label{eq:Gaussian-sum-weak-2}
    \nonumber &\prob{ \bigcup_{\pi\in\calT_{k}} \sth{ \en(\calH^f_{\pi}) - \en(\calH^f_{\pi}[F_\pi])\ge \tau_k } } \\ \nonumber \le&~ n^{3k} \exp\pth{-\frac{\rho\tau_k}{6}+\frac{\rho^2|\maE_\pi|}{14}+\frac{\log 5}{8}k}
    \\\overset{\mathrm{(a)}}{\le}&~ n^{3k} \exp\pth{-\frac{\rho^2 |\maE_\pi|}{84}+\frac{\log 5}{8}k}
    \overset{\mathrm{(b)}}{\le} n^{3k}  \exp\pth{-4k\log n} = n^{-k},
\end{align}
where $\mathrm{(a)}$ is because $\tau_k\ge \frac{1}{2}\rho N_k$; $\mathrm{(b)}$ follows from $|\maE_\pi| = N_k\ge \frac{mk}{3}$ and $m\rho^2 \ge 1100\log n$.
By summing~\eqref{eq:Gaussian-sum-weak-1} and~\eqref{eq:Gaussian-sum-weak-2} over $1\le k\le m$, we obtain that  \begin{align*}
    \prob{\hat{\pi}\neq \pi^*}=&~\sum_{k=1}^m\prob{d(\pi^*,\hat{\pi}) = k}\\
    \le&~ \sum_{k=1}^m \prob{ \bigcup_{\pi\in\calT_{k}} \sth{ \en(\calH^f_{\pi^*}) - \en(\calH^f_{\pi^*}[F_\pi])<\tau_k } }+\prob{ \bigcup_{\pi\in\calT_{k}} \sth{ \en(\calH^f_{\pi}) - \en(\calH^f_{\pi}[F_\pi])\ge \tau_k } }\\\le&~ \sum_{k=1}^m \pth{m^{-k}+n^{-k}}\le \frac{m^{-1}}{1-m^{-1}}+\frac{n^{-1}}{1-n^{-1}} = \frac{1}{m-1}+\frac{1}{n-1}.
\end{align*}

\subsection{Proof of Proposition~\ref{prop: upbd_gaussian_general}}\label{subapd:upbd_gaussian_general}

    Let $\tau_k = \frac{\rho-1}{3}\log\pth{\frac{1}{1-\rho}} N_k$ and $C_4 =100$. We first upper bound the bad event of signal.
    Recall~\eqref{eq:weak-signal-Gaussian-2}, since $A_i-B_i \sim \maN(0,2-2\rho)$ for $(A_i,B_i)\sim \gaussianrho$, we obtain that $\sum_{i=1}^{N_k} \frac{(A_i-B_i)^2}{2(1-\rho)}$ follows the chi-squared distribution with $N_k$ degrees of freedom. Then, it follows from~\eqref{eq:weak-signal-Gaussian-2} and $\binom{m}{k}\le m^k$ that \begin{align}
        \nonumber &~\prob{ \bigcup_{\pi\in\calT_{k}} \sth{ \en(\calH^f_{\pi^*}) - \en(\calH^f_{\pi^*}[F_\pi])<\tau_k } }\\\nonumber\le&~ m^k \prob{\sum_{i=1}^{N_k}-\frac{1}{2}(A_i-B_i)^2<\frac{\rho-1}{3}\log\pth{\frac{1}{1-\rho}}N_k}\\\nonumber=&~ m^k \prob{\sum_{i=1}^{N_k} \frac{(A_i-B_i)^2}{2(1-\rho)} >\frac{1}{3}\log\pth{\frac{1}{1-\rho}}N_k}\\\nonumber \overset{\mathrm{(a)}}{\le}&~ m^k\exp\sth{-\frac{N_k}{2}\qth{\frac{1}{3}\log\pth{\frac{1}{1-\rho}}-1-\log\pth{\frac{1}{3}\log\pth{\frac{1}{1-\rho}}}}}\\\overset{\mathrm{(b)}}{\le}&~ m^k\exp\pth{-\frac{mk\log\pth{1/(1-\rho)}}{48}}\overset{\mathrm{(c)}}{\le} n^{-k}\label{eq:Gaussian-strong-sum-1},
    \end{align}
    where $\mathrm{(a)}$ uses \eqref{eq:concentration_for_chisquare} in Lemma~\ref{lem:chisquare}; $\mathrm{(b)}$ is because $ N_k=\binom{m}{2}-\binom{m-k}{2}=mk\pth{1-\frac{k+1}{2m}}\ge \frac{km}{3}$ and $x-1-\log x\ge \frac{3x}{8}$ for $x = \frac{1}{3}\log\pth{\frac{1}{1-\rho}}\ge 4$; $\mathrm{(c)}$ is because $m \log(1/(1-\rho))\ge 100\log n$ and $n\ge m$.   
    
    We then upper bound the bad event of noise. It follows from~\eqref{eq:strong-noise-union-bound} and Lemma~\ref{lem:strong noise event-Gaussian-2} that \begin{align}
        \nonumber &~\prob{ \bigcup_{\pi\in\calT_{k}} \sth{ \en(\calH^f_{\pi}) - \en(\calH^f_{\pi}[F_\pi])\ge \tau_k } }\\\nonumber\le&~ n^{3k}   \exp\pth{ -\frac{\rho\tau_k}{4(1-\rho)}- \frac{|\maE_\pi|}{4}\log\pth{\frac{1}{1-\rho}}+\frac{k}{8}\log\pth{\frac{1}{1-\rho}}} \\ 
        \nonumber\overset{\mathrm{(a)}}{\le}&~ n^{3k} \exp\pth{-\frac{|\maE_\pi|}{6}\log\pth{\frac{1}{1-\rho}}+\frac{k}{8}\log\pth{\frac{1}{1-\rho}}}\\\nonumber\overset{\mathrm{(b)}}{\le}&~ \exp\pth{3k\log n-\frac{mk}{18}\log\pth{\frac{1}{1-\rho}}+\frac{k}{8}\log\pth{\frac{1 }{1-\rho}}}\\=&~ \exp\qth{\pth{3k\log n-\frac{mk}{24}\log\frac{1}{1-\rho}}+\pth{\frac{k}{8}\log\frac{1}{1-\rho}-\frac{mk}{72}\log\frac{1}{1-\rho}}}\overset{\mathrm{(c)}}{\le} n^{-k}\label{eq:Gaussian-strong-sum-2},
    \end{align}
    where $\mathrm{(a)}$ is because $|\maE_\pi|=N_k$ and $-\frac{\rho\tau_k}{4(1-\rho)}- \frac{|\maE_\pi|}{4}\log\pth{\frac{1}{1-\rho}} = {\frac{(\rho-3)|\maE_\pi|}{12}}\log\pth{\frac{1}{1-\rho}}\le -\frac{|\maE_\pi|}{6}\log\pth{\frac{1}{1-\rho}}$; $\mathrm{(b)}$ follows from $|\maE_\pi| = N_k = \frac{mk}{2}\pth{1-\frac{k+1}{2m}}\ge \frac{mk}{3}$; $\mathrm{(c)}$ is because $m\ge 100\pth{\frac{\log n}{\log(1/(1-\rho))}\vee 1}$ implies $3k\log n-\frac{mk}{24}\log\frac{1}{1-\rho}\le -k\log n$ and $\frac{k}{8}\log\frac{1}{1-\rho}-\frac{mk}{72}\log\frac{1}{1-\rho}\le 0$.
    By summing over~\eqref{eq:Gaussian-strong-sum-1} and~\eqref{eq:Gaussian-strong-sum-2}, we obtain that \begin{align*}
    \prob{\hat{\pi}\neq \pi^*} = &~\sum_{k=1}^m\prob{d(\pi^*,\hat{\pi}) = k}\\
    \le&~ \sum_{k=1}^m \prob{ \bigcup_{\pi\in\calT_{k}} \sth{ \en(\calH^f_{\pi^*}) - \en(\calH^f_{\pi^*}[F_\pi])<\tau_k } }+\prob{ \bigcup_{\pi\in\calT_{k}} \sth{ \en(\calH^f_{\pi}) - \en(\calH^f_{\pi}[F_\pi])\ge \tau_k } }\\\le&~ \sum_{k=1}^m \pth{n^{-k}+n^{-k}}\le  \frac{2n^{-1}}{1-n^{-1}} = \frac{2}{n-1}.
\end{align*}

\section{Proof of Lemmas}

\subsection{Proof of Lemma \ref{lem: upper bound for kappa}}\label{apdsec: upper bound for kappa}
\paragraph{\ER Model} We consider $f(x,y) = xy$ in the \ER model. The lower-order cumulants can be calculated directly: \begin{align}
    \kappa^{\sfC}_1(t) & = \log(1+p_{11}(e^t-1)), \quad \kappa^{\sfP}_1(t) = \log(1+p^2(e^t-1)),\label{eq:cumulants-C1}\\
    \kappa^{\sfC}_2(t) & = \log(1+2p^2(e^t-1)+p_{11}^2(e^t-1)^2). \label{eq:cumulants-C2}
\end{align}
We first evaluate the moment generating function for paths. 
Consider a path $P$ of size $\ell$ denoted by $\langle e_1 e_2 \dots e_{\ell} \rangle$ as illustrated in Figure~\ref{fig:general_path}.
For each $i=1,\dots,\ell$, define $A_{i-1}\triangleq \beta_{e_i}(G_1)$ and $B_i\triangleq \beta_{\pi(e_i)}(G_2)$. Then $(A_i,B_i)\sim\multibern(p,p,\rho)$. 
By definition~\eqref{eq:calE-counter}, 
\[
\beta_P(\maH_\pi^f) = \sum_{i=1}^{\ell}\beta_{e_i}(G_1)\beta_{\pi(e_i)}(G_2)
= \sum_{i=1}^{\ell} A_{i-1}B_i.
\]
\begin{figure}[htb]
    \begin{center}
        \begin{tikzpicture}
            \draw (1,3) coordinate (A1);
            \draw (1,2) coordinate (A2);
            \draw (1,1) coordinate (A3);
            \draw (1,-0.5) coordinate (A4);
            \draw (3,3) coordinate (B1);
            \draw (3,2) coordinate (B2);
            \draw (3,1) coordinate (B3);
            \draw (3,-0.5) coordinate (B4);
            \draw[blue] (A1) -- node[above right,scale = 0.7]{} (B1) node[sloped, pos=0.5,scale=2]{\arrowIn};
            \draw[blue] (A2) -- (B2) node[sloped, pos=0.5,scale=2]{\arrowIn};
            \draw[blue] (A3) -- (B3) node[sloped, pos=0.5,scale=2]{\arrowIn};
            \draw[blue] (A4) -- (B4) node[sloped, pos=0.5,scale=2]{\arrowIn};
            \draw[red,dashed] (A2) -- node[above left,scale = 0.7] {$\pi^*$} (B1) node[sloped, pos=0.5,scale=2]{\arrowIn};
            \draw[red,dashed] (A3) -- (B2) node[sloped, pos=0.5,scale=2]{\arrowIn};
            \draw (A1) node[left,scale=0.7]{$e_1$} node[above=2pt,scale=0.7]{$A_0$} node{$\bullet$};
            \draw (A2) node[left,scale=0.7]{$e_2$} node[above=2pt,scale=0.7]{$A_1$} node{$\bullet$};
            \draw (A3) node[left,scale=0.7]{$e_3$} node[above=2pt,scale=0.7]{$A_2$} node{$\bullet$};
            \draw (A4) node[left,scale=0.7]{$e_\ell$} node[above=2pt,scale=0.7]{$A_{\ell-1}$} node{$\bullet$};
            \draw (B1) node[right,scale=0.7]{$\pi(e_1)$} node[above=2pt,scale=0.7]{$B_1$} node{$\bullet$};
            \draw (B2) node[right,scale=0.7]{$\pi(e_2)$} node[above=2pt,scale=0.7]{$B_2$} node{$\bullet$};
            \draw (B3) node[right,scale=0.7]{$\pi(e_3)$} node[above=2pt,scale=0.7]{$B_3$} node{$\bullet$};
            \draw (B4) node[right,scale=0.7]{$\pi(e_\ell)$} node[above=2pt,scale=0.7]{$B_{\ell}$} node{$\bullet$};
            \draw (2,0.75) node[below] {$\vdots$};

            \draw[->,thick] (4.5,1.5) -- (5.5,1.5);
            
            \draw (6.5,1.5) coordinate(C1) node[below=3pt,scale=0.7]{$e_1$} node{$\bullet$};
            \draw (7.5,1.5) coordinate(C2) node[below=3pt,scale=0.7]{$e_2$} node[above=2pt,scale=0.7]{$\pi(e_1)$} node{$\bullet$};
            \draw (8.5,1.5) coordinate(C3) node[below=3pt,scale=0.7]{$e_3$} node[above=2pt,scale=0.7]{$\pi(e_2)$} node{$\bullet$};
            \draw (10,1.5) coordinate(C4) node[below=3pt,scale=0.7]{$e_{\ell}$} node[above=2pt,scale=0.7]{$\pi(e_{\ell-1})$} node{$\bullet$};
            \draw (11,1.5) coordinate(C5) node[below=3pt,scale=0.7]{} node[above=2pt,scale=0.7]{$\pi(e_\ell)$} node{$\bullet$};
            \draw[blue] (C1) -- (C2) node[sloped, pos=0.5,scale=2]{\arrowIn};
            \draw[blue] (C2) -- (C3) node[sloped, pos=0.5,scale=2]{\arrowIn};
            \path (C3) -- (C4) node[pos=0.5,scale=1]{$\dots$};
            \draw[blue] (C4) -- (C5) node[sloped, pos=0.5,scale=2]{\arrowIn};
        \end{tikzpicture}
    \end{center}
    \caption{Illustration of a path of size $\ell$.}
    \label{fig:general_path}
\end{figure}
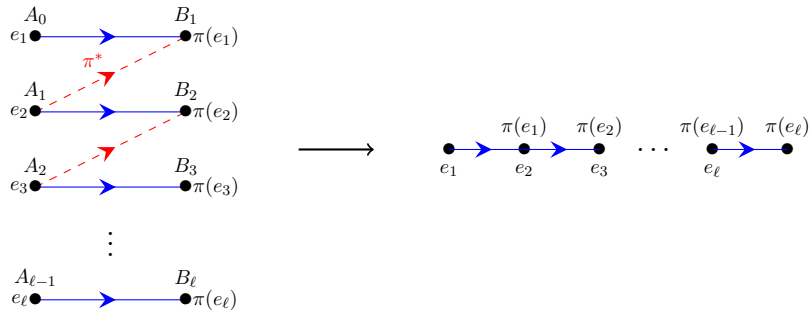
For notational simplicity, we introduce an auxiliary random variable $B_0$ that is correlated with $A_0$ such that $(A_0,B_0)\sim\multibern(p,p,\rho)$. 
Then
\begin{align}
    m_\ell
    \triangleq\expect{e^{t\beta_P\pth{\maH^f_\pi}}}
    &=\expect{ \expect{ \prod_{i=1}^\ell e^{t A_{i-1}B_i} | B_0\dots B_\ell} } 
    =\expect{  \prod_{i=1}^\ell \expect{ e^{t A_{i-1}B_i} | B_{i-1}, B_i}  }\nonumber \\
    &=\sum_{b_0,\dots,b_\ell\in\{0,1\} } \prod_{i=0}^\ell \prob{B_i=b_i} \prod_{i=1}^\ell  \expect{ e^{t A_{i-1}b_i} | B_{i-1}=b_{i-1} }. \label{eq:mgf-path-1}
\end{align}

Define $M(b_{i-1},b_i)\triangleq \prob{B_i=b_i} \expect{ e^{t A_{i-1}b_i} | B_{i-1}=b_{i-1} } $ for $b_{i-1},b_i\in\{0,1\}$ and a matrix 
\[
M\triangleq 
\begin{bmatrix}
    M(0,0) & M(0,1) \\
    M(1,0) & M(1,1)
\end{bmatrix}
= 
\begin{bmatrix}
    \bar p & (\bar p + p_{01} (e^t-1))p/\bar p \\
    \bar p & p+p_{11}(e^t-1)
\end{bmatrix},
\]
where $\bar p \triangleq 1-p$. Recall that $\prob{B_i=1}=p$. 
Then, we obtain that
\[
m_\ell
= \sum_{b_0,\dots,b_\ell\in\{0,1\} } \prob{B_0=b_0} M(b_0,b_1) \dots M(b_{\ell-1},b_\ell)
= [\bar p, p] M^\ell \begin{bmatrix}1\\ 1\end{bmatrix}.
\]
The trace and determinant of $M$ are given by
\[
T\triangleq \tr(M) = 1+p_{11}(e^t-1),
\quad 
D\triangleq \det(M) = \rho p \bar p (e^t-1) >0.
\]
Since $D<p_{11} (e^t-1)$, the discriminant is $T^2-4D>0$.
Hence, the matrix $M$ has two distinct eigenvalues denoted by $\lambda_1>\lambda_2$. Since $\lambda_1+\lambda_2=T>0$ and $\lambda_1\lambda_2 = D>0$, we have $\lambda_1>\lambda_2>0$, and the general term of $m_\ell$ is 
\begin{equation}
\label{eq:mgf-path}
m_\ell = \alpha_1 \lambda_1^\ell + \alpha_2 \lambda_2^\ell.
\end{equation}
The coefficients $\alpha_1$ and $\alpha_2$ can be determined via the first two terms $m_0=1$ and $m_1$. Then we get
\[
m_\ell = 
\pth{\frac{1}{2}+\frac{2m_1-T}{2\sqrt{T^2-4D}}} \lambda_1^\ell+\pth{\frac{1}{2}-\frac{2m_1-T}{2\sqrt{T^2-4D}}} \lambda_2^\ell.
\]
Furthermore, by plugging $m_1=1+p^2 (e^t-1)$, we get $T-m_1=D$ and thus $m_1(T-m_1)> D$, which is equivalent to $|2m_1-T|< \sqrt{T^2-4D}$. 
Therefore, both coefficients $\alpha_1,\alpha_2 \in (0,1)$.

The analysis for cycles follows from similar arguments.
Consider a cycle $C$ of size $\ell$ denoted by $[e_1\dots e_\ell]$ as illustrated in Figure~\ref{fig:general_circle}. 
For each $i=1,\dots,\ell$, define $A_{i-1}\triangleq \beta_{e_i}(G_1)$ and $B_i\triangleq \beta_{\pi(e_i)}(G_2)$. 
We also let $B_0=B_\ell$ for notational simplicity. 
Then $(A_i,B_i)\sim\multibern(p,p,\rho)$ for $i=0,\dots,\ell-1$. 
\begin{figure}[htb]
    \begin{center}
        \begin{tikzpicture}[global scale = 1]
            \draw (1,3) coordinate (A1);
            \draw (1,2) coordinate (A2);
            \draw (1,1) coordinate (A3);
            \draw (1,-0.5) coordinate (A4);
            \draw (3,3) coordinate (B1);
            \draw (3,2) coordinate (B2);
            \draw (3,1) coordinate (B3);
            \draw (3,-0.5) coordinate (B4);
            \draw[blue] (A1) -- node[above right,scale = 0.7] {} (B1) node[sloped, pos=0.5,scale=2]{\arrowIn};
            \draw[blue] (A2) -- (B2) node[sloped, pos=0.5,scale=2]{\arrowIn};
            \draw[blue] (A3) -- (B3) node[sloped, pos=0.5,scale=2]{\arrowIn};
            \draw[blue] (A4) -- (B4) node[sloped, pos=0.5,scale=2]{\arrowIn};
            \draw[red,dashed] (A2) -- node[above left,scale = 0.7] {$\pi^*$} (B1) node[sloped, pos=0.5,scale=2]{\arrowIn};
            \draw[red,dashed] (A3) -- (B2) node[sloped, pos=0.65,scale=2]{\arrowIn};
            \draw[red,dashed] (A1) -- (B4) node[sloped, pos=0.7,scale=2]{\arrowIn};
            \draw (A1) node[left,scale = 0.7] {$e_1$} node{$\bullet$};
            \draw (A2) node[left,scale = 0.7] {$e_2$} node{$\bullet$};
            \draw (A3) node[left,scale = 0.7] {$e_3$} node{$\bullet$};
            \draw (A4) node[left,scale = 0.7] {$e_\ell$} node{$\bullet$};
            \draw (B1) node[right,scale = 0.7] {$\pi(e_1) $} node{$\bullet$};
            \draw (B2) node[right,scale = 0.7] {$\pi(e_2) $} node{$\bullet$};
            \draw (B3) node[right,scale = 0.7] {$\pi(e_3) $} node{$\bullet$};
            \draw (B4) node[right,scale = 0.7] {$\pi(e_\ell)$} node{$\bullet$};
            \draw (2,0.75) node[below] {$\vdots$};

            \draw[->,thick] (4.5,1.5) -- (5.5,1.5);
            
            \draw (6.5,1.5) coordinate (C1);
            \draw (7.5,1.5) coordinate (C2);
            \draw (8.5,1.5) coordinate (C3);
            \draw (10,1.5) coordinate (C4);
            \draw (11,1.5) coordinate (C5);
            \draw[blue] (C1) -- (C2) node[sloped, pos=0.5,scale=2]{\arrowIn};
            \draw[blue] (C2) -- node[above,scale = 0.7] {$\pi$} (C3) node[sloped, pos=0.5,scale=2]{\arrowIn};
            \path (C3) -- (C4) node[pos=0.5,scale=1]{$\dots$};
            \draw[blue] (C4) -- (C5) node[sloped, pos=0.5,scale=2]{\arrowIn};
            \draw[blue] (C5) edge[bend right=30] node[sloped, pos=0.5,scale=2]{\arrowOut} (C1);
            \draw (C1) node[below,scale = 0.7] {$\begin{array}{c} e_1 \\ \pi(e_\ell) \end{array}$} node{$\bullet$};
            \draw (C2) node[below,scale = 0.7] {$\begin{array}{c} e_2 \\ \pi(e_1) \end{array}$} node{$\bullet$};
            \draw (C3) node[below,scale = 0.7] {$\begin{array}{c} e_3 \\ \pi(e_2) \end{array}$} node{$\bullet$};
            \draw (C4) node[below,scale = 0.7] {$\begin{array}{c} e_{\ell-1} \\ \pi(e_{\ell-2}) \end{array}$} node{$\bullet$};
            \draw (C5) node[below,scale = 0.7] {$\begin{array}{c} e_{\ell} \\ \pi(e_{\ell-1}) \end{array}$} node{$\bullet$};
        \end{tikzpicture}
    \end{center}
    \caption{Illustration of a cycle of size $\ell$.}
    \label{fig:general_circle}
\end{figure}
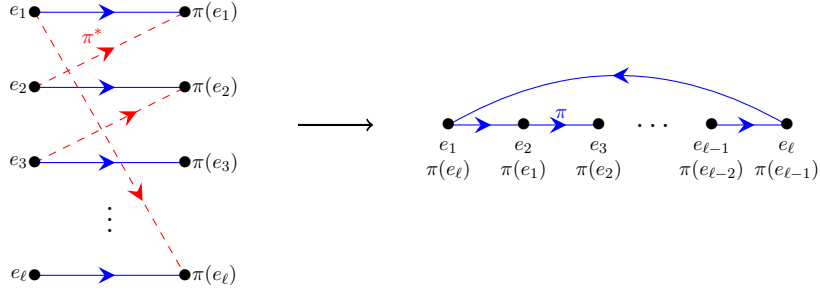
Following a similar argument as~\eqref{eq:mgf-path-1}, we have 
\begin{align*}
    \tilde m_\ell
    \triangleq\expect{e^{t\beta_C\pth{\maH^f_\pi}}}
    & =\sum_{b_1,\dots,b_\ell=b_0\in\{0,1\} } \prod_{i=1}^\ell \prob{B_i=b_i} \prod_{i=1}^\ell  \expect{ e^{t A_{i-1}b_i} | B_{i-1}=b_{i-1} } \\
    & = \sum_{b_1,\dots,b_\ell=b_0\in\{0,1\} } M(b_0,b_1) M(b_1,b_2) \dots M(b_{\ell-1},b_0).
\end{align*}
Applying the eigenvalue decomposition of $M$ again, we obtain that 
\begin{equation}
\label{eq:mgf-cycle}
\tilde m_\ell = \tr(M^\ell) = \lambda_1^\ell + \lambda_2^\ell.
\end{equation}

By definition, $\kappa^{\sfP}_\ell(t) = \log m_\ell$ and $\kappa^{\sfC}_\ell(t) = \log \tilde m_\ell$. 
To upper bound the cumulants, it suffices to consider $m_\ell$ and $\tilde m_\ell$.
In \eqref{eq:mgf-path}, we have $\alpha_1,\alpha_2\in(0,1)$ and $\lambda_1>\lambda_2>0$. 
By monotonicity, it follows that $m_\ell\le \tilde m_\ell$ and thus 
\[
\kappa^{\sfP}_\ell(t) \le \kappa^{\sfC}_\ell(t).
\]
For $x\in\reals^n$ and $\ell\ge 2$, we have $\Norm{x}_\ell \le \Norm{x}_2 \le \Norm{x}_1$.
It follows from the formula of $\tilde m_\ell$ in \eqref{eq:mgf-cycle} that $\tilde m_\ell^{1/\ell} \le \tilde m_2^{1/2} \le \tilde m_1$. Equivalently,
\[
\frac{1}{2}\kappa^{\sfC}_2(t) \le \kappa^{\sfC}_1(t),
\quad 
\kappa^{\sfC}_\ell(t)\le \frac{\ell}{2} \kappa^{\sfC}_2(t)\quad \forall~\ell\ge 2.
\]
The last inequality $2\kappa^{\sfP}_1(t) \le \kappa^{\sfC}_2(t)$ follows by comparing the explicit formula $\kappa^{\sfP}_1(t)=\log(1+p^2 (e^t-1))$ with $\kappa^{\sfC}_2(t)$ in~\eqref{eq:cumulants-C2} and using $p_{11}\ge p^2$.

Finally, since the summands over different connected components are independent, it follows that 
\begin{align*}
\log \expect{e^{t\beta_\maE(\maH_\pi^f)}}
& = \sum_{P\in\calP} \kappa_{|P|}^{\sfP}(t) + \sum_{C\in\calC} \kappa_{|C|}^{\sfC}(t) \\
& \le \sum_{P\in\calP} \frac{|P|}{2}\kappa_{2}^{\sfC}(t) + \sum_{C\in\calC: |C|\ge 2} \frac{|C|}{2}\kappa_{2}^{\sfC}(t) + \sum_{C\in\calC: |C|=1} \kappa_{1}^{\sfC}(t) \\
& = \frac{|\calE|}{2}\kappa_{2}^{\sfC}(t) + |\{C\in\calC: |C|=1\}| \pth{\kappa_{1}^{\sfC}(t) - \frac{1}{2}\kappa_{2}^{\sfC}(t) },
\end{align*}
where the last equality uses fact that $|\calE| = \sum_{P\in\calP} |P| + \sum_{C\in\calC} |C|$.

\begin{remark}
    We have two bounds for large $\ell$ in the  \ER model, namely $\kappa_\ell^\sfP(t)\le \kappa_\ell^\sfC(t)$ and $\kappa_\ell^\sfC(t)\le \frac{\ell}{2}\kappa_2^\sfC(t)$. 
 For the first bound, we apply $\frac{1}{\ell} \log\pth{\alpha_1 \lambda_1^\ell+\alpha_2 \lambda_2^\ell}\le \frac{1}{\ell} \log(\lambda_1^\ell+\lambda_2^\ell)$, where $0<\alpha_2<\alpha_1<1,\alpha_1+\alpha_2 = 1$ and $\lambda_1>\lambda_2>0$.
Consequently, 
$\lambda_1 - \frac{\log 2}{\ell} 
\le \frac{1}{\ell}\kappa_\ell^\sfP(t)
\le \frac{1}{\ell}\kappa_\ell^\sfC(t)
\le \lambda_1 + \frac{\log 2}{\ell}.$
Hence, the first bound is essentially tight for large $\ell$.
 The second bound, previously used in \cite{wu2022settling}, applies the inequality  $\Vert x\Vert_\ell \le \Vert x\Vert_2$, which becomes less tight as $\ell$ increases. 
Nevertheless, it suffices for our analysis as the probability of long cycles occurring is relatively small.
\end{remark}

\paragraph{Gaussian Wigner Model, Part 1}
In this part, we focus on the Gaussian Wigner model with $f(x,y) = xy$. The lower-order cumulants can be calculated directly:
\begin{align}
    \kappa_1^\sfC(t) &= -\frac{1}{2}\log(1-2t\rho-t^2(1-\rho^2)),
    \quad 
    \kappa_1^\sfP(t) =-\frac{1}{2}\log\pth{1-t^2}, \label{eq:cumulants-C1-Gaussian}\\
    \kappa_2^\sfC(t) &= -\frac{1}{2}\log\pth{1-2t^2(1+\rho^2)+t^4(1-\rho^2)^2}\label{eq:cumulants-C2-Gaussian}.
\end{align}
We first evaluate the moment generating function for paths. Consider a path $P$ of size $\ell$ denoted by $\langle e_1 e_2 \dots e_{\ell} \rangle$ as illustrated in Figure~\ref{fig:general_path}.
For each $i=1,\dots,\ell$, define $A_{i-1}\triangleq \beta_{e_i}(G_1)$ and $B_i\triangleq \beta_{\pi(e_i)}(G_2)$. Then $(A_i,B_i)\sim\mathcal{N}\pth{\begin{pmatrix} 0 \\ 0 \end{pmatrix}, \begin{pmatrix} 1 & \rho \\ \rho & 1 \end{pmatrix}}$. 
By definition~\eqref{eq:calE-counter}, 
\[
\beta_P(\maH^f_\pi) = \sum_{i=1}^{\ell}\beta_{e_i}(G_1)\beta_{\pi(e_i)}(G_2)
= \sum_{i=1}^{\ell} A_{i-1}B_i.
\]
For the sake of notational simplicity, we introduce an auxiliary variable $B_0$ that is correlated with $A_0$ such that $(A_0,B_0)\sim \mathcal{N}\pth{\begin{pmatrix} 0 \\ 0 \end{pmatrix}, \begin{pmatrix} 1 & \rho \\ \rho & 1 \end{pmatrix}}$. Then \begin{align}
    m_\ell
    \triangleq\expect{e^{t\beta_P\pth{\maH^f_\pi}}}
    &=\expect{ \expect{ \prod_{i=0}^{\ell-1} e^{t A_{i}B_{i+1}} | B_0\dots B_\ell} } 
    =\expect{  \prod_{i=0}^{\ell-1} \expect{ e^{t A_{i}B_{i+1}} | B_{i}, B_{i+1}}  }\nonumber \\
    &=\expect{\prod_{i=0}^{\ell-1} \exp\pth{t\rho B_{i}B_{i+1}+\frac{1}{2}t^2(1-\rho^2)B_{i+1}^2}} \label{eq:MGF-path-gauss},
\end{align}
where the last equality follows from $A_{i}|B_{i}\sim \maN(\rho B_{i},1-\rho^2)$ and $\expect{\exp(tZ)} = \exp\pth{t\mu+t^2\nu^2/2}$ for $Z\sim \maN(\mu,\nu^2)$.

Define $\mathbf{B}_\ell \triangleq [B_0,B_1,\cdots,B_\ell]$ and a matrix \begin{align*}
    \mathbf{W}_\ell \triangleq \begin{bmatrix}1 & -t\rho & 0 & \cdots & 0 \\ -t\rho & 1-t^2(1-\rho^2) & -t\rho & \cdots & 0 \\ 0&-t\rho & 1-t^2(1-\rho^2)  & \cdots & 0 \\ \vdots & \vdots & \vdots & \ddots & -t\rho \\ 0 & 0 & \cdots 0  & -t\rho & 1-t^2(1-\rho^2)\end{bmatrix},
\end{align*}
where $(\mathbf{W}_\ell)_{00} = 1,(\mathbf{W}_\ell)_{ii} = 1-t^2(1-\rho^2)$ for any $1\le i\le \ell$ and $(\mathbf{W}_\ell)_{ij} = -t\rho$ for any $|i-j| = 1$ with $0\le i,j\le \ell$. By~\eqref{eq:MGF-path-gauss}, we have \begin{align}
   \nonumber  m_\ell &= \expect{\prod_{i=0}^{\ell-1} \exp\pth{t\rho B_{i}B_{i+1}+\frac{1}{2}t^2(1-\rho^2)B_{i+1}^2}}\\\nonumber &=\int\cdots \int\pth{\frac{1}{\sqrt{2\pi}}}^{\ell+1}\exp\pth{-\frac{1}{2}\sum_{i=0}^\ell B_i^2}\exp\pth{\sum_{i=0}^{\ell-1}\pth{t\rho B_{i}B_{i+1}+\frac{1}{2}t^2(1-\rho^2)B_{i+1}^2}}\,\mathrm{d}B_0\cdots \mathrm{d}B_{\ell}\\\label{eq:MGF-ml-part1}&=\int\cdots \int \pth{\frac{1}{\sqrt{2\pi}}}^{\ell+1} \exp\pth{-\frac{1}{2} \mathbf{B}_\ell^\top  \mathbf{W}_\ell \mathbf{B}_\ell}\, \mathrm{d}B_0\cdots \mathrm{d}B_{\ell}.
\end{align}
If $\mathbf{W}_\ell$ is positive definite, then $m_\ell = \det\pth{\mathbf{W}_\ell}^{-1/2}$ by~\eqref{eq:MGF-ml-part1}. 
We then prove $\mathbf{W}_\ell$ is positive definite and compute the explicit formula of $\det\pth{\mathbf{W}_\ell}$.
Expanding the last column of $\mathbf{W}_\ell$ yields that \begin{align*}
    \det(\mathbf{W}_\ell) = \pth{1-t^2(1-\rho^2)} \det(\mathbf{W}_{\ell-1})-t^2 \rho^2 \det(\mathbf{W}_{\ell-2}),\text{ for any }\ell\ge 2.
\end{align*}
Therefore, the general term of $\det(\mathbf{W}_\ell)$ is determined by \begin{align*}
    \det(\mathbf{W}_\ell) = \alpha_1 \lambda_1^\ell +\alpha_2 \lambda_2^\ell,
\end{align*}
where $\lambda_1>\lambda_2$ are two roots of equation $x^2 - \pth{1-t^2(1-\rho^2)} x+ t^2 \rho^2 = 0$. The coefficients $\alpha_1$ and $\alpha_2$ can be determined via the first two terms $\det(\mathbf{W}_0) =1$ and $\det(\mathbf{W}_1) =1-t^2$. Specifically, \begin{align*}
    \alpha_1 = {\frac{1}{2}+\frac{1-t^2(1+\rho^2)}{2\sqrt{(1-t^2(1-\rho^2))^2-4t^2\rho^2}}}, \quad \alpha_2= {\frac{1}{2}-\frac{1-t^2(1+\rho^2)}{2\sqrt{(1-t^2(1-\rho^2))^2-4t^2\rho^2}}}.
\end{align*}
Since $1-t^2(1-\rho^2)-2t\rho = (1-t\rho -t)(1-t\rho+t)>0$ by $t<\frac{1}{1+\rho}$, we obtain $(1-t^2(1-\rho^2))^2-4t^2\rho^2>0$, and hence \begin{align*}
    \frac{1-t^2(1+\rho^2)}{\sqrt{(1-t^2(1-\rho^2))^2-4t^2\rho^2}} = \sqrt{\frac{(1-t^2(1-\rho^2))^2-4t^2\rho^2+4t^4\rho^2}{(1-t^2(1-\rho^2))^2-4t^2\rho^2}}>1,
\end{align*}
which implies that $\alpha_1>1$ and $\alpha_2<0$. Since $\lambda_1+\lambda_2 = 1-t^2(1-\rho^2)>0$ by $t<\frac{1}{1+\rho}$ and $\lambda_1\lambda_2 = t^2 \rho^2>0$ and the discriminant $(1-t^2(1-\rho^2))^2-4t^2\rho^2>0$ , we have $\lambda_1>\lambda_2>0$. Therefore, \begin{align*}
    \det(\mathbf{W}_\ell) =  \alpha_1 \lambda_1^\ell+\alpha_2 \lambda_2^\ell = \alpha_1(\lambda_1^\ell-\lambda_2^\ell)+\lambda_2^\ell>0, \text{ for any }\ell \ge 0,
\end{align*}
which implies that $\mathbf{W}_\ell$ is positive definite by Sylvester’s criterion (see, e.g., \cite[Theorem 7.2.5]{horn2012matrix}). Denote $\mathbf{W}_\ell = \mathbf{J}_\ell ^\top \mathbf{J}_\ell$ with $\mathbf{J}_\ell\in \mathbb{R}^{(\ell+1)\times (\ell+1)}$. Then, by~\eqref{eq:MGF-ml-part1}, \begin{align*}
    m_\ell 
    &= \int\cdots \int \pth{\frac{1}{\sqrt{2\pi}}}^{\ell+1} \exp\pth{-\frac{1}{2} \mathbf{B}_\ell^\top  \mathbf{W}_\ell \mathbf{B}_\ell}\, \mathrm{d}B_0\cdots \mathrm{d}B_{\ell}\\
    &= \int\cdots \int \pth{\frac{1}{\sqrt{2\pi}}}^{\ell+1} \exp\pth{-\frac{1}{2} \mathbf{B}_\ell^\top  \mathbf{J}_\ell^\top \mathbf{J}_\ell \mathbf{B}_\ell}\, \mathrm{d}B_0\cdots \mathrm{d}B_{\ell} = \qth{\det\pth{\mathbf{J}_\ell}}^{-1} = \qth{\det\pth{\mathbf{W}_\ell}}^{-1/2}.
\end{align*}

The analysis for cycles follows from similar arguments.
Consider a cycle $C$ of size $\ell$ denoted by $[e_1\dots e_\ell]$ as illustrated in Figure~\ref{fig:general_circle}. 
For each $i=1,\dots,\ell$, define $A_{i-1}\triangleq \beta_{e_i}(G_1)$ and $B_i\triangleq \beta_{\pi(e_i)}(G_2)$. 
We also let $B_0=B_\ell$ for notational simplicity. 
Then $(A_i,B_i)\sim \mathcal{N}\pth{\begin{pmatrix} 0 \\ 0 \end{pmatrix}, \begin{pmatrix} 1 & \rho \\ \rho & 1 \end{pmatrix}}$ for $i=0,\dots,\ell-1$. Recall that $\lambda_1+\lambda_2=1-t^2(1-\rho^2)$ and $\lambda_1\lambda_2 = t^2\rho^2$.
Following a similar argument as \eqref{eq:MGF-path-gauss}, we have \begin{align*}
    \ti{m}_\ell &\triangleq \expect{e^{t\beta_C(\maH^f_\pi)}} = \expect{\prod_{i=0}^{\ell-1} \exp\pth{t\rho B_{i}B_{i+1}+\frac{1}{2} t^2(1-\rho^2) B_{i+1}^2}} \\&= \int \cdots \int  
 \pth{\frac{1}{\sqrt{2\pi}}}^\ell \exp\pth{-\frac{1}{2} \sum_{i=0}^{\ell-1} B_i^2}\exp\pth{\sum_{i=0}^{\ell-1} \pth{t\rho B_{i}B_{i+1}+\frac{1}{2}t^2(1-\rho^2)B_{i+1}^2}} \,\mathrm{d} B_0\cdots \mathrm{d} B_{\ell-1}\\&= \int \cdots \int  
 \pth{\frac{1}{\sqrt{2\pi}}}^\ell \exp\pth{-\frac{1}{2} \sum_{i=0}^{\ell-1} \pth{\lambda_1^{1/2} B_{i}- \lambda_2^{1/2} B_{i+1}}^2} \,\mathrm{d} B_0\cdots \mathrm{d} B_{\ell-1}\\ &=\int\cdots \int \pth{\frac{1}{\sqrt{2\pi}}}^{\ell} \exp\pth{-\frac{1}{2} \mathbf{B}_{\ell-1}^\top  \ti{\mathbf{J}}_{\ell-1}^\top \ti{\mathbf{J}}_{\ell-1} \mathbf{B}_{\ell-1}}\, \mathrm{d}B_0\cdots \mathrm{d}B_{\ell-1} = \qth{\det(\ti{\mathbf{J}}_{\ell-1})}^{-1},
\end{align*}
where \begin{align*}
    \ti{\mathbf{J}}_{\ell-1} \triangleq \begin{bmatrix}
        \lambda_1^{1/2} & -\lambda_{2}^{1/2} &0& \cdots  &0 \\
        0&\lambda_1^{1/2 }&-\lambda_2^{1/2 } &\cdots &0\\
        0&0& \lambda_1^{1/2 }&\cdots &0\\
        \vdots & \vdots &\vdots &\ddots &-\lambda_2^{1/2 }\\
        -\lambda_2^{1/2 } &0&\cdots 0&0& \lambda_1^{1/2 }
    \end{bmatrix}
\end{align*} and hence $\det(\ti{\mathbf{J}}_{\ell-1}) = \lambda_1^{\ell/2} - \lambda_2^{\ell/2}$. Therefore, \begin{align*}
    m_{\ell}^{-2} - \ti{m}_\ell^{-2} &= \alpha_1 \lambda_1^\ell +\alpha_2 \lambda_2^\ell - \pth{\lambda_1^{\ell/2}-\lambda_2^{\ell/2}}^2 \\&=(\alpha_1-1)\pth{\lambda_1^{\ell} - \lambda_2^\ell}+\lambda_2^{\ell/2}\pth{\lambda_1^{\ell/2}-\lambda_2^{\ell/2}}+\lambda_1^{\ell/2} \lambda_2^{\ell/2}>0, 
\end{align*}
which implies that $m_\ell < \ti{m}_\ell$ for any $\ell\ge 1$. 

Next, we prove that $\ti{m}_\ell\le \pth{\ti{m}_2}^{\ell/2}$ for $\ell \ge 2$.  Indeed,  since $\pth{\frac{\lambda_1-\lambda_2}{\lambda_1}}^{\ell/2}+\pth{\frac{\lambda_2}{\lambda_1}}^{\ell/2}\le \frac{\lambda_1-\lambda_2}{\lambda_1}+\frac{\lambda_2}{\lambda_1}=1$ when $\ell\ge 2$, we obtain
$\ti{m}_\ell^{-1} = \lambda_1^{\ell/2} - \lambda_2^{\ell/2}\ge (\lambda_1-\lambda_2)^{\ell/2} = \ti{m}_{2}^{-\ell/2}$. We also have $\ti{m}_{2}^{-1} - \ti{m}_1^{-2} = 2\lambda_2^{1/2}\pth{\lambda_1^{1/2}-\lambda_2^{1/2}}>0$. Recall that $\kappa_\ell^\sfC(t)=\log\pth{\ti{m}_\ell}, \kappa_\ell^\sfC(t)=\log\pth{m_\ell}$, and $m_\ell<\ti{m}_\ell$.
Consequently, \begin{align*}
    \frac{1}{2}\kappa_2^\sfC(t)\le \kappa_1^\sfC(t),\quad \kappa_\ell^{\sfP}(t) \le \kappa_\ell^\sfC(t)\le \frac{\ell}{2} \kappa_2^\sfC(t)\quad\forall \ell\ge 2.
\end{align*}
The last inequality $\kappa_1^\sfP(t)\le \frac{1}{2}\kappa_2^\sfC(t)$ follows from \begin{align*}
    m_1^{-4} -\ti{m}_2^{-2} = \pth{1-t^2}^2- \qth{1-2t^2-2t^2\rho^2+t^4(1-\rho^2)^2} = t^2\rho^2\pth{2+2t^2-t^2\rho^2}\overset{\mathrm{(a)}}{>}0,
\end{align*}
where $\mathrm{(a)}$ is because $t^2 \rho^2 < \frac{\rho^2}{(1+\rho)^2}<1$ and $2+2t^2>1$.

Finally, since the summands over different connected components are independent, it follows that 
\begin{align*}
\log \expect{e^{t \beta_\maE(\maH_\pi^f)}}
& = \sum_{P\in\calP} \kappa_{|P|}^{\sfP}(t) + \sum_{C\in\calC} \kappa_{|C|}^{\sfC}(t) \\
& \le \sum_{P\in\calP} \frac{|P|}{2}\kappa_{2}^{\sfC}(t) + \sum_{C\in\calC: |C|\ge 2} \frac{|C|}{2}\kappa_{2}^{\sfC}(t) + \sum_{C\in\calC: |C|=1} \kappa_{1}^{\sfC}(t) \\
& = \frac{|\calE|}{2}\kappa_{2}^{\sfC}(t) + |\{C\in\calC: |C|=1\}| \pth{\kappa_{1}^{\sfC}(t) - \frac{1}{2}\kappa_{2}^{\sfC}(t) },
\end{align*}
where the last equality uses fact that $|\calE| = \sum_{P\in\calP} |P| + \sum_{C\in\calC} |C|$.

\paragraph{Gaussian Wigner Model, Part 2} In this part, we focus on the Gaussian Wigner model with $f(x,y) = -\frac{1}{2}(x-y)^2$.
The lower-order cumulants can be calculated directly:
\begin{align}
    \label{eq:kappa1=general}\kappa_1^\sfC(t) &= -\frac{1}{2}\log\pth{1+2t(1-\rho)},
    \quad
    \kappa_1^\sfP(t) = -\frac{1}{2}\log(1+2t),\\\label{eq:kappa2-general} \kappa_2^\sfC(t) &= -\frac{1}{2}\log\pth{(1+2t)^2-4t^2\rho^2}.
\end{align}
We first evaluate the moment generating function for paths. Consider a path $P$ of size $\ell$ denoted by $\langle e_1 e_2 \dots e_{\ell} \rangle$ as illustrated in Figure~\ref{fig:general_path}.
For each $i=1,\dots,\ell$, define $A_{i-1}\triangleq \beta_{e_i}(G_1)$ and $B_i\triangleq \beta_{\pi(e_i)}(G_2)$. Then $(A_i,B_i)\sim\mathcal{N}\pth{\begin{pmatrix} 0 \\ 0 \end{pmatrix}, \begin{pmatrix} 1 & \rho \\ \rho & 1 \end{pmatrix}}$. 
By definition~\eqref{eq:calE-counter}, 
\[
\beta_P(\maH_\pi^f) = \sum_{i=1}^{\ell}-\frac{1}{2}\pth{\beta_{e_i}(G_1)-\beta_{\pi(e_i)}(G_2)}^2
= \sum_{i=1}^{\ell} -\frac{1}{2}\pth{A_{i-1}-B_i}^2.
\]
For the sake of notational simplicity, we introduce an auxiliary variable $B_0$ that is correlated with $A_0$ such that $(A_0,B_0)\sim \mathcal{N}\pth{\begin{pmatrix} 0 \\ 0 \end{pmatrix}, \begin{pmatrix} 1 & \rho \\ \rho & 1 \end{pmatrix}}$. Then 
\begin{align}
    m_\ell
    \triangleq\expect{e^{t\beta_P\pth{\maH_\pi^f}}}
    &=\expect{\expect{ \prod_{i=0}^{\ell-1} e^{-\frac{t}{2} \pth{A_{i}-B_{i+1}}^2}| B_0\dots B_\ell} } 
    =\expect{\prod_{i=0}^{\ell-1} \expect{ e^{-\frac{t}{2} \pth{A_{i}-B_{i+1}}^2} | B_{i}, B_{i+1}}}\nonumber \\
    &=\expect{\prod_{i=0}^{\ell-1} \frac{1}{\sqrt{1+t(1-\rho^2)}}
    \exp\pth{\frac{-\frac{t}{2}\pth{\rho B_{i}-B_{i+1}}^2}{1+t(1-\rho^2)}}} \label{eq:MGF-path-gauss-general},
\end{align}
where the last equation follows from $(A_{i}-B_{i+1})\mid (B_{i},B_{i+1})\sim \maN(\rho B_{i}-B_{i+1},1-\rho^2)$ and $\expect{\exp(-tZ^2)} = \frac{1}{\sqrt{1+2t\nu^2}}\exp\pth{\frac{-\mu^2 t}{1+2t\nu^2}}$ for $Z\sim \maN(\mu,\nu^2)$ when $t>0$. 

Define $\mathbf{B}_\ell \triangleq [B_0,B_1,\cdots,B_\ell]$ and a matrix \begin{align*}
    \mathbf{W}_\ell \triangleq \begin{bmatrix}1+t & -t\rho & 0 & \cdots & 0 \\ -t\rho & 1+2t & -t\rho & \cdots & 0 \\ 0&-t\rho & 1+2t  & \cdots & 0 \\ \vdots & \vdots & \vdots & \ddots & -t\rho \\ 0 & 0 & \cdots 0  & -t\rho & 1+t(2-\rho^2)\end{bmatrix},
\end{align*}
where $(\mathbf{W}_\ell)_{00} = 1+t, (\mathbf{W}_\ell)_{\ell\ell} = 1+t(2-\rho^2),(\mathbf{W}_\ell)_{ii} = 1+2t$ for any $1\le i\le \ell-1$ and $(\mathbf{W}_\ell)_{ij} = -t\rho$ for any $|i-j| = 1$ with $0\le i,j\le \ell$. We note that the symmetric matrix 
 $\mathbf{W}_\ell$ is strictly diagonally dominant with positive diagonal entries, and then it is positive definite. This follows from the eigenvalues of symmetric matrix being real, and the Gershgorin's circle theorem \cite{gersgorin1931uber}.
 Denote $\mathbf{W}_\ell = \mathbf{J}_\ell ^\top \mathbf{J}_\ell$ with $\mathbf{J}_\ell\in \mathbb{R}^{(\ell+1)\times (\ell+1)}$. 
 Let $C_i = \frac{B_i}{\sqrt{1+t(1-\rho^2)}}$ for all $0\le i\le \ell$ and denote $\mathbf{C}_\ell\triangleq [C_0,C_1,\cdots,C_\ell]$.
By~\eqref{eq:MGF-path-gauss-general}, we note that \begin{align}
    \nonumber m_\ell &= \expect{\prod_{i=0}^{\ell-1} \frac{1}{\sqrt{1+t(1-\rho^2)}}
    \exp\pth{\frac{-\frac{t}{2}\pth{\rho B_{i}-B_{i+1}}^2}{1+t(1-\rho^2)}}}\\\nonumber
    &=\int \cdots \int \frac{\sqrt{1+t(1-\rho^2)}}{\pth{\sqrt{2\pi(1+t(1-\rho^2))}}^{\ell+1}}\exp\pth{-\frac{1}{2}\sum_{i=0}^{\ell} B_i^2-\frac{t}{2}\sum_{i=0}^{\ell-1}\frac{\pth{\rho B_i-B_{i+1}}^2}{1+t(1-\rho^2)}}\,\mathrm{d}B_0\cdots \mathrm{d}B_\ell\\\nonumber&=\int\cdots \int\frac{\sqrt{1+t(1-\rho^2)}}{\pth{\sqrt{2\pi}}^{\ell+1}}\exp\pth{-\frac{1}{2}\sum_{i=0}^\ell (1+t(1-\rho^2))C_i^2 -\frac{t}{2}\sum_{i=0}^{\ell-1}\pth{\rho C_i-C_{i+1}}^2}\,\mathrm{d}C_0\cdots \mathrm{d}C_\ell\\\nonumber
    &=\sqrt{1+t(1-\rho^2)}\int \cdots \int \pth{\frac{1}{\sqrt{2\pi}}}^{\ell+1} \exp\pth{-\frac{1}{2}\mathbf{C}_\ell^\top \mathbf{W}_\ell \mathbf{C}_\ell }\,\mathrm{d}C_0\cdots \mathrm{d}C_\ell\\\nonumber&=\sqrt{1+t(1-\rho^2)}\int \cdots \int \pth{\frac{1}{\sqrt{2\pi}}}^{\ell+1} \exp\pth{-\frac{1}{2}\mathbf{C}_\ell^\top \mathbf{J}_\ell^\top \mathbf{J}_\ell \mathbf{C}_\ell }\,\mathrm{d}C_0\cdots \mathrm{d}C_\ell 
    \\\label{eq:MGF-ml-part2}&= \sqrt{1+t(1-\rho^2)}\qth{\det\pth{\mathbf{J}_\ell}}^{-1} = \qth{\frac{\det\pth{\mathbf{W}_\ell}}{1+t(1-\rho^2)}}^{-1/2}
\end{align}

We then compute the explicit formula of $\det\pth{\mathbf{W}_\ell}$.
Let \begin{align*}
    \mathbf{U}_\ell \triangleq \begin{bmatrix}1+t & -t\rho & 0 & \cdots & 0 \\ -t\rho & 1+2t & -t\rho & \cdots & 0 \\ 0&-t\rho & 1+2t  & \cdots & 0 \\ \vdots & \vdots & \vdots & \ddots & -t\rho \\ 0 & 0 & \cdots 0  & -t\rho & 1+2t\end{bmatrix},
\end{align*}
where $(\mathbf{U}_\ell)_{00} = 1+t, (\mathbf{U}_\ell)_{ii} = 1+2t$ for any $1\le i\le \ell$ and $(\mathbf{U}_\ell)_{ij} = -t\rho$ for any $|i-j| = 1$ with $0\le i,j\le \ell$. Expanding the last column of $\mathbf{W}_\ell$ and $\mathbf{U}_\ell$ yields that \begin{align*}
    \det(\mathbf{W}_\ell) &= \pth{1+t(2-\rho^2)} \det(\mathbf{U}_{\ell-1})-t^2 \rho^2 \det(\mathbf{U}_{\ell-2}),\text{ for any }\ell\ge 2,\\
    \det(\mathbf{U}_\ell) &= \pth{1+2t} \det(\mathbf{U}_{\ell-1})-t^2 \rho^2 \det(\mathbf{U}_{\ell-2}),\text{ for any }\ell\ge 2.
\end{align*}
Therefore, the general form of $\det(\mathbf{U}_\ell)$ is determined by $\det(\mathbf{U}_\ell) = \alpha_1' \lambda_1^\ell+\alpha_2'\lambda_2^\ell$, where $\lambda_1>\lambda_2$ are two roots of the equation $x^2-(1+2t)x+t^2\rho^2=0$. Since $\lambda_1+\lambda_2 = 1+2t>0$ and $\lambda_1\lambda_2=t^2\rho^2>0$ and the discriminant $(1+2t)^2-4t^2\rho^2>0$, we obtain $\lambda_1>\lambda_2>0$.
Consequently, the general form of $\det(\mathbf{W}_\ell)$ is given by \begin{align*}
    \det(\mathbf{W}_\ell) = \alpha_1'' \lambda_1^\ell+\alpha_2''\lambda_2^\ell.
\end{align*}
The coefficients $\alpha_1''$ and $\alpha_2''$ can be determined via the first two terms $\det(\mathbf{W}_1) = (1+t(1-\rho^2))(1+2t)$ and $\det(\mathbf{W}_2) = (1+t(1-\rho^2))(1+4t+t^2(4-\rho^2))$. Then we get \begin{align*}
    \det(\mathbf{W}_\ell) = (1+t(1-\rho^2))\qth{\pth{\frac{1 }{2}+\frac{1+2t}{2\sqrt{(1+2t)^2-4t^2\rho^2}}} \lambda_1^\ell+\pth{\frac{1 }{2}-\frac{1+2t}{2\sqrt{(1+2t)^2-4t^2\rho^2}}} \lambda_2^\ell}.
\end{align*}
Denote $\det(\mathbf{W}_\ell) = (1+t(1-\rho^2))\pth{\alpha_1 \lambda_1^\ell+\alpha_2 \lambda_2^\ell}$. Then, $\alpha_1+\alpha_2 = 1$.
Since $1+2t>\sqrt{(1+2t)^2-4t^2\rho^2}$, we obtain that $\alpha_1>0$ and $\alpha_2<0$. 
 Then, by~\eqref{eq:MGF-ml-part2},\begin{align*}
    m_\ell  = \qth{\frac{\det\pth{\mathbf{W}_\ell}}{1+t(1-\rho^2)}}^{-1/2} = \pth{\alpha_1\lambda_1^\ell+\alpha_2\lambda_2^\ell}^{-1/2}.
\end{align*}

The analysis for cycles follows from similar arguments.
Consider a cycle $C$ of size $\ell$ denoted by $[e_1\dots e_\ell]$ as illustrated in Figure~\ref{fig:general_circle}. 
For each $i=1,\dots,\ell$, define $A_{i-1}\triangleq \beta_{e_i}(G_1)$ and $B_i\triangleq \beta_{\pi(e_i)}(G_2)$. 
We also let $B_0=B_\ell$ for notational simplicity. 
Then $(A_i,B_i)\sim \mathcal{N}\pth{\begin{pmatrix} 0 \\ 0 \end{pmatrix}, \begin{pmatrix} 1 & \rho \\ \rho & 1 \end{pmatrix}}$ for $i=0,\dots,\ell-1$. Recall that $\lambda_1+\lambda_2=1+2t$ and $\lambda_1\lambda_2 = t^2\rho^2$.
Following a similar argument as \eqref{eq:MGF-path-gauss-general}, we have \begin{align*}
    \ti{m}_\ell &\triangleq \expect{e^{t\beta_C(\maH_\pi^f)}} = \expect{\prod_{i=0}^{\ell-1} \frac{1}{\sqrt{1+t(1-\rho^2)}}
    \exp\pth{\frac{-\frac{t}{2}\pth{\rho B_{i}-B_{i+1}}^2}{1+t(1-\rho^2)}}}\\&=\int \cdots\int \pth{\frac{1}{\sqrt{2\pi(1+t(1-\rho^2))}}}^\ell \exp\pth{-\frac{1}{2}\sum_{i=0}^{\ell-1} B_i^2-\frac{t}{2}\sum_{i=0}^{\ell-1}\frac{\pth{\rho B_i-B_{i+1}}^2}{1+t(1-\rho^2)}}\,\mathrm{d} B_0\cdots \mathrm{d} B_{\ell-1}\\&=\int\cdots\int \pth{\frac{1}{\sqrt{2\pi}}}^{\ell}\exp\pth{-\frac{1}{2}\sum_{i=0}^{\ell-1}(1+2t)C_i^2+\sum_{i=0}^{\ell-1}t\rho C_iC_{i+1} }\,\mathrm{d} C_0\cdots \mathrm{d} C_{\ell-1} \\&= \int \cdots \int  
 \pth{\frac{1}{\sqrt{2\pi}}}^\ell \exp\pth{-\frac{1}{2} \sum_{i=0}^{\ell-1} \pth{\lambda_1^{1/2} C_{i}- \lambda_2^{1/2} C_{i+1}}^2} \,\mathrm{d} C_0\cdots \mathrm{d} C_{\ell-1}\\ &=\int\cdots \int \pth{\frac{1}{\sqrt{2\pi}}}^{\ell} \exp\pth{-\frac{1}{2} \mathbf{C}_{\ell-1}^\top  \ti{\mathbf{J}}_{\ell-1}^\top \ti{\mathbf{J}}_{\ell-1} \mathbf{C}_{\ell-1}}\, \mathrm{d}C_0\cdots \mathrm{d}C_{\ell-1} = \qth{\det(\ti{\mathbf{J}}_{\ell-1})}^{-1},
\end{align*}
where \begin{align*}
    \ti{\mathbf{J}}_{\ell-1} = \begin{bmatrix}
        \lambda_1^{1/2} & -\lambda_{2}^{1/2} &0& \cdots  &0 \\
        0&\lambda_1^{1/2 }&-\lambda_2^{1/2 } &\cdots &0\\
        0&0& \lambda_1^{1/2 }&\cdots &0\\
        \vdots & \vdots &\vdots &\ddots &-\lambda_2^{1/2 }\\
        -\lambda_2^{1/2 } &0&\cdots 0&0& \lambda_1^{1/2 }
    \end{bmatrix}
\end{align*} and hence $\det(\ti{\mathbf{J}}_{\ell-1}) = \lambda_1^{\ell/2} - \lambda_2^{\ell/2}$. Therefore, \begin{align*}
    m_{\ell}^{-2} - \ti{m}_\ell^{-2} &= \alpha_1 \lambda_1^\ell +\alpha_2 \lambda_2^\ell - \pth{\lambda_1^{\ell/2}-\lambda_2^{\ell/2}}^2 \\&=(\alpha_1-1)\pth{\lambda_1^{\ell} - \lambda_2^\ell}+\lambda_2^{\ell/2}\pth{\lambda_1^{\ell/2}-\lambda_2^{\ell/2}}+\lambda_1^{\ell/2} \lambda_2^{\ell/2}>0, 
\end{align*}
which implies that $m_\ell < \ti{m}_\ell$ for any $\ell\ge 1$. 

Next, we prove that $\ti{m}_\ell\le \pth{\ti{m}_2}^{\ell/2}$ for $\ell \ge 2$.  Indeed,  $\ti{m}_\ell^{-1} = \pth{\lambda_2+(\lambda_1-\lambda_2)}^{\ell/2} - \lambda_2^{\ell/2}\ge (\lambda_1-\lambda_2)^{\ell/2} = \ti{m}_{2}^{-\ell/2}$, and $\ti{m}_{2}^{-1} - \ti{m}_1^{-2} = 2\lambda_2^{1/2}\pth{\lambda_1^{1/2}-\lambda_2^{1/2}}>0$. Recall that $\kappa_\ell^\sfC(t)=\log\pth{\ti{m}_\ell}, \kappa_\ell^\sfC(t)=\log\pth{m_\ell}$, and $m_\ell<\ti{m}_\ell$.
Consequently, \begin{align*}
    \frac{1}{2}\kappa_2^\sfC(t)\le \kappa_1^\sfC(t),\quad \kappa_\ell^{\sfP}(t) \le \kappa_\ell^\sfC(t)\le \frac{\ell}{2} \kappa_2^\sfC(t)\quad\forall \ell\ge 2.
\end{align*}
The last inequality $\kappa_1^\sfP(t)\le \frac{1}{2}\kappa_2^\sfC(t)$ follows from \begin{align*}
    m_1^{-4} -\ti{m}_2^{-2} = \pth{1+2t}^2- \qth{(1+2t)^2-4t^2\rho^2} = 4t^2\rho^2>0.
\end{align*}

Finally, since the summands over different connected components are independent, it follows that 
\begin{align*}
\log \expect{e^{t \beta_\maE(\maH_\pi^f)}}
& = \sum_{P\in\calP} \kappa_{|P|}^{\sfP}(t) + \sum_{C\in\calC} \kappa_{|C|}^{\sfC}(t) \\
& \le \sum_{P\in\calP} \frac{|P|}{2}\kappa_{2}^{\sfC}(t) + \sum_{C\in\calC: |C|\ge 2} \frac{|C|}{2}\kappa_{2}^{\sfC}(t) + \sum_{C\in\calC: |C|=1} \kappa_{1}^{\sfC}(t) \\
& = \frac{|\calE|}{2}\kappa_{2}^{\sfC}(t) + |\{C\in\calC: |C|=1\}| \pth{\kappa_{1}^{\sfC}(t) - \frac{1}{2}\kappa_{2}^{\sfC}(t) },
\end{align*}
where the last equality uses fact that $|\calE| = \sum_{P\in\calP} |P| + \sum_{C\in\calC} |C|$.

\begin{remark}
    We have two bounds for large $\ell$ in the Gaussian Wigner model, namely $\kappa_\ell^\sfP(t)\le \kappa_\ell^\sfC(t)$ and $\kappa_\ell^\sfC(t)\le \frac{\ell}{2}\kappa_2^\sfC(t)$. 
 For the first bound, we apply $\frac{1}{\ell} \log\pth{\alpha_1 \lambda_1^\ell+\alpha_2 \lambda_2^\ell}\ge \frac{1}{\ell} \log\pth{\pth{\lambda_1^{\ell/2}-\lambda_2^{\ell/2}}^2}$, where $\alpha_2<0<1<\alpha_1,\alpha_1+\alpha_2 = 1$ and $\lambda_1>\lambda_2>0$.
Consequently, 
$-2\pth{\lambda_1 + \frac{\log \alpha_1}{\ell} }
\le \frac{1}{\ell}\kappa_\ell^\sfP(t)
\le \frac{1}{\ell}\kappa_\ell^\sfC(t)
\le -2\pth{\lambda_1+\frac{1}{\ell} \log\qth{1-2\pth{\frac{\lambda_2}{\lambda_1}}^{\ell/2}}}.$
Hence, the first bound is essentially tight for large $\ell$.
 The second bound, previously used in \cite{wu2022settling}, applies the inequality  $(x+y)^{\ell/2} -y^{\ell/2}\ge x^{\ell/2}$, which becomes less tight as $\ell$ increases. 
Nevertheless, it suffices for our analysis as the probability of long cycles occurring is relatively small.
\end{remark}

\subsection{Proof of Lemma \ref{lem: strong noise event}}\label{apdsec: strong noise event}
We first upper bound the higher-order cumulants.
By Lemma \ref{lem: upper bound for kappa}, for any $t>0$,
    \begin{align*}
         \log \expect{e^{t \beta_{\maE_\pi}(\maH_\pi^f)}} \le \frac{|\maE_\pi|}{2} \kappa^{\sfC}_2(t)+L\pth{\kappa^{\sfC}_1(t)-\frac{1}{2} \kappa^{\sfC}_2(t)},
    \end{align*}
    where $L$ denotes the number of self-loops. The self-loop for $e$ only happens when $\pi(e) = \pi^*(e)$. For $uv\in \binom{V(G_1)}{2}\backslash \binom{F_\pi}{2}$, by the definition of $F_\pi$, $\pi(u) \neq \pi^*(u)$ or $\pi(v) \neq \pi^*(v)$. Therefore, $\pi(uv) = \pi^*(uv)$ implies that $\pi(u) = \pi^*(v)$ and $\pi(v) = \pi^*(u)$. Since $d(\pi^*,\pi) = k$, we must have $L\le \frac{k}{2}$.
Applying the formulas \eqref{eq:cumulants-C1} and \eqref{eq:cumulants-C2} and the fact that $p_{11}\le p$, we obtain  
\begin{align*}
    \kappa^{\sfC}_2(t) 
    &\le \log\pth{1+2p^2(e^t-1)+p^2(e^t-1)^2}= \log\pth{1+p^2(e^{2t}-1)}\le p^2(e^{2t}-1)
\end{align*}
and  
\begin{align*}
         \kappa^{\sfC}_1(t)-\frac{1}{2}\kappa^{\sfC}_2(t) 
        &= \frac{1}{2}\log\qth{1+\frac{2(p_{11}-p^2)}{p_{11}^2(e^t-1)+2p^2+(e^t-1)^{-1}}}\\
        &\overset{\mathrm{(a)}}{\le} \frac{1}{2}\log\qth{1+\frac{2(p_{11}-p^2)}{2(p_{11}+p^2)}} \overset{\mathrm{(b)}}{\le} \frac{\gamma}{2(\gamma+2)},
    \end{align*}
where $\mathrm{(a)}$ is because $p_{11}^2(e^t-1)+(e^t-1)^{-1}\ge 2\sqrt{p_{11}^2(e^t-1)(e^t-1)^{-1}} = 2p_{11}$ and $\mathrm{(b)}$ is because $p_{11} = p^2(1+\gamma)$ and $\log(1+x)\le x$ for any $x\ge 0$. 
Therefore, we obtain
\[
\log \expect{e^{t\beta_{\maE_\pi}(\maH_\pi^f)}} 
\le  \frac{|\calE_{\pi}|}{2} p^2(e^{2t}-1)+\frac{k\gamma}{4(\gamma+2)}.
\]

We then apply the Chernoff bound the provide an upper bound for $\prob{ \beta_{\maE_\pi}(\maH_\pi^f) \ge \tau_k }$.
For any $t>0$,
\begin{align*}
    \prob{ \beta_{\maE_\pi}(\maH_\pi^f) \ge \tau_k }
    \le \exp\pth{-t\tau_k+ \frac{|\calE_{\pi}|}{2} p^2(e^{2t}-1) + \frac{k\gamma}{4(2+\gamma)} }. 
\end{align*}
Let $t = \frac{1}{2}\log\pth{\frac{\tau_k}{|\maE_\pi|p^2}}$. 
Then $t>0$ by the assumption $\tau_k>|\maE_\pi|p^2$. 
We obtain that
\begin{align*}
    \prob{\beta_{\maE_\pi}(\maH_\pi^f) \ge \tau_k}\le \exp\pth{-\frac{\tau_k}{2}\log\pth{\frac{\tau_k}{|\maE_\pi|p^2}}+\frac{\tau_k}{2}-\frac{|\maE_\pi|p^2}{2}+\frac{k\gamma}{4(2+\gamma)}}.
\end{align*}

\subsection{Proof of Lemma~\ref{lem: strong noise event-Gaussian}}\label{apdsec: strong noise event-Gaussian}

We first upper bound the higher-order cumulants. Let $t = \frac{\rho}{3\sqrt{1+\rho^2}}$. Then $0<t<\frac{1}{1+\rho}$ since $0< \rho< 1$. By Lemma~\ref{lem: upper bound for kappa}, 
    \begin{align*}
         \log \expect{e^{t \beta_{\maE_\pi}(\maH_\pi^f)}} \le \frac{|\maE_\pi|}{2} \kappa^{\sfC}_2(t)+L\pth{\kappa^{\sfC}_1(t)-\frac{1}{2} \kappa^{\sfC}_2(t)},
    \end{align*}
    where $L$ denotes the number of self-loops. The self-loop for $e$ only happens when $\pi(e) = \pi^*(e)$. For $uv\in \binom{V(G_1)}{2}\backslash \binom{F_\pi}{2}$, by the definition of $F_\pi$, $\pi(u) \neq \pi^*(u)$ or $\pi(v) \neq \pi^*(v)$. Therefore, $\pi(uv) = \pi^*(uv)$ implies that $\pi(u) = \pi^*(v)$ and $\pi(v) = \pi^*(u)$. Since $d(\pi^*,\pi) = k$, we must have $L\le \frac{k}{2}$.
Applying the formulas~\eqref{eq:cumulants-C1-Gaussian} and \eqref{eq:cumulants-C2-Gaussian}, we obtain that \begin{align*}
    \kappa_1^\sfC(t)-\frac{1}{2}\kappa_2^\sfC(t) =& \frac{1}{4}\log\qth{\frac{\pth{1-t^2(1-\rho^2)-2t\rho}\pth{1-t^2(1-\rho^2)+2t\rho}}{\pth{1-t^2(1-\rho^2)-2t\rho}^2}}\\=&\frac{1}{4}\log\qth{1+\frac{4t\rho}{1-t^2(1-\rho^2)-2t\rho}}\overset{\mathrm{(a)}}{\le} \frac{1}{4}\log\pth{1+\frac{12\rho^2}{9-7\rho^2+\rho^4}}\overset{\mathrm{(b)}}{\le} \frac{\log 5}{4},
\end{align*}
where $\mathrm{(a)}$ is because $\frac{4t\rho}{1-t^2(1-\rho^2)-2t\rho} = \frac{12\rho^2}{9\sqrt{1+\rho^2}-\frac{\rho^2(1-\rho^2)}{\sqrt{1+\rho^2}} - 6\rho^2}\le \frac{12\rho^2}{9-\rho^2(1-\rho^2)-6\rho^2}=\frac{12\rho^2}{9-7\rho^2+\rho^4}$; $\mathrm{(b)}$ is because $\frac{12x^2}{9-7x^2+x^4}$ is increasing on $(0,1)$ and $\rho<1$.

By the Chernoff bound, we have that
\begin{align*}
    \prob{\beta_{\maE_\pi}\pth{\maH_\pi^f} \ge \tau_k} \le \exp\pth{-t\tau_k+\frac{|\maE_\pi|}{2} \kappa^{\sfC}_2(t)+L\pth{\kappa^{\sfC}_1(t)-\frac{1}{2} \kappa^{\sfC}_2(t)}}.
\end{align*}
It remains to upper bound $-t\tau_k+\frac{|\maE_\pi|}{2}\kappa_2^\sfC(t)$. We note that
\begin{align*}
    -t\tau_k+\frac{|\maE_\pi|}{2}\kappa_2^\sfC(t)&=-t\tau_k-\frac{|\maE_\pi|}{4}\log\qth{1-\frac{2\rho^2}{9}+\frac{\rho^4(1-\rho^2)^2}{81(1+\rho^2)^2}}\\&\overset{\mathrm{(a)}}{\le} -t\tau_k+\frac{9|\maE_\pi|}{28}\qth{\frac{2\rho^2}{9}-\frac{\rho^4(1-\rho^2)^2}{81(1+\rho^2)^2}}\overset{\mathrm{(b)}}{\le} -\frac{\rho\tau_k}{6}+\frac{\rho^2|\maE_\pi|}{14},
\end{align*}
where $\mathrm{(a)}$ is because $\frac{2\rho^2}{9}-\frac{\rho^4(1-\rho^2)^2}{81(1+\rho^2)^2} = \frac{\rho^2\pth{18(1+\rho^2)^2-\rho^2\pth{1-\rho^2}^2}}{81(1+\rho^2)^2}\ge 0$ and $\log(1-x)\ge -\frac{9x}{7}$ for any $0\le x\le \frac{2}{9}$; $\mathrm{(b)}$ is because $t\ge \frac{\rho}{6}$. Therefore, we obtain that \begin{align*}
     \prob{\beta_{\maE_\pi}\pth{\maH_\pi^f} \ge \tau_k}\le \exp\pth{-\frac{\rho\tau_k}{6}+\frac{\rho^2|\maE_\pi|}{14}+\frac{\log 5}{8}k}.
\end{align*}

\subsection{Proof of Lemma~\ref{lem:strong noise event-Gaussian-2}}\label{apdsec: strong noise event-Gaussian-2}
    We first upper bound the higher-order cumulants. By Lemma~\ref{lem: upper bound for kappa}, for any $t>0$, \begin{align*}
        \log \expect{e^{t\beta_{\maE_\pi}\pth{\maH_\pi^f}}} \le \frac{|\calE_\pi|}{2} \kappa^{\sfC}_2(t)+L\pth{\kappa^{\sfC}_1(t)-\frac{1}{2} \kappa^{\sfC}_2(t)},
    \end{align*}
    where $L$ denotes the number of self-loops. The self-loop for $e$ only happens when $\pi(e) = \pi^*(e)$. For $uv\in \binom{V(G_1)}{2}\backslash \binom{F_\pi}{2}$, by the definition of $F_\pi$, $\pi(u) \neq \pi^*(u)$ or $\pi(v) \neq \pi^*(v)$. Therefore, $\pi(uv) = \pi^*(uv)$ implies that $\pi(u) = \pi^*(v)$ and $\pi(v) = \pi^*(u)$. Since $d(\pi^*,\pi) = k$, we must have $L\le \frac{k}{2}$.
Let $t = \frac{\rho}{4(1-\rho)}$. Applying the formulas~\eqref{eq:kappa1=general} and \eqref{eq:kappa2-general}, we obtain that \begin{align*}
    \kappa_1^\sfC(t) - \frac{1}{2}\kappa_2^\sfC(t) &= \frac{1}{4}\log\qth{\frac{(1+2t-2t\rho)(1+2t+2t\rho)}{(1+2t(1-\rho))^2}}\\&=\frac{1}{4}\log \pth{1+\frac{4t\rho}{1+2t-2t\rho}}\le  \frac{1}{4}\log\pth{\frac{1}{1-\rho}},
\end{align*}
where the last inequality follows from $1+\frac{4t\rho}{1+2t-2t\rho} = \frac{2+\rho^2-\rho}{2+\rho}\cdot \frac{1}{1-\rho}\le \frac{1}{1-\rho}$.

By the Chernoff bound, we have that
\begin{align*}
    \prob{\beta_{\maE_\pi}\pth{\maH_\pi^f} \ge \tau_k} \le \exp\pth{-t\tau_k+\frac{|\maE_\pi|}{2} \kappa^{\sfC}_2(t)+L\pth{\kappa^{\sfC}_1(t)-\frac{1}{2} \kappa^{\sfC}_2(t)}}.
\end{align*}
It remains to upper bound $-t\tau_k+\frac{|\maE_\pi|}{2}\kappa_2^\sfC(t)$.
We note that \begin{align*}
    -t\tau_k+\frac{|\maE_\pi|}{2}\kappa_2^\sfC(t) &= -t\tau_k-\frac{|\maE_\pi|}{4} \log\qth{(1+2t)^2-4t^2\rho^2}\\&\le -t\tau_k-\frac{|\maE_\pi|}{4}\log(1+4t)=  -\frac{\rho\tau_k}{4(1-\rho)} - \frac{|\maE_\pi|}{4}\log\pth{\frac{1}{1-\rho}}.
\end{align*}
Therefore, we obtain that
\begin{align*}
    \prob{\beta_{\maE_\pi}(\maH_\pi^f)\ge \tau_k}\le \exp\pth{ -\frac{\rho\tau_k}{4(1-\rho)}- \frac{|\maE_\pi|}{4}\log\pth{\frac{1}{1-\rho}}+\frac{k}{8}\log\pth{\frac{1}{1-\rho}}}.
\end{align*}

\subsection{\hd{Proof of Lemma~\ref{lem:mutual-pack}}}\label{apd:proof-mutual-pack}
\hd{ The size of $\calM_\delta$ follows from the standard volume argument \cite[Theorem 27.3]{polyanskiy2025information}.
For $r\in [m]$, let $B(\pi,r)\triangleq \{\pi':d(\pi,\pi')\le r\}$ denote the ball of radius $r$ centered at $\pi$. Then, we obtain 
\[
|\calM_\delta|\ge \frac{|\maS_{n,m}|}{\max_\pi |B(\pi,(1-\delta)m-1)|}\ge \frac{|\maS_{n,m}|}{\max_\pi |B(\pi,(1-\delta)m)|}.
\]
It remains to evaluate the cardinality of $\maS_{n,m}$ and upper bound the volume of the ball under our distance metric $d$.
It is straightforward to obtain that $|\maS_{n,m}|=\binom{n}{m}^2 m!.$
Let $k=\delta m$.
Note that all elements from $B(\pi,m-k)$ have at least $k$ common mappings. 
To upper bound $|B(\pi,m-k)|$, we first choose $k$ elements from the domain of $\pi$ and map to the same value as $\pi$, and the remaining domain and range of size $m-k$ and the mapping are selected arbitrarily. We get
$
|B(\pi,m-k)|\le \binom{m}{k}\binom{n-k}{m-k}^2 (m-k)!.$
Consequently,
\begin{equation}
\label{eq:packing}
|\calM_\delta|
\ge \frac{\binom{n}{m}^2m!}{\binom{m}{k}\binom{n-k}{m-k}^2(m-k)!}
= \pth{\frac{\binom{n}{k}}{\binom{m}{k}}}^2k!
> \pth{\frac{n^2k}{e^3m^2}}^k
\ge \pth{\frac{\delta n}{e^3}}^k,
\end{equation}
where we use the inequalities that $(\frac{n}{k})^k \le \binom{n}{k}<(\frac{e n}{k})^k$ and $k!\ge (k/e)^k$.}

\hd{In the \ER model, for any $\pi$ with domain $S$ and range $T$ such that $|S| = |T| = m$,
we arbitrarily pick a bijection $\sigma:V(G_1)\mapsto V(G_2)$ such that $\sigma|_{S} = \pi$. 
Then, the conditional distribution $\maP_{G_1,G_2|\pi}$ can be factorized into
\[
\maP_{G_1,G_2|\pi} 
=\prod_{e\in \binom{S}{2}} P(e,\pi(e)) \prod_{e\in \binom{V(G_1)}{2}\backslash \binom{S}{2}} Q(e,\sigma(e)),
\]
where $P= \multibern(p,p,\rho)$ and $Q= \multibern(p,p,0)$.
Pick $\maQ$ to be an auxiliary null model under which $G_1$ and $G_2$ are independent with the same marginal as $\maP$. 
Then, $\maQ_{G_1,G_2}$ can be factorized into
$$
\maQ_{G_1,G_2} = \prod_{e\in \binom{S}{2}} Q(e,\pi(e)) \prod_{e\in \binom{V(G_1)}{2}\backslash \binom{S}{2}} Q(e,\sigma(e)).
$$
The KL-divergence between the product measures $\maP_{G_1,G_2|\pi}$ and $\maQ_{G_1,G_2}$ can be expressed as \begin{align*}
    D\pth{\maP_{G_1,G_2|\pi}\| \maQ_{G_1,G_2}} 
    = \binom{m}{2} D(P\Vert Q)
\end{align*}
for any $\pi:S\mapsto T$ with $|S| = |T| =m$. By Lemma~\ref{lmm:KL-Bern} and~\eqref{eq:MI-radius},
we obtain 
\begin{equation}\label{eq:MI-ub}
I(\pi^*;G_1,G_2)\le \max_{\pi} D(\maP_{G_1,G_2|\pi}\Vert \maQ_{G_1,G_2})
\le \binom{m}{2}D(P\Vert Q)\le 25\binom{m}{2}p^2 \phi(\gamma).
\end{equation}}

\hd{In the Gaussian Wigner model, let $P= \maN \pth{\begin{pmatrix} 0 \\ 0 \end{pmatrix}, \begin{pmatrix} 1 & \rho \\ \rho & 1 \end{pmatrix}}$ and $Q= \maN \pth{\begin{pmatrix} 0 \\ 0 \end{pmatrix}, \begin{pmatrix} 1 & 0 \\ 0 & 1 \end{pmatrix}}$.
We can similarly pick the auxiliary null model $\maQ$ under which $G_1$ and $G_2$ are independent. Then, $\maQ$ can be factorized into $$
\maQ_{G_1,G_2} = \prod_{e\in \binom{S}{2}} Q(e,\pi(e)) \prod_{e\in \binom{V(G_1)}{2}\backslash \binom{S}{2}} Q(e,\sigma(e)).
$$
By Lemma~\ref{lmm:KL-Bern} and~\eqref{eq:MI-radius},
we obtain 
\begin{equation}\label{eq:MI-ub-Gaussian}
I(\pi^*;G_1,G_2)\le \max_{\pi} D(\maP_{G_1,G_2|\pi}\Vert \maQ_{G_1,G_2})
\le \binom{m}{2}D(P\Vert Q)\le \frac{1}{2}\binom{m}{2}\log\pth{\frac{1}{1-\rho^2}}.
\end{equation}}

\section{Auxiliary results}

\begin{lemma}\label{lem: general case for eta}
    Recall that $\phi(\gamma) = (1+\gamma)\log(1+\gamma)-\gamma$ and $\eta,\gamma>0$. If $\eta \le \frac{\gamma}{4(1+\gamma)}$, then $(1+\gamma)(1-\eta)>1$
    and \begin{align}\label{eq: bound for Hgamma}
        \phi\qth{(1-\eta)(1+\gamma)-1}\ge \frac{1}{4} \phi(\gamma).
    \end{align}
\end{lemma}
\begin{proof}
We note that $(1+\gamma)(1-\eta)\ge 1+\gamma-\frac{\gamma}{4}>1$, and
\begin{align*}
    \phi\qth{(1-\eta)(1+\gamma)-1}&= (1+\gamma)(1-\eta)\log\qth{(1+\gamma)(1-\eta)}-\qth{(1+\gamma)(1-\eta)-1}\\
    &= (1-\eta)\qth{(1+\gamma)\log (1+\gamma)-\gamma}+(1+\gamma)(1-\eta)\log(1-\eta)+\eta\\
    &\ge  (1-\eta)\qth{(1+\gamma)\log (1+\gamma)-\gamma}+(1+\gamma)(-\eta)+\eta\\
    &=  (1-\eta)\qth{(1+\gamma)\log (1+\gamma)-\gamma}-\eta\gamma,
\end{align*}
where the last inequality is due to the fact that $(1-x)\log(1-x)+x\ge 0$ for any $0<x<1$ and $0<\eta \le \frac{\gamma}{4(1+\gamma)}<\frac{1}{4}$.
Since $\eta\le \frac{\gamma}{4(1+\gamma)}$ and $(1+\gamma)\log(1+\gamma)- \gamma\ge \frac{\gamma^2}{2(1+\gamma)}$, we obtain $\eta\gamma\le \frac{\gamma^2}{4(1+\gamma)}\le \frac{1}{2}\qth{(1+\gamma)\log(1+\gamma)-\gamma}$. Therefore,
\begin{align*}
    \phi\qth{(1-\eta)(1+\gamma)-1}&\ge (1-\eta)\qth{(1+\gamma)\log (1+\gamma)-\gamma}-\eta\gamma\\
    &\ge \pth{\frac{1}{2}-\eta} \phi(\gamma)\ge \frac{1}{4}\phi(\gamma),
\end{align*}
where the last inequality is because $0<\eta \le \frac{\gamma}{4(1+\gamma)}<\frac{1}{4}$.
\end{proof}

\begin{lemma}\label{lem: sum over delta m}
For any $m\ge 10$,
\begin{align*}
    \sum_{k=1}^{m-1} \exp\qth{-mh\pth{\frac{k}{m}}} \le \frac{4\log m+2}{m},
\end{align*}
where $h(x) = -x\log x-(1-x)\log(1-x)$ is the binary entropy function.
\end{lemma}
\begin{proof}
    We note that \begin{align*}
    \sum_{k=1 }^{m-1}  \exp\qth{-mh\pth{\frac{k}{m}}} 
    &\overset{\mathrm{(a)}}{\le} 2\sum_{1\le k\le \frac{m}{2}} \exp\qth{-k\log\pth{\frac{m}{k}}}\\
    &\overset{\mathrm{(b)}}{\le} 2\sum_{1\le k\le 2\log m} \exp\qth{-k\log\pth{\frac{m}{k}}}+2\sum_{2\log m+1\le k\le \frac{m}{2}} 2^{-k}\\
    &\overset{\mathrm{(c)}}{\le} 2\cdot \exp\pth{-\log m}\cdot(2 \log m)+2\cdot 2^{-2\log m}\\&\overset{\mathrm{(d)}}{\le} \frac{4\log m+2}{m},
\end{align*}
where $\mathrm{(a)}$ is because $h(x) = h(1-x)$ and $h(x)\ge -x\log x$; $\mathrm{(b)}$ is because $\log\pth{\frac{m}{k}}\ge \log 2$ when $k\le \frac{m}{2}$; $\mathrm{(c)}$ is because $k\log\pth{\frac{m}{k}}\ge \log m$ for $1\le k\le 2\log m$ when $m\ge 10$; $\mathrm{(d)}$ is because $2\cdot 2^{-2\log m}\le \frac{2}{m}$.
\end{proof}

\begin{lemma}\label{lmm:KL-Bern}
For $P_1=\multibern(p,p,\rho)$ and $Q_1= \multibern(p,p,0)$, we have  
\begin{align*}
    D(P_1\Vert Q_1)\le 25 p^2\phi(\gamma).
\end{align*}
For $P_2= \gaussianrho$ and $Q_2= \mathcal{N}\pth{\begin{bmatrix} 0 \\ 0 \end{bmatrix}, \begin{pmatrix} 1 & 0\\ 0 & 1 \end{pmatrix}}$, we have
\begin{align*}
    D(P_2\Vert Q_2) = \frac{1}{2}\log\pth{\frac{1}{1-\rho^2}}.
\end{align*}
\end{lemma}
\begin{proof}
Recall $p_{ab}$ for any $a,b\in \{0,1\}$ defined in~\eqref{eq:def_of_pab}. Then, we have
    \begin{align*}
    D(P_1\Vert Q_1) =&~ \sum_{\{a,b\}\in \{0,1\}} p_{ab}\log\qth{\frac{p_{ab}}{p^{a+b}(1-p)^{2-a-b}}}\\ =&~ \qth{p^2+\rho p(1-p)}\log\qth{1+\frac{\rho (1-p)}{p}}+2p(1-p)(1-\rho)\log(1-\rho)\\&~+\qth{(1-p)^2+\rho p(1-p)}\log\pth{1+\frac{\rho p}{1-p}}\\
    \overset{\mathrm{(a)}}{\le}&~ \qth{p^2+\rho p(1-p)}\log\qth{1+\frac{\rho (1-p)}{p}}+2p(1-p)(1-\rho)\cdot (-\rho)\\&~+\qth{(1-p)^2+\rho p(1-p)}\cdot \frac{\rho p}{1-p}\\
    \overset{\mathrm{(b)}}{=}&~p^2\qth{(1+\gamma)\log(1+\gamma)-\gamma}+\rho^2\qth{2p(1-p)+p^2},
\end{align*}
where $\mathrm{(a)}$ uses $\log(1+x)\le x$ for any $x>-1$; $\mathrm{(b)}$ follows from $\gamma = \frac{\rho(1-p)}{p}$.
Since $\log(1+x)\ge \frac{x}{x+1}+\frac{x^2}{2(x+1)^2}$ for any $x\ge 0$, we obtain that $p^2\qth{(1+\gamma)\log(1+\gamma)-\gamma}\ge \frac{p^2\gamma^2}{2(\gamma+1)}$.
When $\gamma<3$, since $\frac{p^2\gamma^2}{2(\gamma+1)}\ge \frac{p^2\gamma^2}{8} = \frac{\rho^2\cdot 3(1-p)^2}{24}\ge \frac{\rho^2\qth{2p(1-p)+p^2}}{24}$ for $0<p\le \frac{1}{2}$, we obtain that
\begin{align*}
    D(P_1\Vert Q_1)\le p^2\qth{(1+\gamma)\log(1+\gamma)-\gamma}+\rho^2\qth{2p(1-p)+p^2}\le 25p^2\qth{(1+\gamma)\log(1+\gamma)-\gamma}.
\end{align*}
When $\gamma\ge 3$, since  $\rho^2\qth{2p(1-p)+p^2}\le 3\rho p(1-p)$ and $(\log 4-1)\rho p(1-p)=p^2\gamma(\log 4-1)\le p^2\qth{(1+\gamma)\log(1+\gamma)-\gamma} $, we obtain
\begin{align*}
   D(P_1\Vert Q_1)& \le p^2\qth{(1+\gamma)\log(1+\gamma)-\gamma}+\rho^2\qth{2p(1-p)+p^2}\\&\le \pth{\frac{3}{\log 4-1}+1}p^2\qth{(1+\gamma)\log(1+\gamma)-\gamma}\\
   &\le  25p^2\qth{(1+\gamma)\log(1+\gamma)-\gamma}.
\end{align*}
Therefore, we get $D(P_1\Vert Q_1)\le 25 p^2 \phi(\gamma)$.

We denote $P_2(a,b)$ and $Q_2(a,b)$ the probability density function under $P_2$ and $Q_2$, respectively. Then, the KL-divergence between $P_2$ and $Q_2$ is given by \begin{align*}
    D(P_2\Vert Q_2) &= \iint P_2(a,b) \log\pth{\frac{P_2(a,b)}{Q_2(a,b)}}\,\mathrm{d}a\mathrm{d}b\\
    & = \iint P_2(a,b)\qth{\frac{1}{2}\log\pth{\frac{1}{1-\rho^2}}+\frac{\rho ab}{1-\rho^2}-\frac{\rho^2(a^2+b^2)}{2(1-\rho^2)}}
    \,\mathrm{d}a\mathrm{d}b\\ &=\frac{1}{2}\log\pth{\frac{1}{1-\rho^2}}+\frac{\rho^2}{1-\rho^2}-\frac{2\rho^2}{2(1-\rho^2)} = \frac{1}{2}\log\pth{\frac{1}{1-\rho^2}}.\qedhere
\end{align*}
\end{proof}

\begin{lemma}[Chernoff's inequality for Binomials]\label{lem:Chernoff_bound}
    Suppose $\xi\sim \Bin(n,p)$, denote $\mu = np$, for any $\delta>0$,
    \begin{align}
        \prob{\xi\ge (1+\delta)\mu}&\le \exp\sth{-\mu\qth{(1+\delta)\log(1+\delta)-\delta}}\label{eq:chernoff_bound_right_log},\\\prob{\xi\ge (1+\delta)\mu}&\le \exp\pth{-\frac{\delta\mu}{2+\delta}}\label{eq:chernoff_bound_right_nolog}.
    \end{align}
For any $0<\delta<1$, we have   \begin{equation}\label{eq:chernoff_bound_left}
        \prob{\xi\le (1-\delta)\mu}\le \exp\pth{-\frac{\delta^2\mu}{2}}.
    \end{equation}
\end{lemma}

\begin{proof}
    By Theorems 4.4 and 4.5 in \cite{mitzenmacher_probability_2005} we have \eqref{eq:chernoff_bound_right_log} and \eqref{eq:chernoff_bound_left}. Since $(1+\delta)\log(1+\delta)-\delta\ge \frac{\delta^2}{2+\delta}$, we obtain \eqref{eq:chernoff_bound_right_nolog} from \eqref{eq:chernoff_bound_right_log}.
\end{proof}

\begin{lemma}[Hanson-Wright inequality]\label{lem:Hanson-Wright}
    Let $X,Y\in \mathbb{R}^n$ be standard Gaussian vectors such that the pairs $(X_i,Y_i)\sim \gaussianrho$ are independent for $i=1,\cdots,n$. Let $M_0\in \mathbb{R}^{n\times n}$ be any deterministic matrix. There exists some universal constant $c_0>0$ such that, \begin{align*}
        \prob{\left| X^\top M_0 Y -\rho \mathrm{Tr}(M_0)\right|\ge c_0\pth{\Vert M_0\Vert_F \sqrt{\log(1/\delta)}\vee\Vert M_0\Vert_2\log(1/\delta)}}\le \delta.
    \end{align*} 
\end{lemma}

\begin{proof}
    Note that $X^{\top} M_0Y = \frac{1}{4}(X+Y)^\top M_0(X+Y)-\frac{1}{4}(X-Y)^\top M_0(X-Y)$ and \begin{align*}
        \expect{(X+Y)^\top M_0(X+Y)} = (2+2\rho)\mathrm{Tr}(M_0),\expect{(X-Y)^\top M_0(X-Y)} = (2-2\rho)\mathrm{Tr}(M_0).
    \end{align*}
    By Hanson-Wright inequality \cite{hanson1971bound}, there exists some universal constant $c_0$ such that \begin{align*}
        &\prob{\left| \frac{1}{4}(X+Y)^\top M_0(X+Y) - \frac{2+2\rho}{4}\mathrm{Tr}(M_0)\right| \ge \frac{c_0}{2}\pth{\Vert M_0\Vert_F \sqrt{\log(1/\delta)}\vee\Vert M_0\Vert_2\log(1/\delta)}}\le \frac{\delta}{2},\\&\prob{\left| \frac{1}{4}(X-Y)^\top M_0(X-Y) - \frac{2-2\rho}{4}\mathrm{Tr}(M_0)\right| \ge \frac{c_0}{2}\pth{\Vert M_0\Vert_F \sqrt{\log(1/\delta)}\vee\Vert M_0\Vert_2\log(1/\delta)}}\le \frac{\delta}{2}
    \end{align*} for any $\delta>0$.
    Consequently, \begin{align*}
        &~\prob{\left| X^\top M_0 Y -\rho \mathrm{Tr}(M_0)\right|\ge c_0\pth{\Vert M_0\Vert_F \sqrt{\log(1/\delta)}\vee\Vert M_0\Vert_2\log(1/\delta)}}\\ \le &~\prob{\left| \frac{1}{4}(X+Y)^\top M_0(X+Y) - \frac{2+2\rho}{4}\mathrm{Tr}(M_0)\right| \ge \frac{c_0}{2}\pth{\Vert M_0\Vert_F \sqrt{\log(1/\delta)}\vee\Vert M_0\Vert_2\log(1/\delta)}}\\&+ \prob{\left| \frac{1}{4}(X-Y)^\top M_0(X-Y) - \frac{2-2\rho}{4}\mathrm{Tr}(M_0)\right| \ge \frac{c_0}{2}\pth{\Vert M_0\Vert_F \sqrt{\log(1/\delta)}\vee\Vert M_0\Vert_2\log(1/\delta)}}\\
        \le&~ \delta.\qedhere
    \end{align*}
\end{proof}

\begin{lemma}[Chernoff's inequality for Chi-squared distribution]\label{lem:chisquare}
    Suppose $\xi$ follows the chi-squared distribution with $n$ degrees of freedom. Then, for any  $\delta>0$, \begin{align}\label{eq:concentration_for_chisquare}
        \prob{\xi>(1+\delta) n}\le  \exp\pth{-\frac{n}{2}\pth{\delta-\log\pth{1+\delta}}}.
    \end{align}
\end{lemma}
\begin{proof}
    The result follows from \cite[Theorem 1]{ghosh2021exponential}.
\end{proof}
\bibliographystyle{alpha}
\bibliography{main}

\end{document}